\documentclass[conference]{IEEEtran}

\def\BibTeX{{\rm B\kern-.05em{\sc i\kern-.025em b}\kern-.08emT\kern-.1667em\lower.7ex\hbox{E}\kern-.125emX}}

\newif\iftr     % Full technical report
\newif\ifall    % Various stuff that might be useful but for now we don't want to use it
\newif\ifconf   % Submission to a conf or journal, with space contraints
\newif\ifsq     % Squeeze space?
\newif\ifnonb   % Non blind submission
\newif\iftodos
\newif\ifrev

\revfalse

\newif\ifsqCAP
\newif\ifsqVS
\newif\ifsqEN
\newif\ifsqTIT

\sqfalse
\sqCAPfalse
\sqENfalse
\sqVSfalse
\sqTITfalse

\newcommand{\tr}[1]{\iftr #1 \fi}
\newcommand{\all}[1]{\ifall #1 \fi}
\newcommand{\cnf}[1]{\ifconf #1 \fi}

\usepackage{etex}
\usepackage{balance}
\usepackage{epstopdf}
\usepackage{placeins}

\usepackage{cite}

\usepackage{graphicx}
\usepackage{float}
\usepackage{dblfloatfix}
\usepackage{multirow}
\usepackage{rotating}
\usepackage[switch]{lineno} % Line numbers
\usepackage{makecell}
\usepackage{tabulary}
\usepackage{parcolumns}
\usepackage{tikz}
\usetikzlibrary{tikzmark}

\usepackage{xpatch}
\expandafter\xpatchcmd
\csname pgfk@/tikz/every picture/.@cmd\endcsname
{\thepage}{\arabic{page}}{}{}

\tikzstyle{comment} = [draw, fill=blue!70, text=white, text width=3cm, minimum height=1cm, rounded corners, align=left, font=\scriptsize]
\tikzstyle{background_alg} = [draw, fill=blue!20, opacity=0.4, inner sep=4pt, rounded corners=2pt]

\usetikzlibrary{shapes}
\usetikzlibrary{plotmarks}
\usetikzlibrary{calc, fit}

\usepackage{enumitem}

\usepackage{amsthm}
\usepackage{amsmath,amssymb,amsfonts}
\usepackage{mathtools,mathrsfs}

\newtheorem{thm}{Theorem}
\newtheorem{clm}{Claim}
\newtheorem{crl}{Corollary}
\newtheorem{lma}{Lemma}

\DeclarePairedDelimiter{\ceil}{\lceil}{\rceil}

\usepackage{soul}
\usepackage{fontawesome}
\usepackage{pifont}
\usepackage{textcomp}
\usepackage{booktabs}
\usepackage{url}
\usepackage{pbox}
\usepackage[normalem]{ulem}
\usepackage[10pt]{moresize}

\usepackage[hidelinks]{hyperref}
\usepackage{cleveref}
\usepackage[utf8]{inputenc}
\crefname{section}{§}{§§}
\Crefname{section}{§}{§§}

\newcommand{\macb}[1]{\textbf{{#1}}}
\newcommand{\macbs}[1]{\textbf{{#1}}}

\newcommand{\noAnswer}{\textcolor{black}{\faQuestionCircle}}

\ifsqCAP
\usepackage[font={scriptsize}]{caption}
\usepackage[font={scriptsize}]{subcaption}
\else
\usepackage[font={footnotesize}]{caption}
\usepackage[font={footnotesize}]{subcaption}
\fi

\newcommand{\vspaceSQ}[1]{\ifsqVS\vspace{#1}\fi}
\newcommand{\enlargeSQ}[1]{\ifsqEN\enlargethispage{\baselineskip}\fi}

\ifsqTIT
\newcommand{\subparagraph}{}
\usepackage[compact]{titlesec}
\titlespacing*{\section}{0pt}{6pt}{2pt}
\titlespacing*{\subsection}{0pt}{5pt}{1pt}
\titlespacing*{\subsubsection}{0pt}{5pt}{1pt}
\fi

\usepackage{xcolor}
\definecolor{darkgrey}{RGB}{70,70,70}
\definecolor{lightgrey}{RGB}{200,200,200}
\definecolor{lyellow}{RGB}{255,255,100}
\definecolor{llyellow}{RGB}{250,250,180}
\definecolor{lgreen}{RGB}{144,238,144}

\usepackage[customcolors]{hf-tikz}
\hfsetbordercolor{white}
\hfsetfillcolor{vlgray}

\definecolor{vlgray}{rgb}{0.77 0.77 0.77}
\definecolor{ablack}{rgb}{0.2 0.2 0.2}
\definecolor{vllgray}{rgb}{0.9 0.9 0.9}
\definecolor{bblue}{rgb}{0.7 0.7 0.99}

\usepackage{colortbl}

\usepackage{inconsolata}
\usepackage{listings}

\ifsq
\else
\fi

\newcommand{\maciej}[1]{\textcolor{blue}{[Maciej: #1]}}
\newcommand{\m}[1]{\textcolor{blue}{[Maciej: #1]}}

\newcommand{\armon}[1]{\textcolor{blue}{[Armon: #1]}}

\newcounter{highlight}

\newcounter{Ahighlight}

\renewcommand{\maciej}[1]{}

\usepackage{scalerel,stackengine}
\stackMath
\newcommand\rwh[1]{%
\savestack{\tmpbox}{\stretchto{%
  \scaleto{%
        \scalerel*[\widthof{\ensuremath{#1}}]{\kern-.6pt\bigwedge\kern-.6pt}%
                  {\rule[-\textheight/2]{1ex}{\textheight}}%WIDTH-LIMITED BIG WEDGE
                              }{\textheight}% 
}{0.5ex}}%
\stackon[1pt]{#1}{\tmpbox}%
}

\usepackage[hang,flushmargin]{footmisc}

\DeclarePairedDelimiter\abs{\lvert}{\rvert}

\renewcommand{\epsilon}{\ensuremath\varepsilon}

\renewcommand{\phi}{\ensuremath{\varphi}}

\if 0

\usepackage[linesnumbered,ruled]{algorithm2e}
\usepackage{multicol}
\SetKwComment{Comm}{$\triangleright$\ }{}
\SetAlFnt{\scriptsize}
\SetAlCapFnt{\scriptsize}
\SetAlCapNameFnt{\scriptsize}
\SetKwInOut{Input}{Input}
\SetKwInOut{Output}{Output}

\makeatletter
\NewDocumentCommand{\LeftComment}{s m}{%
\Statex \IfBooleanF{#1}{\hspace*{\ALG@thistlm}}\(\triangleright\) #2}
\makeatother

\fi

\allfalse
\conffalse
\trtrue
\nonbtrue

\newif\ifbsp
\bspfalse

\begin{document}

%\linenumbers

\ifsq
\title{\mbox{Revision: High-Performance Parallel Graph Coloring}\\\mbox{with Strong Guarantees on Work, Depth, and Quality}\vspace{-1.5em}}
\else
\title{High-Performance Parallel Graph Coloring with\\ Strong Guarantees on Work, Depth, and Quality}
\fi

\ifnonb
\author{\IEEEauthorblockN{Maciej Besta$^1$, Armon Carigiet$^1$, Zur Vonarburg-Shmaria$^{1}$, Kacper Janda$^{2}$, Lukas Gianinazzi$^{1}$, Torsten Hoefler$^1$}
\IEEEauthorblockA{
\textit{$^1$Department of Computer Science, ETH Zurich}\\
\textit{$^2$Faculty of Computer Science, Electronics and Telecommunications, AGH-UST 
%
%\vspace{-0.25em}
%
} \\
}
}
\fi

\maketitle

\thispagestyle{plain}
\pagestyle{plain}

\begin{abstract}
We develop the first parallel graph coloring heuristics with strong theoretical
guarantees on \emph{work} and \emph{depth} and \emph{coloring quality}. The key
idea is to design a \emph{relaxation} of the \emph{vertex degeneracy order}, a
well-known graph theory concept, and to color vertices in the order dictated by
this relaxation.  This introduces a tunable amount of parallelism into the
degeneracy ordering that is otherwise hard to parallelize. This simple idea
enables significant benefits in several key aspects of graph coloring.  For
example, one of our algorithms ensures polylogarithmic depth and a bound on the
number of used colors that is superior to all other parallelizable schemes,
while maintaining work-efficiency.
\ifall
This relaxation decomposes the graph into few low-degree vertex partitions,
while introducing a tunable amount of parallelism into an ordering that is
otherwise inherently sequential in various scenarios.
\fi
In addition to provable guarantees, the developed algorithms have competitive
run-times for several real-world graphs, while almost always providing superior
coloring quality. Our degeneracy ordering relaxation is of separate interest
for algorithms outside the context of coloring.
\end{abstract}

\vspace{0.5em}
\ifnonb
{\noindent\small\bf This is the full version of a paper published at\\
 ACM/IEEE Supercomputing'20 under the same title}
\else
\ifconf
{\small\noindent\macb{[Anonymized] Code and report:}\\\url{https://www.dropbox.com/s/fd7r1z031x5b6hk/pgc.zip?dl=0}}
\fi
\iftr
{\small\noindent\macb{[Anonymized] Code:}\\\url{https://www.dropbox.com/s/fd7r1z031x5b6hk/pgc.zip?dl=0}}
\fi
\fi

\section{Introduction}
\label{sec:intro}

\ifsq\enlargethispage{\baselineskip}\fi

Graph coloring, more specifically \emph{vertex coloring}, is a well studied
problem in computer science, with many practical applications in domains such
as sparse linear algebra computations~\cite{coleman1983estimation,
jones1994scalable, gebremedhin2005what, besta2020communication, slimsell,
kwasniewski2019red, solomonik2017scaling}, conflicting task
scheduling~\cite{kaler2016executing, arkin1987scheduling, marx2004graph,
ramaswami1989distributed}, networking and routing~\cite{ghrab2013coloring,
li2000partition, besta2018slim, besta2014slim, di2019network,
besta2019fatpaths, besta2020highperformance, javedankherad2020content,
dey2013fuzzy}, register allocation~\cite{chaitin1982register,
de2018transformations}, and many others~\cite{lewis2015guide}. A \emph{vertex
coloring} of a graph~$G$ is an assignment of colors to vertices, such that no
two neighboring vertices share the same color. A $k$-coloring is a vertex
coloring of $G$ which uses $k$ distinct colors. The minimal amount of colors
$k$ for which a $k$-coloring can be found for $G$ is referred to as the
\emph{chromatic number}~$\chi(G)$. An \emph{optimal coloring}, also sometimes
referred to as \emph{the coloring problem} or a $\chi$-coloring, is the problem
of coloring $G$ with $\chi(G)$ colors. Finding such and optimal
coloring was shown to be NP-complete~\cite{garey1997some}.

Nonetheless, colorings with a reasonably low number of colors can in practice
be computed quite efficiently in the sequential setting using
\emph{heuristics}. One of the most important is the \emph{Greedy}
heuristic~\cite{welsh1967an}, which sequentially colors vertices by choosing,
for each selected vertex~$v$, the smallest color not already taken by $v$'s
  neighbors. This gives a \emph{guarantee} for a coloring of $G$ with at most
  $\Delta + 1$ colors, where $\Delta$ is the maximum degree in~$G$. To further
  improve the coloring quality (i.e., \#colors used), Greedy is in practice
  often used with a certain \emph{vertex ordering heuristic}, which decides the
  \emph{order in which Greedy colors the vertices}.
Example heuristics are \emph{first-fit} (FF)~\cite{welsh1967an} which uses
the default order of the vertices in $G$, \emph{largest-degree-first}
(LF)~\cite{welsh1967an} which orders vertices according to their degrees,
\emph{random} (R)~\cite{jones1993parallel} which chooses vertices 
uniformly at random, \emph{incidence-degree}
(ID)~\cite{coleman1983estimation} which picks vertices with the largest
number of uncolored neighbors first, \emph{saturation-degree}
(SD)~\cite{brelaz1979new}, where a vertex whose neighbors use the largest
number of distinct colors is chosen first, and \emph{smallest-degree-last}
(SL)~\cite{matula1983smallest} that removes lowest degree vertices, recursively
colors the resulting graph, and then colors the removed vertices.
All these ordering heuristics, combined with Greedy, have the inherent problem
of \emph{no parallelism}.

\ifsq\enlargethispage{\baselineskip}\fi

Jones and Plassmann combined this line of work with earlier parallel schemes
for deriving maximum independent sets~\cite{karp1985afast, luby1986simple} and
  obtained a \emph{parallel} graph coloring algorithm (JP) that colors a vertex~$v$
  once all of $v$'s neighbors that come later in the provided ordering have
  been colored. They showed that JP, combined with a \emph{random vertex
  ordering} (JP-R), runs in expected depth $O(\log{n}/\log\log{n})$ and
  $O(n+m)$ work for constant-degree graphs ($n$ and $m$ are \#vertices and
  \#edges in~$G$, respectively). Recently, Hasenplaugh et
  al.~\cite{whasenplaugh2014ordering} extended JP with the
  \emph{largest-log-degree first} (LLF) and \emph{smallest-log-degree-last}
  (SLL) orderings with better bounds on depth; these orderings approximate
  the LF and SL orderings, respectively.
%
% to $O(\log{n} + \log{\Delta}( \min (\sqrt{m}, \Delta +
% \log^2{\Delta}\log{n}/\log{\log n})))$ for graphs of \emph{arbitrary
% degree}~\cite{whasenplaugh2014ordering}.
%
There is also another (earlier) work~\mbox{\cite{patwary2011new}} that --
similarly to JP-SLL -- approximates SL with the \emph{``ASL'' ordering}. 
The resulting coloring combines JP with ASL, we denote it as JP-ASL~\mbox{\cite{patwary2011new}}.
However, it offers \emph{no bounds} for work or depth.

\ifrev\marginpar{\vspace{-4em}\colorbox{orange}{\textbf{R-3}}}\fi

Overall, \emph{there is no parallel algorithm with strong
theoretical guarantees on work and depth and quality}.
Whilst having a reasonable theoretical run-time, JP-R may offer
colorings of poor quality~\cite{whasenplaugh2014ordering, jrallwright1995}. On
the other hand, JP-LF and JP-SL, which provide a better coloring quality, run
in $\Omega(n)$ or $\Omega(\Delta^2)$ for some
graphs~\cite{whasenplaugh2014ordering}. This was addressed by the recent JP-LLF
and JP-SLL algorithms~\cite{whasenplaugh2014ordering} that produce colorings of
similarly good quality to their counterparts JP-LF and JP-SL, and run in an
expected depth that is within a logarithmic factor of JP-R. 
However, \emph{no guaranteed upper bounds on the coloring quality (\#colors),
better than the trivial $\Delta + 1$ bound from Greedy, exist for JP-LLF,
JP-SLL, or JP-ASL}. 

\ifrev\marginpar{\vspace{-1em}\colorbox{orange}{\textbf{R-3}}}\fi

To alleviate these issues, we present \emph{\textbf{the first} graph coloring
algorithms with provably good bounds on \textbf{work} \ul{and} \textbf{depth}
\ul{and} \textbf{quality}, simultaneously ensuring high performance and
competitive quality in practice}. The key idea is to use a novel vertex
ordering, the \emph{{provably} approximate degeneracy ordering} (\textbf{ADG},
contribution~\textbf{\#1}) when selecting which vertex is the next to be
colored. 
The exact degeneracy ordering is -- intuitively -- an ordering obtained by
iteratively removing vertices of smallest degrees. Using the degeneracy
ordering with JP leads to the best possible Greedy coloring
quality~\cite{matula1983smallest}.
Still, computing the exact degeneracy ordering is hard to parallelize: for some
graphs, it leads to $\Omega(n)$ coloring
run-time~\cite{whasenplaugh2014ordering}.
\ifall
\maciej{remove?}
As a matter of fact, using the \emph{exact} degeneracy ordering -- without
explicitly naming it -- was used in the JP-SL coloring
algorithm~\cite{matula1983smallest}. The maximal \#colors used in JP-SL is
(provably) $d + 1$, where $d$ is the \emph{degeneracy
of~$G$}~\cite{matula1983smallest} (the degeneracy of a graph~$G$ is --
intuitively -- an upper bound on the minimal degree of every induced subgraph
of $G$ (details are in~\cref{sec:back_deg}).
\fi
To tackle this, we (provably) relax the strict degeneracy order by assigning
the same rank (in the ADG ordering) to a \emph{batch} of vertices that --
intuitively -- have \emph{similarly small degrees}.  This approach also results
in \emph{provably higher parallelization} because each batch of vertices can be
processed in parallel.

\ifall\m{bp}
To tackle this, we sacrifice a (provably) small amount of \emph{accuracy} in
the degeneracy ordering (i.e., how well one preserves the exact degeneracy
order). We achieve this by selecting \emph{batches} of vertices to be assigned
the same rank in the generated ordering. Two vertices are assigned the same
rank if their degrees \emph{are similar enough}.  Intuitively, while relaxing
the strict degeneracy ordering, this also results in \emph{provably more
parallelization} because each batch of vertices can be processed in parallel.
\fi

This simple idea, when applied to graph coloring, gives surprisingly rich
outcome.  We use it to develop three novel graph coloring algorithms that
enhance two relevant lines of research.
We first combine ADG with JP, obtaining \textbf{JP-ADG}
(contribution~\textbf{\#2}), a coloring algorithm that is
\emph{parallelizable}: vertices with the same ADG rank are colored in parallel.
It has the expected worst-case depth of $O(\log^2{n} +\log{\Delta}( d\log{n} +
\log{d}\log^2{n}/\log{\log n}))$. Here, $d$ is the \emph{degeneracy of a
graph~$G$}: an upper bound on the minimal degree of every induced subgraph of
$G$ (detailed in~\cref{sec:back_deg})~\cite{matula1983smallest}.
\tr{JP-ADG is also \emph{work-efficient} ($O(n+m)$ work) and has \emph{good
coloring quality}: it uses at most $2(1+\epsilon)d + 1$ colors\footnote{The
provided bound on the number of colors in JP-ADG and in DEC-ADG-ITR is exactly
$\lceil 2(1+\epsilon)d \rceil + 1$, and $\lceil (2+\epsilon)d \rceil$ in
DEC-ADG. However, for clarity of notation, we will omit $\lceil \cdot \rceil$,
for both numbers of colors and for vertex neighbors, whenever it does not change
any conclusions or insights.}, for any $\epsilon > 0$.}
\cnf{JP-ADG is also \emph{work-efficient} ($O(n+m)$ work) and has \emph{good
coloring quality}: it uses at most $2(1+\epsilon)d + 1$ colors, for any
$\epsilon > 0$.}
Moreover, we also combine ADG with another important line of graph coloring
algorithms that are \emph{not} based on JP but instead use
\emph{speculation}~\cite{ccatalyurek2011distributed, bozdaug2008framework,
besta2017push, gebremedhin2000scalable, boman2005scalable,
gebremedhin2000graph, ccatalyurek2012graph, saule2012early,
sariyuce2012scalable, rokos2015fast, grosset2011evaluating, deveci2016parallel,
finocchi2005experimental, patwary2011new}. Here, vertices are colored independently
(``speculative coloring''). Potential \emph{coloring conflicts} (adjacent
vertices assigned the same colors) are resolved by repeating coloring attempts.
Combining ADG with this design gives \textbf{DEC-ADG}
(contribution~\textbf{\#3}), \emph{the first scheme based on speculative
coloring with provable strong guarantees on all key aspects of parallel graph
coloring}: {work} $O(n+m)$, {depth} $O(\log d \log^2 n)$ , and {quality}
$(2+\epsilon)d$. 
Finally, we combine key design ideas in DEC-ADG with an existing recent
algorithm~\cite{ccatalyurek2012graph} (referred to as ITR) also based on
speculative coloring. We derive an algorithm called \textbf{DEC-ADG-ITR} that
improves coloring quality of ITR \emph{both in theory and practice}.
%
% \emph{\textbf{improves run-times and coloring qualities of ITR to the point
% of being faster than most optimized recent JP based routines, including
% fastest LF and LLF}} (contribution~\textbf{\#4}).

\ifsq\enlargethispage{\baselineskip}\fi

\ifrev\marginpar{\vspace{2em}\colorbox{orange}{\textbf{ALL}}}\fi

We conduct the most extensive theoretical analysis of graph coloring algorithms
so far, considering 20 parallel graph coloring routines with provable
guarantees (contribution~\textbf{\#5}). 
All our algorithms offer substantially better bounds than past work. 
\iftr
Compared to the most recent JP-SLL and JP-LLF
colorings~\cite{whasenplaugh2014ordering}, JP-ADG gives a strong theoretic
coloring guarantee that is much better than $\Delta + 1$ because $d \ll \Delta$
for many classes of sparse graphs, such as scale-free
  networks~\cite{barabesi1999emergence} and planar
  graphs~\cite{lickwhite1970degenerate}. It also provides an interesting novel
  tradeoff in the depth. On one hand, it depends on $\log^2 n$ while JP-SLL and
  JP-LLF depend on $\log n$. However, while JP-SLL and JP-LLF depend linearly
  on $\Delta$ or $\sqrt{m}$ (which are usually large in today's graphs), the
  depth of JP-ADG depends linearly on degeneracy~$d$, and, as we also show, $d
  \ll \sqrt{m}$ and $d \ll \Delta$. 
\fi
We also perform a broad empirical evaluation, illustrating that our algorithms
(1) are competitive in run-times for several real-world graphs, while (2)
offering superior coloring quality.
\iftr
Only JP-SL and JP-SLL use comparably few colors, but they are at least
1.5$\times$ slower. Our routines are of interest for coloring both small and
large graphs, for example in online execution scheduling and offline data
analytics, respectively. 
\fi
We conclude that \emph{our algorithms offer the best coloring quality at the
smallest required runtime overhead}.

\ifall
\maciej{fix}
Formally, our $2(1+\epsilon)$SL ordering heuristic computes a
\emph{$2(1+\epsilon)$-approximate degeneracy ordering} (for $\epsilon>0$) and
consequently limits the amount of colors used by $2(1+\epsilon)d + 1$, using
$O(\log^2{n})$ depth and $O(n+m)$ work in a PRAM CREW setting. In combination
with JP, we get a coloring algorithm JP-$2(1+\epsilon)$SL, which performs
$O(n+m)$ operations, i.e., it has linear work and has an expected depth of
$O(\log^2{n} +\log{\Delta}( d\log{n} + \log{d}\log^2{n}/\log{\log n}))$ in a
CRCW setting. Thus is up to a logarithmic factor similar to JP-SLL's depth.
However, in terms of linear factors, the depth is an improvement or comparable
to other JP variants, where $\Delta$ and $\sqrt{m}$ dominate. This is the case,
because $d \leq \Delta$ and $\sqrt{m} \geq d / 2$ holds for \emph{all} graphs
$G$ with degeneracy~$d$. 
\fi

\ifbsp
We also provide a distributed-memory version of our graph coloring algorithm.
We support it with the communication-cost analysis based on the Bulk-Synchronous
Parallel (BSP) model~\cite{??}. We discuss its advantages over other distributed
graph coloring algorithms. \maciej{TODO: Finish once ready}
\fi

In a brief summary, we offer the following:

\begin{description}[noitemsep, leftmargin=1.5em]
\item$\bullet$ The first parallel algorithm for deriving the (approximate)
graph degeneracy ordering (ADG).
\item$\bullet$ The first parallel graph coloring algorithm (JP-ADG), in
a line of heuristics based on Jones and Plassman's scheme, with strong
bounds on work, depth, and coloring quality.
\item$\bullet$ The first parallel graph coloring algorithm (DEC-ADG),
in a line of heuristics based on speculative coloring, with
strong bounds on work, depth, and coloring quality
%
%, all while provably relying only on concurrent reads.
%
\item$\bullet$ A use case of how ADG can seamlessly enhance
an existing state-of-the-art graph coloring scheme (DEC-ADG-ITR).
\ifbsp
\item$\bullet$ \maciej{Add BSP when ready}
\fi
\item$\bullet$ The most extensive (so far) theoretical analysis of parallel
graph coloring algorithms, showing advantages of our algorithms over
state-of-the-art in several dimensions.
\item$\bullet$ Superior coloring quality offered by our
algorithms over tuned modern schemes for many real-world graphs.
\end{description}

\noindent
We note that degeneracy ordering is used beyond graph
coloring~\cite{cazals2008note, DBLP:conf/isaac/EppsteinLS10,
DBLP:journals/tcs/TomitaTT06, DBLP:conf/latin/Farach-ColtonT14}; thus, our
\emph{ADG scheme is of separate interest}.

\section{Fundamental Concepts}
\label{sec:back}

We start with background; Table~\ref{tab:symbols} lists 
key symbols. 
Vertex coloring was already described in Section~\ref{sec:intro}.

%\enlargethispage{\baselineskip}

\subsection{Graph Model and Representation}

We model a graph $G$ as a tuple $(V,E)$; $V$ is a set of vertices and $E
\subseteq V \times V$ is a set of edges; $|V|=n$ and $|E|=m$.  We focus on
graph coloring problems where edge directions are not relevant.  Thus, $G$ is
undirected. The maximum, minimum, and average degree of a given graph $G$ are
$\Delta$, $\delta$, and $\widehat{\delta}$, respectively. The neighbors and the
degree of a given vertex~$v$ are $N(v)$ and $deg(v)$, respectively. $G[U] = (U,
E[U])$ denotes an induced subgraph of~$G$: a graph where $U \subseteq V$ and
$E[U] = \{(v,u) \mid v \in U \land u \in U\}$, i.e., $E[U]$ contains edges with
both endpoints in $U$. $N_U(v)$ and $deg_U(v)$ are the neighborhood and the
degree of~$v \in V$ in $G[U]$. 
The vertices are identified by integer IDs that define a total order: $V=\{1,\ldots,n\}$. 
\iftr
These IDs define a total order $\succ$ on the vertices that is used to sort the
neighborhoods.
\fi
We store $G$ using CSR, the standard graph representation that consists of
$n$~sorted arrays with neighbors of each vertex ($2m$ words) and offsets to
each array ($n$ words).

\begin{table}[h!]
\centering
\setlength{\tabcolsep}{2pt}
\ifsq
\vspace{-1em}
\renewcommand{\arraystretch}{0.7}
\fi
\scriptsize
%\sf
%
\begin{tabulary}{\columnwidth}{ll}
\toprule
$G$ & A graph $G=(V,E)$; $V$ and $E$ are sets of vertices and edges.\\
$G[U]$ & $G[U]=(U,E[U])$ is a subgraph of~$G$ induced on $U \subseteq V$.\\
$n,m$&Numbers of vertices and edges in $G$; $|V| = n, |E| = m$.\\
$\Delta, \delta, \widehat{\delta}$&Maximum degree, minimum degree, and average degree of $G$.\\
% $a, d$ & The arboricity and the degeneracy of $G$. \\
$d$ & The degeneracy of $G$. \\
$deg(v), N(v)$ & The degree and the neighborhood of a vertex~$v \in V$.\\
$deg_U(v)$ & The degree of $v$ in a subgraph induced by the vertex set $U \subseteq V$.\\
$N_U(v)$ & The neighborhood of $v$ in a subgraph induced by $U \subseteq V$.\\
$\rho_X(v)$ & A priority function $V \to \mathbb{R}$ associated with vertex ordering~$X$.\\
$P$ & The number of processors (in a given PRAM machine). \\
\bottomrule
\end{tabulary} 
\vspace{-0.5em}
\caption{Selected symbols used in the paper. When we use a symbol
in the context of a specific loop iteration~$\ell$, we add $\ell$ in brackets or as subscript
(e.g., $\widehat{\delta}_\ell$ is $\widehat{\delta}$ in iteration~$\ell$).}
\label{tab:symbols}
\ifsq\vspace{-1.5em}\fi
\end{table}

\ifsq\vspace{-0.25em}\fi

\iftr
\subsection{Degeneracy, Coreness, and Related Concepts}
\label{sec:back_deg}
\fi

\ifconf
\subsection{Degeneracy and Related Concepts}
\label{sec:back_deg}
\fi

A graph $G$ is \emph{$s$-degenerate}~\cite{chrobak1991planar} if, in each of
its induced subgraphs, there is a vertex with a degree of at most $s$. The
\emph{degeneracy}~$d$ of $G$~\cite{erdHos1966chromatic, kirousis1996linkage,
freuder1982sufficient, bader2003automated} is the \emph{smallest} $s$, such
that $G$ is still $s$-degenerate. The \emph{degeneracy ordering} of
$G$~\cite{matula1983smallest} is an ordering, where each vertex~$v$ has at most
$d$ neighbors that are ordered higher than~$v$. Then, a \emph{$k$-approximate
degeneracy ordering} differs from the exact one in that $v$ has at most $k
\cdot d$ neighbors ranked higher in this order. A \emph{\ul{partial}
$k$-approximate degeneracy ordering} is a similar ordering, where multiple
vertices can be ranked equally, and we have that each vertex has at most $k
\cdot d$ neighbors with \emph{equal or higher} rank. A \emph{partial
$k$-approximate degeneracy ordering} can be trivially extended into a
\emph{$k$-approximate degeneracy ordering} by imposing an (arbitrary) order on
vertices ranked equally. Both \emph{degeneracy} and a \emph{degeneracy
ordering} of $G$ can be computed in linear time by sequentially removing
vertices of smallest degree~\cite{matula1983smallest}.

\iftr
Closely related notions are a \emph{$k$-core} of $G$~\cite{seidman1983network}:
A connected component that is left over after iteratively removing vertices
with degree less than $k$ from $G$. The \emph{coreness} of a vertex $v$ is
defined as the largest possible $k$, such that $v$ is part of a subgraph $S$ of
$G$ with minimum degree $k$.  
\fi

\subsection{Models for Algorithm Analysis}

\ifall
Here, we outline the necessary concepts related to modeling the cost of parallel
algorithms and briefly describe our approach.
We use the DAG model and the ideal parallel machine for modeling algorithm
executions~\cite{blumofe1998space}, and the work-depth
analysis~\cite{blelloch2010parallel} for performance analysis. We sometimes use
certain concepts from the PRAM model~\cite{blelloch1996programming}.
\fi

As a \textbf{compute model}, we use the \emph{DAG model of dynamic
multithreading}~\cite{blumofe1999scheduling, blumofe1998space}.  In this model,
a specific computation (resulting from running some parallel program) is
modeled as a \emph{directed acyclic graph} (DAG). Each \emph{node} in a DAG
corresponds to a constant time operation. \emph{In-edges} of a node model the
data used for the operation. As operations run in constant time, there are
$O(1)$ in-edges per node. The \emph{out-edges} of a node correspond to the
computed output. A node can be executed as soon as all predecessors finish
executing. 

\ifsq\enlargethispage{\baselineskip}\fi

Following related work~\cite{whasenplaugh2014ordering,
dhulipala2018theoretically}, we assume that a parallel computation (modeled as
a DAG) runs on the \emph{ideal parallel computer} (\textbf{machine model}).
Each instruction executes in unit time and there is support for concurrent
reads, writes, and read-modify-write atomics (any number of such instructions
finish in $O(1)$ time).
\iftr
These are standard assumptions used in all recent parallel graph coloring
algorithms~\cite{whasenplaugh2014ordering, dhulipala2018theoretically}. 
\fi
We develop algorithms based on these assumptions \emph{but we also
provide algorithms that use weaker assumptions} (algorithms that only rely on 
\emph{concurrent reads}).

We use the \textbf{work-depth (W--D) analysis} for bounding run-times of
parallel algorithms in the DAG model. The \emph{work} of an algorithm is the
total number of nodes and the \emph{depth} is defined as the longest directed
path in the DAG~\cite{Bilardi2011, blelloch2010parallel}. 

\ifall
\macbs{PRAM }
While we do not explicitly use PRAM, we outline it as some considered
algorithms are developed in this model.
The PRAM model extends the well-known sequential RAM model to support
parallelism. A $p$-processor PRAM machine~\cite{blelloch2010parallel,
JaJa2011, aggarwal1989communication, fich1993complexity} consists of $p$
processors which all have access to an unbounded shared memory module. Each
local operation takes unit time, as in the RAM model. Operations proceed
synchronously across all processors. As all processors can access the same
memory cell, different PRAM versions describe what forms of concurrent accesses
are allowed (i.e., how conflicting accesses are resolved). The
\emph{exclusive-read-exclusive-write} (EREW) PRAM allows no simultaneous
accesses. A \emph{concurrent-read-exclusive-write} (CREW) PRAM allows
simultaneous read accesses to the same memory location, but only sequential
writes. The \emph{concurrent-read-concurrent-write} (CRCW) PRAM allows both
simultaneous read and write accesses to the same memory location. To define the
behavior of potential \emph{write conflicts}, CRCW is further categorized. The
\emph{common} CRCW PRAM model assumes that, if there are any write conflicts,
all processors write the same value to the memory location. In the
\emph{arbitrary} CRCW PRAM model, one arbitrary processor succeeds in writing
if a write conflict happens. The \emph{priority} CRCW PRAM model assumes a
  total order on its processors and allows the processor with the highest
  priority to succeed in case of a write conflict.
While we do not explicitly use PRAM, we outline it as some considered
algorithms are developed in this model.
A $p$-processor PRAM machine~\cite{blelloch2010parallel, JaJa2011,
aggarwal1989communication, fich1993complexity} consists of $p$ processors which
all have access to an unbounded shared memory module. Each local operation
takes unit time. Operations proceed synchronously across all processors.
Different PRAM versions define what forms of concurrent accesses (i.e., taking
$O(1)$ time) \emph{to the same memory cell} are allowed. These are
\emph{exclusive-read-exclusive-write} (EREW) PRAM (allows only sequential
accesses), \emph{concurrent-read-exclusive-write} (CREW) PRAM  (allows
concurrent reads but only sequential writes), and
\emph{concurrent-read-concurrent-write} (CRCW) PRAM (allows concurrent reads
and writes). 
\fi

\macbs{Our Analyses vs.~PRAM}
In our W-D analysis, two used machine model variants (1) only need concurrent
reads and (2) may also need concurrent writes. These variants are analogous to
those of the well-known PRAM model~\cite{blelloch2010parallel, JaJa2011,
aggarwal1989communication, fich1993complexity}: CREW and CRCW, respectively.
Thus, when describing a W--D algorithm that \emph{only relies on concurrent
reads}, we use a term ``the CREW setting''.  Similarly, for a W--D algorithm
that \emph{needs concurrent writes}, we use a term ``the CRCW setting''.

\iftr
The machine model used in this work is similar to PRAM. One key difference is
that we do not rely on the (unrealistic) PRAM assumption of the synchronicity
of all steps of the algorithm (across \emph{all} processors).  
\fi
The well-known Brent's result states that any deterministic algorithm with
work~$W$ and depth~$D$ can be executed on $P$ processors in time~$T$ such that
$\max\{W/P, D\} \le T \le W/P + D$~\cite{brent1974parallel}. Thus, all our
results are applicable to a PRAM setting.

\ifall
\macbs{Work-Depth}
The work-depth model~\cite{cormen2009introduction, blumofe1999scheduling,
blumofe1998space} avoids PRAM's machine-dependent details that can complicate
the algorithm design. Here, an algorithm is modeled as a \emph{directed-acyclic
graph} (DAG). Each \emph{node} in a DAG corresponds to a constant time operation. A
node's \emph{in-edges} model the data used for the operation. As operations run in
constant time, there are $O(1)$ in-edges per node. The \emph{out-edges} of a node
correspond to the computed output. A node can be executed as soon as all
predecessors have finished executing. The \emph{work} of an algorithm is then
defined as the total number of nodes and the \emph{depth} is defined as the
longest directed path in the DAG~\cite{Bilardi2011, blelloch2010parallel}. This
model is similar to PRAM, since a PRAM execution can be expressed as a
computation DAG.
\fi

\subsection{Compute Primitives}
\label{sec:primitives}

We use a $Reduce$ operation. It takes as input a set~$S = \{s_1, ...,
s_n\}$ implemented as an array (or a bitmap). It uses a function $f: S \to
\mathbb{N}$ called the \emph{operator}; $f(s)$ must be defined for any $s \in S$.
$Reduce$ calculates the sum of elements in $S$ with respect to $f$: $f(s_1)
+ ... + f(s_n)$. This takes $O(\log n)$ depth and $O(n)$
work in the CREW setting~\cite{ladner1980parallel, Snir2011}, where $n$ is the array size.
We use $Reduce$ to implement $Count(S)$, which computes the size~$|S|$.  For this, the
associated operator $f$ is defined as $f(s) = 1$ if $s \in S$, and $f(s) = 0$
otherwise.
\iftr
Reduce, as well as the more general PrefixSum operation, are well studied
parallel primitives, which are used in numerous parallel computations and
algorithms~\cite{Snir2011}.
\fi
We also assume a \emph{DecrementAndFetch} (DAF) to be
available; it atomically decrements its operand and returns a new
value~\cite{whasenplaugh2014ordering}. We use DAF to implement \emph{Join} to
synchronize processors (Join decrements its operand, returns the new value, and
releases a processor under a specified condition).
\iftr
For details, we refer to the paper of Hasenplaugh et al.
\cite{whasenplaugh2014ordering}.
\fi

\iftr
\subsection{Randomization}\label{sec:randomization}

Randomization has been used in numerous algorithms~\cite{luby1986simple,
krager1996anew, horae1962quicksort, lenstra1992arigorous,
krager1995arandomized, strassen1977afast, gazit1991anoptimal}. We distinguish
between \emph{Monte Carlo} algorithms, which return a correct result
w.h.p.\footnote{\footnotesize A statement holds with high probability (w.h.p.)
if it holds with a probability of at least $1-\frac{1}{n^c}$ for all $c$} and
  \emph{Las Vegas} algorithms, which always return the correct result but have
  probabilistic run-time bounds.  The JP algorithm, with an ordering heuristic
  that employs randomness like ADG, is a Las Vegas algorithm. 
For a large part of the analysis we use Random Variables to describe various
random events. To prove that statements hold w.h.p. we mainly use simple Markov
and Chernoff bounds.
Further, we also employ another technique of \emph{coupling} random variables
in our analysis. A coupling of two random variables $X$ and $Y$ is defined as
a new variable $(X', Y')$ over the joint probability space, such that the marginal
distribution of $X'$ and $Y'$ coincides with the distribution of $X$ and $Y$
respectively. More precisely we can define this property as follows: Let $f(x)$
be the distribution of $X$ and $g(y)$ be the distribution of $Y$. Then a random
variable $(X',Y')$ is a coupling of $X$ and $Y$ if and only if both statements
$\sum_{y}{Pr[X' = x, Y' = y]} \equiv f(x)$ and
$\sum_{x}{Pr[X' = x, Y' = y]} \equiv g(y)$ hold~\cite{roch2015modern}.

\fi

\ifall
By the asynchronous design of JP, we only need to have access to randomness in
the preprocessing stage. This works since JP itself is not a randomized
algorithm and we can compute the randomized priority function $\rho$ in
advance. For all JP priority functions that we are aware of, one random number
per vertex suffices~\cite{whasenplaugh2014ordering, jones1993parallel}.
\fi

\section{Parallel Approximate Degeneracy Ordering} \label{ch:algorithm}

We first describe \textbf{ADG}, a parallel algorithm for computing a partial
approximate degeneracy ordering. ADG outputs vertex
\emph{priorities}~$\rho_{\text{ADG}}$, which are then used by our coloring
algorithms (Section~\ref{sec:col-algs}). Specifically, these priorities produce
an order in which to color the vertices (ties are broken randomly). 
%
% Note that when used with JP, ties between vertices with the same
% $\rho_{\text{ADG}}$ are broken with a random priority function.

\ifall
We first describe ADG, a parallel algorithm for computing a partial approximate
degeneracy ordering, which comes in two variants, ADG-AVG
(\cref{sec:algorithm_adg}) and ADG-MED (\cref{sec:algorithm_adg-med}).  These
variants offer slightly different tradeoffs between performance and accuracy.
Both variants compute vertex \emph{priorities}, which are then used by our
coloring algorithms (Section~\ref{sec:col-algs}).
Specifically, these priorities impose an ordering on vertices that is used when
scheduling which vertex is the next to be colored.
\fi

ADG is shown in Algorithm~\ref{alg:adg-avg}.  ADG is similar to
SL~\cite{matula1983smallest}, which iteratively removes vertices of the
smallest degree to construct the exact degeneracy ordering. \emph{{The key
difference and our core idea is to repeatedly remove \textbf{in parallel} all
vertices with degrees smaller than $(1 + \epsilon)\widehat{\delta}$}}. The
parameter $\epsilon \ge 0$ controls the approximation accuracy. We
multiply $1+\epsilon$ by the \ul{average} degree $\widehat{\delta}$ as it
enables good bounds on quality and run-time, as we show in
Lemma~\ref{lma:adg-avg_runtime} and~\ref{lma:adg-avg_correctness}. Compared to SL (which has depth $O(n)$), ADG has
depth $O(\log^2 n)$ and obtains a partial $2(1+\epsilon)$-approximate
degeneracy ordering.

\ifall
Our first approximate degeneracy ordering (ADG-AVG) is shown in
Algorithm~\ref{alg:adg-avg}.  It computes the ordering
function~$\rho_{\text{ADGA}}$. 
ADG-AVG is similar to the sequential SL algorithm~\cite{matula1983smallest},
which iteratively removes vertices of the smallest degree to construct the
exact degeneracy ordering. \emph{\textbf{The key difference and our core idea
is to remove \textbf{in parallel} all vertices with degrees smaller than $(1 +
\epsilon)\widehat{\delta}$}}. Here, $\epsilon \ge 0$ is a parameter that
controls how far we agree to approximate the exact degeneracy ordering. We
multiply $1+\epsilon$ by the \ul{average} degree $\widehat{\delta}$ (thus the
name ``ADG-\ul{AVG}'') as it enables good bounds on quality and run-time, as we
show later.
Ultimately, ADG-AVG performs $O(\log{n})$ steps in the worst case (the exact
degeneracy ordering needs $O(n)$ steps), obtaining a partial
$2(1+\epsilon)$-approximate degeneracy ordering.
\fi

\begin{lstlisting}[float=h,label=alg:adg-avg,
caption=\textmd{\textbf{ADG}, our algorithm for computing the $2(1+\epsilon)$-approximate degeneracy ordering;
it runs in the CRCW setting.
%
}]
/* Input: A graph $G(V,E)$.
 * Output: A priority (ordering) function $\rho:V \to \mathbb{R}$. */

$D = [\ deg(v_1)\ deg(v_2)\ ...\ deg(v_n)\ ]$ //An array with vertex degrees
$\ell = 1$; $U = V$ //$U$ is the induced subgraph used in each iteration $\ell$|\label{ln:adg-U-start}|

while $U \neq \emptyset $ do: |\label{ln:adg_main_start}|
  $|U| = Count(U)$; //Derive $|U|$ using a primitive Count|\label{ln:adg-avg-deg-start}|
  $cnt = Reduce(U)$; //Derive the sum of degrees in $U$: $\sum_{v \in U}{D[v]}$
  $\widehat{\delta} = \frac{cnt}{\abs{U}}$ //Derive the average degree for vertices in $U$|\label{ln:adg-avg-deg}|

  //$R$ contains vertices assigned priority in a given iteration:
  $R = \{ u \in U \mid D[u] \leq (1+\epsilon)\widehat{\delta}\ \}$ |\label{ln:R_def}|
  UPDATE($U$, $R$, $D$) //Update $D$ to reflect removing $R$ from $U$ |\label{ln:UPD}|
  $U = U \setminus R$ //Remove selected low-degree vertices (that are in $R$)|\label{ln:remove_R}|
  for all $v \in R$ do in parallel: //Set the priority of vertices
    $\rho_{\text{ADG}}(v) = \ell$ //The priority is the current iteration number $\ell$
  $\ell = \ell + 1$ |\label{ln:adg_main_end}|

//Update $D$ to reflect removing vertices in $R$ from a set $U$:
UPDATE($U$, $R$, $D$):
for all $v \in R$ do in parallel:
	for all $u \in N_U(v)$ do in parallel:
		DecrementAndFetch($D[u]$)
\end{lstlisting}
%
%%% CRCW

\ifsq\enlargethispage{\baselineskip}\fi

\ifall
The implementation has a similar structure as the $(2+\epsilon)$-approximate
maximal densest subgraph algorithms from a paper of Dhulipala et al.
\cite{dhulipala2018theoretically}. 
\fi

In ADG, we maintain a set~$U \subseteq V$ of active vertices that starts as $V$
(Line~\ref{ln:adg-U-start}). In each step
(Lines~\ref{ln:adg-avg-deg-start}--\ref{ln:adg_main_end}), we use $\epsilon$
and $\widehat{\delta}$ to select vertices with small enough degrees
(Line~\ref{ln:R_def}); The average degree $\widehat{\delta}$ is computed in
Lines~\ref{ln:adg-avg-deg-start}--\ref{ln:adg-avg-deg}.  The selected vertices
form a set $R$ and receive a priority $\rho_{\text{ADG}}$ equal to the step
counter~$\ell$. We then remove them from the set~$U$ (Line~\ref{ln:remove_R})
and update the degrees~$D$ accordingly (Line~\ref{ln:UPD}).  We continue until
the set $U$ is empty.
\iftr
When used with JP, ties between vertices with the same $\rho_{\text{ADGA}}$ are
broken with another priority function $\rho'$ (e.g., random ordering).
\fi

\textbf{\ul{Design Details} }
\iftr
For the theoretical analysis we use the following design assumptions.
\fi
We implement $D$ as an array and use $n$-bit dense bitmaps for $U$ and $R$.
This enables updating vertex degrees in $O(1)$ and resolving $v
\in U$ and $v \in R$ in $O(1)$ time.
%
%\iftr
%
Constructing $R$ in each step can be
implemented in $O(1)$ depth and $O(\abs{U})$ work. The operation $U = U
\backslash R$ takes $O(1)$ depth and $O(\abs{R})$ work by overwriting the
bitmap~$U$.
To calculate the average degree on Line~\ref{ln:adg-avg-deg}, we derive
$|U|$ and sum all degrees of vertices in~$U$. The former is done with a
$Count$ over $U$. The latter uses $Reduce$ with the
associated operator $f(v) = D[v]$. As both $Reduce$ operations run in
$O(\log n)$ depth and $O(\abs{U})$ work, the same holds for the average degree
calculation.

\ifsq\enlargethispage{\baselineskip}\fi

\textbf{\ul{Depth} }
First, note that each line in the \texttt{while} loop runs in $O(\log n)$ depth, 
as discussed above.  We will now prove that the \texttt{while} loop iterates
$O(\log n)$ many times, giving the total ADG depth of $O(\log^2 n)$.  The
key notion is that, in each iteration, we remove \emph{a constant fraction of
vertices} due to the way that we construct $R$ (based on the \emph{average degree}
$\widehat{\delta}$).

%%% CREW -> CRCW
%
\begin{lma}
\label{lma:adg-avg_runtime}
For a constant $\epsilon > 0$, ADG does $O(\log n)$ iterations and has
$O(\log^{2}{n})$ depth in the CRCW setting.
\end{lma}

\begin{proof}
At each step~$\ell$ of the algorithm we can have at most $\frac{n}{1+\epsilon}$
vertices with a degree larger than $(1+\epsilon)\widehat{\delta}_\ell$. This
can be seen from the fact, that the sum of degrees in the current subgraph can
be at most $n$ times the average degree $\widehat{\delta}_\ell$. For vertices
with a degree exactly $(1+\epsilon)\widehat{\delta}_\ell$ we get
$\frac{n}{1+\epsilon} \cdot (1+\epsilon)\widehat{\delta}_\ell =
n\widehat{\delta}_\ell$, which would result in a contradiction if we had more
than $\frac{n}{1+\epsilon}$ vertices with larger degree. Thus, if we remove all
vertices with degree $\leq (1+\epsilon)\widehat{\delta_\ell}$, \emph{we remove
a constant fraction of vertices in each iteration} (at least
$\frac{\epsilon}{1+\epsilon}n$ vertices), which implies that ADG performs
$O(\log{n})$ iterations in the worst case, immediately giving the $O(\log^2 n)$
depth.
\iftr
To see this explicitly, one can define a simple recurrence relation for the
number of iterations $T(n) \leq 1 + T\left(\frac{1}{1+\epsilon}n\right)$, $T(1)
= 1$; solving it gives $T(n) \leq \left \lceil\frac{\log n}{\log{ (1
+\epsilon)}} + 1 \right \rceil \in O(\log n)$.
\fi
\end{proof}

\textbf{\ul{Work} }
\ifconf
The proof of {work} is similar and included in the extended report due to space
constraints; it also uses the fact that a constant fraction of vertices is
removed in each iteration.
\fi
\iftr
The proof of {work} is similar; it also uses the fact that a constant fraction
of vertices is removed in each iteration.
Intuitively, (1) we show that each \texttt{while} loop iteration performs
$O\left(\left(\sum_{v \in R}{deg(v)} \right) + \abs{U_i}\right)$ work (where
$U_i$ is the set $U$ in iteration~$i$), and (2) we bound
$\sum_{i=1}^{k}{\abs{U_i}}$ by a geometric series, implying that it is still in
$O(n)$.
\fi

\iftr
\begin{lma}
\label{lma:adg-avg_work}
For a constant $\epsilon > 0$, ADG does $O(n+m)$ work in the CRCW setting.
\end{lma}

\begin{proof}
Let $k$ be the number of iterations we perform and let $U_i$ be the set $U$ in
iteration $i$. To calculate the total work performed by ADG, we first consider
the work in one iteration.  As explained in ``Design Details'', deriving the
average degree takes $O(\abs{U_i})$ work in one iteration. Initializing $R$
takes $O(\abs{U_i})$ and removing $R$ from $U_i$ takes $O(\abs{R})$. UPDATE
takes $O\left(\sum_{v \in R}{deg(v)}\right)$ work.  Thus, the total work in one
iteration is in $O\left(\left(\sum_{v \in R}{deg(v)} \right) +
\abs{U_i}\right)$.
As each vertex becomes included in $R$ in a single unique iteration, this gives
$\sum_{i=1}^k \sum_{v \in R_i}{deg(v)} \in O(m)$.
Moreover, since we remove a constant number of vertices in each iteration from
$U$ (at least $\frac{\epsilon}{1+\epsilon}$ as shown above in the proof of the
depth of ADG), we can bound $\sum_{i=1}^{k}{\abs{U_i}}$ by a geometric series,
implying that it is still in $O(n)$. This can be seen from the fact that
$\sum_{i=1}^{k}{\abs{U_i}} \leq \sum_{i=0}^{k}{\left ( \frac{1}{1+\epsilon}
\right)^i n} \leq \sum_{i=0}^{\infty}{\left ( \frac{1}{1+\epsilon} \right)^i n}
= \frac{(1+\epsilon)}{\epsilon}n$ if $\frac{1}{1+\epsilon} < 1$ (which holds as
$\epsilon > 0$).
Ultimately, we have $O\left(\sum_{i=0}^k{\left(\sum_{v \in R_i}{deg(v)} \right)
+ \abs{U_i}}\right) \in O(m) + O(n) \in O(m+n)$.
\end{proof}
\fi

\textbf{\ul{Approximation ratio} }
We now prove that the approximation ratio of ADG on the degeneracy order is
$2(1+\epsilon)$.
First, we give a small lemma used throughout the analysis.

\begin{lma} \label{lma:adeg}
  Every induced subgraph of a graph~$G$ with degeneracy~$d$, has an average
  degree of at most $2d$.
\end{lma}

\begin{proof}
  By the definition of a $d$-degenerate graph, in every induced subgraph~$G[U]$,
  there is a vertex $v$ with $deg_U(v) \leq d$. If we remove $v$ from $G[U]$, at
  most $d$ edges are removed. Thus, if we iteratively remove such vertices from
  $G[U]$, until only one vertex is left, we remove at most $d \cdot (\abs{U} -
  1)$ edges. We conclude that $\widehat{\delta}(G[U]) = \frac{1}{\abs{U}} \sum_{v
    \in U}{deg_{U}(v)} \leq 2d$.
\end{proof}

\begin{lma} \label{lma:adg-avg_correctness}
ADG computes a partial $2(1+\epsilon)$-approximate degeneracy ordering
of~$G$.
\end{lma}

\begin{proof}
By the definition of~$R$ (Line~\ref{ln:R_def}), all vertices removed in step
$\ell$ have a degree of at most $(1+\epsilon)\widehat{\delta}_\ell$, where
$\widehat{\delta}_\ell$ is the average degree of vertices in subgraph~$U$ in
step~$\ell$. From Lemma~\ref{lma:adeg}, we know that $\widehat{\delta}_\ell
\leq 2d$. Thus, each vertex has a degree of at most
$2(1+\epsilon)d$ in the
subgraph~$G[U]$ (in the current step). Hence, each vertex has at most
$2(1+\epsilon)d$ neighbors that are ranked equal or higher.  The result follows
by the definition of a partial $2(1+\epsilon)$-approximate degeneracy order.
\end{proof}

\ifrev\marginpar{\vspace{4em}\colorbox{orange}{\textbf{R-3}}}\fi

\subsection{Comparison to Other Vertex Orderings}\label{sec:others}
We analyze orderings in Table~\ref{tab:gc-heur}. While
SLL and ASL heuristically approximate SL, which computes a degeneracy
ordering, they do not offer guaranteed approximation factors.
\emph{Only ADG comes with provable bounds on the accuracy
of the degeneracy order while being (provably) parallelizable}. 
\ifall\maciej{fix}
In table ~\ref{tab:gc-heur} we indicate in column ``D`` if a heuristic produces
an ordering, which is a degeneracy ordering or a provable approximation
thereof. Column ``B`` indicates if the ordering has to be computed beforehand
or if it can be computed during the coloring phase. Since later iterations of
ADG/SL/SLL are dependent on results of the steps before (vertices are removed
from the graph) we can not calculate the ordering efficiently at each coloring
step. In contrast a value of a ordering like LF for a vertex single $v$ can be
computed efficiently by just looking at it's degree or neighborhood.
\fi

\begin{table}[h]
\ifsq\vspace{-0.5em}\fi
\centering
\setlength{\tabcolsep}{1pt}
\ifsq
\renewcommand{\arraystretch}{0.7}
\fi
\scriptsize
\sf
\begin{tabulary}{\columnwidth}{lllll}
\toprule
\textbf{Ordering heuristic} & \textbf{Time / Depth} & \textbf{Work} & \textbf{F.?} & \textbf{B., Approx.?} \\ 
\midrule
\iftr
FF (first-fit)~\cite{welsh1967an} & $O(1)$ & $O(1)$ & n/a & n/a \\
R (random)~\cite{jones1993parallel, whasenplaugh2014ordering} & $O(1) $ & $O(n)$ & \faThumbsOUp & n/a \\
ID (incidence-degree)~\cite{coleman1983estimation} & $O(n+m)$ & $O(n+m)$ & n/a & n/a \\
SD (saturation-degree)~\cite{brelaz1979new, whasenplaugh2014ordering} & $O(n+m)$ &  $O(n+m)$  & n/a & n/a \\
LF (largest-degree-first)~\cite{whasenplaugh2014ordering} & $O(1)$ & $O(n)$ & \faThumbsOUp & n/a \\
LLF (largest-log-degree-first)~\cite{whasenplaugh2014ordering} & $O(1)$ & $O(n)$ & \faThumbsOUp & n/a \\
\fi
SLL (smallest-log-degree-last)~\cite{whasenplaugh2014ordering} & $O(\log{\Delta} \log{n})$ & $O(n + m)$ & \faThumbsOUp & \faThumbsDown\\
SL (smallest-degree-last)~\cite{whasenplaugh2014ordering, matula1983smallest} & $O(n)$ & $O(m)$ & \faThumbsOUp & \faThumbsOUp\ exact \\
ASL (approximate-SL)~\cite{patwary2011new} & $O(n)$ & $O(m)$ & \faThumbsOUp & \faThumbsDown \\
\all{\rowcolor{yellow} (?)ASLL (approximate-SLL)~\cite{patwary2011new} & - & $O(m)$ & \faThumbsOUp & \faThumbsDown  \\}
\midrule
\textbf{ADG~[approx.~degeneracy]} & $O\left(\log^2{n}\right)$ & $O(n + m)$ & \faThumbsOUp & \faThumbsOUp\ $2(1+\epsilon)$ \\
%
%\textbf{ADG-MED~[approx.~degeneracy]} & $O\left(\log^2{n}\right)$ & $O(n + m)$ & \faThumbsOUp & \faThumbsOUp\ $4$ \\
%
\bottomrule
\end{tabulary}
\vspace{-0.5em}
\caption{
\iftr
\footnotesize
\else
\ssmall
\fi
\textbf{Ordering heuristics related to
the degeneracy ordering}.
\textbf{``F.~(Free)?'':} Is the scheme free from concurrent writes?
\textbf{``B.~(Bounds)?'':} Are there 
\textbf{provable} bounds and approximation ratio on degeneracy ordering?
%
% \textbf{``EREW, CREW''}: well-known variants of the PRAM model.
%
``\faThumbsOUp'': support, 
``\faThumbsDown'': no support.
Notation is explained in Table~\ref{tab:symbols}
and in Section~\ref{sec:back}.
}
\ifsq\vspace{-1.5em}\fi
\label{tab:gc-heur} \end{table}

\ifrev\marginpar{\vspace{-5em}\colorbox{orange}{\textbf{R-3}}}\fi

\iftr

\subsection{Using Only Concurrent Reads}
\label{sec:adg-crew}

So far, we have presented ADG in the CRCW setting.  We now discuss
modifications required to make it work in the CREW setting.  The key change is
to redesign the UPDATE routine, as illustrated in
Algorithm~\ref{alg:adg-avg-update-crew}.
This preserves the $O(\log^2 n)$ depth, but the work increases
to $O(m + nd)$.

%Calculating the average degree on Line~\ref{ln:adg-avg-deg} needs deriving
%$|U|$ and summing all degrees of vertices in~$U$. The former is done with a
%$Count$ implemented as $Reduce$ over $U$. The latter uses $Reduce$ with the
%associated operator $f(v) = D[v]$. Since both $Reduce$ operations run in
%$O(\log n)$ depth and $O(\abs{U})$ work, the same holds for the average degree
%calculation.

\begin{lstlisting}[aboveskip=0em,abovecaptionskip=0.0em,float=h,label=alg:adg-avg-update-crew,
caption=\textmd{The modified version of the UPDATE routine from \textbf{ADG}
(Algorithm~\ref{alg:adg-avg}) that works in the CREW setting.
%
}]
UPDATE($U_i$, $R_i$, $D$): //$U_i$ and $R_i$ are sets $U$ and $R$ in iteration $i$
  for all $v \in U_i$ do in parallel:
    $D[v] = D[v] - Count(N_{U_i}(v) \cap R_i)$
\end{lstlisting}

\begin{lma}
\label{lma:adg-avg_work}
A variant of ADG as specified in Algorithm~\ref{alg:adg-avg-update-crew} and
in~\cref{sec:adg-crew} does $O(m + nd)$ work in the CREW setting.
\end{lma}

\begin{proof}
The proof is identical to that for ADG in the CRCW setting, the difference is
in analyzing the impact of UPDATE on total work.
To compute $Count(N_{U_i}(v) \cap R_i)$
(cf.~Algorithm~\ref{alg:adg-avg-update-crew}), we use a $Reduce$ over
$N_{U_i}(v)$. Since $R_i \subseteq U_i$, we can define the operator $f$ as
$f(v) = 1$ if $v \in R_i$, and $f(v) = 0$ otherwise.
Thus, when computing total work in iteration~$i$, instead of $\sum_{v \in
R_i}{deg_{U_i} (v)}$, we must consider $\sum_{v \in U_i}{deg_{U_i} (v)}$.
Consequently, we have $\sum_{i=1}^k \sum_{v \in U_i}{deg_{U_i} (v)} \le
\sum_{i=1}^k (2d \cdot |U_i|) \le 2d \cdot \sum_{i=1}^k |U_i| \in O(nd)$ (the
first inequality by Lemma~\ref{lma:adeg}, the second one by the observation
that we remove a constant fraction of vertices in each iteration, cf.~the proof
of Lemma~\ref{lma:adg-avg_runtime}).
\end{proof}

\fi

\section{Parallel Graph Coloring}
\label{sec:col-algs}

\ifsq\enlargethispage{\baselineskip}\fi

We now use our approximate degeneracy ordering to develop new parallel
graph coloring algorithms. We directly enhance the recent line of
works based on \emph{scheduling colors}, i.e., assigning colors to vertices
without generating coloring conflicts (\cref{sec:algorithm_jp-adg}). In two other 
algorithms, we \emph{allow conflicts} but we also \emph{provably} resolve
them \emph{fast} (\cref{sec:algorithm_dec-adg}, \cref{sec:dec-adg-itr}). 

\subsection{Graph Coloring by Color Scheduling (JP-ADG)}
\label{sec:algorithm_jp-adg}

We directly enhance recent works of Hasenplaugh et
al.~\cite{whasenplaugh2014ordering} by combining their Jones-Plassmann (JP)
version of coloring with our ADG, obtaining \textbf{JP-ADG}.
For this, we first overview JP and definitions used in JP. The JP algorithm
uses the notion of a computation DAG $G_\rho(V, E_\rho)$, which is a directed
version of the input graph $G$. Specifically, the DAG $G_\rho$ is used by JP to
\emph{schedule the coloring of the vertices}: The position of a vertex $v$ in
the DAG $G_\rho$ determines the moment the vertex $v$ is colored. The
DAG~$G_\rho$ contains the edges of $G$ directed from the higher priority to
lower priority vertices according to the priority function $\rho$, i.e.~,
$E_\rho = \{ (u,v) \in E \mid \rho(u) > \rho(v) \}$.

JP is described in Algorithm~\ref{jp}.
\iftr
As input, besides~$G$, it takes a priority function $\rho: V \to \mathbb{R}$
which defines a total order on vertices~$V$. 
\fi
First, JP uses $\rho$ to compute a DAG~$G_{\rho}(V, E_\rho)$, where edges
always go from vertices with higher $\rho$ to ones with lower $\rho$
(Alg.~\ref{jp}, Lines \ref{ln:jp_dag_start}--\ref{ln:jp_dag_end}). Vertices can
then be safely assigned a color if all neighbors of higher $\rho$ (predecessors
in the DAG) have been colored. The algorithm does this by first calling JPColor
with the set of vertices that have no predecessors
(Alg.~\ref{jp}, Lines 13--15).  JPColor then colors $v$ by calling GetColor,
which chooses the smallest color not already taken by
$v$'s predecessors. Afterwards, JPColor checks if any of $v$'s successors can
be colored, and if yes, it calls again JPColor on
them.
\iftr
As shown by Hasenplaugh et al.~\cite{whasenplaugh2014ordering}, 
this algorithm can run in $O(|\mathcal{P}|\log{\Delta} + \log{n})$ depth and $O(n +
m)$ work, where $|\mathcal{P}|$ is the size of the longest path in $G_{\rho}$. 
\fi

Now, in JP-ADG, we first call ADG to derive $\rho_{\text{ADG}}$. Then, we run
JP using $\rho_{\text{ADG}}$. More precisely, we use $ \rho = \langle
\rho_{\text{ADG}}, \rho_{\text{R}} \rangle$ where $\rho_R$ randomly breaks ties
of vertices that were removed in the same iteration in
Algorithm~\ref{alg:adg-avg}, and thus have the same rank in
$\rho_{\text{ADG}}$. 
The obtained JP-ADG algorithm is similar to past work based on JP in that it
follows the same ``skeleton'' in which coloring of vertices is guided by the
pre-computed order, on our case $\rho_{\text{ADG}}$. However, as we prove later
in this section, \emph{\textbf{using ADG is key} to our novel bounds on depth,
work, \ul{and} coloring quality}. Intuitively, ADG gives an ordering of
vertices in which each vertex has a \emph{bounded number of predecessors} (by
definition of $s$-degenerate graphs and graph degeneracy~$d$).  We use this to
bound coloring quality and sizes of subgraphs in~$G_\rho$. The latter enables
bounding the maximum path in~$G_\rho$, which in turn gives depth and work
bounds.

\ifsq\enlargethispage{\baselineskip}\fi

\iftr
As for different combinations of JP and orderings, we define them similarly to
past work. JP-R is JP with a random priority function $\rho_{\text{R}}$. JP-FF
uses the natural vertex order. JP-LF uses $\rho(v) = \langle deg(v),
\rho_{\text{R}} \rangle$ with a lexicographic order. JP-SL is defined by
$\rho(v) = \langle \rho_{\text{SL}},\rho_{\text{R}} \rangle $, with a
degeneracy ordering $\rho_{\text{SL}}$. JP-LLF is defined by $\rho = \langle
\ceil{\log{(deg(v))}}, \rho_{\text{R}} \rangle $ and JP-SLL by $ \rho = \langle
\rho_{\text{SLL}}, \rho_{\text{R}} \rangle $, where $\rho_{\text{SLL}}$ is the
order computed by the SLL algorithm from Hasenplaugh et
al.~\cite{whasenplaugh2014ordering}.
\fi

\begin{lstlisting}[float=t,label=jp, belowskip=0.0em,
caption=\textmd{\textbf{JP}, the Jones-Plassman coloring heuristic. 
With $\rho = \langle \rho_{\text{ADG}}, \rho_{\text{R}}
\rangle$, it gives \textbf{JP-ADG} that provides $\left(
2(1+\epsilon)d + 1 \right)$-coloring.}]
/* Input: A graph $G(V,E)$, a priority function $\rho$.
 * Output: An array $C$, it assigns a color to each vertex. */

//|\ul{Part 1}|: compute the DAG $G_{\rho}$ based on $\rho$
$C = [0\ 0\ ...\ 0]$ //Initialize colors
for all $v \in V $ do in parallel: |\label{ln:jp_dag_start}|
  //Predecessors and successors of each vertex $v$ in $G_{\rho}$:
  $pred[v] = \{ u \in N(v) \mid \rho(u) > \rho(v) \}$
  $succ[v] = \{ u \in N(v) \mid \rho(u) < \rho(v) \}$ |\label{ln:jp_dag_end}|
  //Number of uncolored predecessors of each vertex $v$ in $G_{\rho}$:
  $count[v] = \abs{pred[v]}$

//|\ul{Part 2}|: color vertices using $G_{\rho}$
for all $v \in V $ do in parallel:
  //Start by coloring all vertices without predecessors:
  if $ pred[v] == \emptyset $: JPColor$(v)$

JPColor($v$) //JPColor, a routine used in JP
  $C[v] = \text{\texttt{GetColor}}(v)$
  for all $u \in succ[v]$ in parallel:
    //Decrement $u$'s counter to reflect that $v$ is now colored:
    if $Join(count[u]) == 0$: |\label{ln:jp_color_join}|
      JPColor$(u)$ //Color $u$ if it has no uncolored predecessors 

GetColor($v$) //GetColor, a routine used in JPColor
  $C = \{1,2,\dots,\abs{pred[v]} + 1\}$
  for all $u \in pred[v]$ do in parallel: $C = C - \{ C[u] \}$
  return $\min{(C)}$ //Output: the smallest color available. 
\end{lstlisting}

We first prove a general property of JP-ADG, which we will use 
to derive bounds on coloring quality, depth, and work.

\begin{lma} \label{lma:colorbound}
JP, using a priority function $\rho$ that defines a $k$-approximate degeneracy
ordering, colors a graph~$G$ with at most $kd + 1$ colors, for $\epsilon > 0$.
\end{lma}

\begin{proof}
Since $\rho$ defines a $k$-approximate degeneracy ordering, any $v \in V$ has
at most $kd$ neighbors $v'$ with $\rho(v') \geq \rho(v)$ and thus at most $kd$
predecessors in the DAG. Now, we can choose the smallest color available from
$\{1,\dots, kd + 1 \}$ to color $v$, when all of its predecessors have been
colored.
\end{proof}

\textbf{\ul{Coloring Quality} }
The coloring quality now follows from the properties of the priority function
obtained with ADG.

\begin{crl}  \label{crl:adg-avg_color}
With priorities $\rho = \langle \rho_{\text{ADG}}, \rho_{\text{R}}
\rangle$, JP-ADG colors a graph with at most $2(1+\epsilon)d + 1$ colors, for
$\epsilon > 0$.
\end{crl}

\iftr
\begin{proof}
Since $\langle \rho_{\text{ADG}}, \rho_R \rangle$ is a $2(1+\epsilon)$-approximate
degeneracy ordering, the result follows from Lemma~\ref{lma:colorbound}
and~\ref{lma:adg-avg_correctness}.
\end{proof}
\fi

\textbf{\ul{Depth, Work} }
To bound the depth of JP-ADG, we follow the approach by Hasenplaugh et
al.~\cite{whasenplaugh2014ordering}. We analyze the expected length of the
longest path in a DAG induced by JP-ADG to bound its expected depth. 
\iftr
We first provide some additional definitions.  The \emph{induced subgraph} of
$G_\rho$ is a (directed) subgraph of $G_\rho$ induced by a given vertex set. 
\fi
Note that as $\rho$ is a total order on $V$, the DAG $G_\rho$ is \emph{strongly
connected}.  Finally, we denote $\overline{\rho} = \max_{v \in
V}{\{\rho_{\text{ADG}}(v)\}}$.

\begin{lma}
\label{lma:longestpath}
For a priority function $\rho = \langle\rho_{\text{ADG}}, \rho_{R}\rangle$, 
where $\rho_{\text{ADG}}$ is a partial $k$-approximate degeneracy ordering
for a constant $k > 1$, $\rho_{R}$ is a random priority function,
the expected length of the longest path in the DAG $G_{\rho}$ is
$O\left(d\log{n} + \frac{\log{d}\log^2{n}}{\log{\log n}} \right) \enspace $.
\end{lma}

\ifsq\enlargethispage{\baselineskip}\fi

\ifsq\vspace{-0.25em}\fi
\begin{proof}
Let $G_\rho(\ell)$ be the subgraph of $G_{\rho}$ induced by the vertex set $V(\ell) = \{v
\in V \mid \rho_{\text{ADG}} = \ell\}$. Let $\Delta_{\ell}$ be the maximal degree and
$\widehat{\delta}_{\ell}$ be the average degree of the subgraph $G_\rho(\ell)$.
%
%Let $G_{\ell}$ be the subgraph induced by all vertices with equal or higher priority than 
%$\ell$ in $\rho_{ADG}$ i.e. $\bigcup_{i = \ell}^{\overline{\rho}}{V(\ell)}$. 
%We note that $G_\rho(\ell)$ is a subgraph of $G_\ell$.

Since, by the definition of $G_\rho$, there can be no edges in $G_\rho$ that go
from one subgraph $G_\rho(\ell)$ to another $G_\rho(\ell')$ with $\ell' >
\ell$, we can see that a longest (directed) path $\mathcal{P}$ in $G_{\rho}$
will always go through the subgraph $G_\rho(\ell)$ in a monotonically
decreasing order with regards to $\ell$. Therefore, we can split $\mathcal{P}$
into a sequence of (directed) sub-paths $\mathcal{P}_1, \dots,
\mathcal{P}_{\overline{\rho}}$, where $\mathcal{P}_{\ell}$ is a path in
$G(\ell)$.
We have $\abs{\mathcal{P}} = \sum_{i \in \{\rho_{\text{ADG}}(v) \mid v \in
V\}}{\abs{\mathcal{P}_i}}$ and by Corollary~6 from past
work~\cite{whasenplaugh2014ordering}, the expected length of a longest sub-path
$\mathcal{P}_{\ell}$ is in $O(\Delta_{\ell} +
\log{\Delta_{\ell}}\log{n}/\log{\log n} )$, because $G_\rho(\ell)$ is induced
by a random priority function. By linearity of expectation, we have for the
whole path~$\mathcal{P}$:

\ifsq
\vspace{-1em}
\footnotesize
\fi
\begin{align}
\begin{split}
\mathbb{E}\left[\abs{\mathcal{P}}\right] = O\left(\sum_{\ell = 1}^{\overline{\rho}}{\left(\Delta_{\ell} + \log{\Delta_{\ell}} \cdot \frac{\log{n}}{\log{\log n}}\right)}\right)
\end{split}
\end{align}
\ifsq
\normalsize
\vspace{-1em}
\fi

Next, since $\rho_{\text{ADG}}$ is a partial $k$-approximate degeneracy ordering,
all vertices in $G(\ell)$ have at most $kd$ neighbors in $G(\ell)$. Thus,
$\Delta_{\ell} \leq kd$ holds. This and the fact that $\overline{\rho} \in
O(\log n)$ gives:

\ifsq
\vspace{-1em}
\small
\fi
\begin{align} \label{eq:ldd}
& \sum_{i = 1}^{\overline{\rho}}{\Delta_{i}} \leq \sum_{i = 1}^{\overline{\rho}}{d\cdot k} \in O(d\log{n}) & \\
& \sum_{i = 1}^{\overline{\rho}}{\log{\Delta_{i}}} \in O(\log{d}\log{n}) &
\end{align}
\ifsq
\normalsize
\vspace{-1em}
\fi

Thus, for the expected length of a longest path in $G$:

\ifsq
\vspace{-0.5em}
\small
\fi
\begin{equation}
\mathbb{E}\left[\abs{\mathcal{P}}\right] = O\left(d\log{n} + \frac{\log{d}\log^2{n}}{\log{\log n}}\right)
\end{equation}
\ifsq
\vspace{-1em}
\normalsize
\fi
\end{proof}

Our main result follows by combining our bounds on the longest path
$\mathcal{P}$ in the DAG $G_\rho$ and a result by Hasenplaugh et
al.~\cite{whasenplaugh2014ordering}, which shows that JP has $O(\log{n} +
\log{\Delta} \cdot \abs{\mathcal{P}} $) depth.

\begin{thm} \label{thm:adg-avg_rt}
JP-ADG colors a graph~$G$ with degeneracy $d$ in expected depth $O(\log^2{n}
+\log{\Delta} \cdot ( d\log{n} + \frac{\log{d}\log^2{n}}{\log{\log n}} ))$ and $O(n+m)$
work in the CRCW setting.
\end{thm}

\iftr
\begin{proof}
Since $\rho_{\text{ADG}}$ is a partial $2(1+\epsilon)$-approximate degeneracy ordering
(Lemma \ref{lma:adg-avg_correctness}) and since ADG performs at most $O(\log n)$ iterations
(Lemma~\ref{lma:adg-avg_runtime}), the depth follows from Lemma~\ref{lma:longestpath},
Lemma~\ref{lma:adg-avg_runtime} and past work~\cite{whasenplaugh2014ordering}, which shows
that JP runs in $O(\log{n} + \log{\Delta} \cdot \abs{\mathcal{P}} $) depth.
As both JP and ADG perform $O(n+m)$ work, so does JP-ADG.
\end{proof}
\fi

\subsection{Graph Coloring by Silent Conflict Resolution (DEC-ADG)}
\label{sec:algorithm_dec-adg}

Our second coloring algorithm takes a radical step to move away from the long
line of heuristics based on JP.
%
% to break free from relying on concurrent writes.
%
\emph{The key idea is \textbf{to use ADG} to \ul{decompose} the input graph
into low-degree partitions} (thus ``\textbf{\ul{DEC}-ADG}''), shown in
Algorithm~\ref{algo:adg-dec}. 
Here, ADG is again crucial to our bounds. Specifically, vertices with the same
ADG rank form a \emph{partition} that is ``low-degree'': it has a bounded
number of edges to any other such partitions (by the definition of ADG).
Each such partition is then colored separately, with a simple randomized scheme
in Algorithm~\ref{algo:sim-col}. This may generate coloring \emph{conflicts},
i.e., neighboring vertices with identical colors. Such conflicts are resolved
``silently'' by repeating the coloring on conflicting vertices as many times as
needed. As ADG bounds counts of edges between partitions, it also bounds
counts of conflicts, improving depth and quality. 

We first detail Algorithm~\ref{algo:adg-dec}. A single low-degree partition
$G(\ell)$ produced by the iteration~$\ell$ of ADG is the induced subgraph
of~$G$ over the vertex set~$R$ removed in this iteration (Line~\ref{ln:R_def},
Alg.~\ref{alg:adg-avg}). Formally, $G(i) = G[R(i)]$ where $R(i) = \{v \in V
\mid \rho(v) = i\}$ and $\rho$ is the partial $k$-approximate degeneracy order
produced by ADG (cf.~\cref{sec:back_deg}).
Thus, in DEC-ADG, we first run ADG to derive the ordering~$\rho$ and also the
number $\overline{\rho}$ of low-degree partitions ($\overline{\rho} \in O(\log
n)$). Here, we use ADG*, a slightly modified ADG that also records -- as an
array $\mathcal{G} \equiv [ G(1)\ ...\ G(\overline{\rho}) ]$ -- each low-degree
partition. Then, we iterate over these partitions (starting from
$\overline{\rho}$) and color each with SIM-COL (``SIMple coloring'',
Alg.~\ref{algo:sim-col}). 
\iftr
We discuss SIM-COL in more detail later in this section, its semantics are
that it colors a given arbitrary graph~$G$ (in our context $G$ is the $\ell$-th
partition $G(\ell)$) using $(1+\mu)\Delta$ colors, where $\mu > 0$ is an arbitrary
value.
\fi
To keep the coloring consistent with respect to already colored partitions, we
maintain bitmaps $B_v$ that indicate colors already taken by $v$'s neighbors in
already colored partitions: If $v$ cannot use a color~$c$, the $c$-th bit in
$B_v$ is set to 1.
\iftr
These bitmaps are updated before each call of SIM-COL
(Lines~\ref{ln:dec-adg_update_bmaps_start}--\ref{ln:dec-adg_update_bmaps_end}).
\fi

\begin{lstlisting}[float=t, label=algo:adg-dec,
caption=\textmd{\textbf{DEC-ADG}, the second proposed parallel coloring heuristic
that provides a $\left(2(1+\epsilon)d \right)$-coloring.
Note that we use factors $\epsilon/4$ and $\epsilon/12$ for more straightforward
proofs (this is possible as $\epsilon$ can be an arbitrary non-negative value).
%
% our graph coloring algorithm with $O(d \log^2 n)$
% depth, $O(n+m)$ work, and $2(1+\epsilon)d$ coloring quality.
%
}]
/* Input: $G(V,E)$ (input graph). 
 * Output: An array $C$, it assigns a color to each vertex. */

$C = [0\ 0\ ...\ 0]$ //Initialize an array of colors
//Run ADG* to derive a $2(1+\epsilon/12)$-approximate degeneracy ordering
//$\mathcal{G} \equiv [ G(1)\ ...\ G(\overline{\rho})]$ contains $\overline{\rho}$ low-degree partitions
//We have $G(i) = G[R(i)]$ where $R(i) = \{v \in V \mid \rho(v) = i\}$
$(\rho, \mathcal{G})$ = ADG*($G$) |\label{ln:adg-dec_part}|

//Initialize bitmaps $B_v$ to track colors forbidden for each vertex
$\forall_{v \in V}\ B_v = [00...0]$ //Each bitmap is $\left\lceil 2(1+\epsilon/12)(1+\mu)d \right\rceil + 1$ bits
SIM-COL($G(\overline{\rho})$, $\{B_v \mid v \in R(\overline{\rho})\}$) //First, we color $G(\overline{\rho})$

for $\ell$ from $\overline{\rho}-1$ down to $1$ do: //For all low-degree partitions
  $Q = R(\overline{\rho}) \cup \cdots \cup R(\ell + 1)$ //A union of already colored partitions
  for all $v \in R(\ell)$ do in parallel: |\label{ln:dec-adg_update_bmaps_start}|
  	for all $u \in N_Q(v)$ do in parallel: // For $v$'s colored neighbors
  		$B_v = B_v \cup C[u]$ // Update colors forbidden for $v$|\label{ln:adg-dec_update_bmaps}||\label{ln:dec-adg_update_bmaps_end}|
  SIM-COL($G(\ell)$, $\{B_v \mid v \in R(\ell)\}$) //Run Algorithm |\ref{algo:sim-col}||\label{ln:adg-dec_sim-col}|
\end{lstlisting}
\ifall
SIM-COL($G(\overline{\rho})$, $\{B_v \mid v \in R(\overline{\rho})\}$, $deg_{R(\overline{\rho})}$) //First, we color $G(\overline{\rho})$
 SIM-COL($G(\ell)$, $\{B_v \mid v \in R(\ell)\}$, $deg_{Q \cup R(\ell)}$) //Run Algorithm |\ref{algo:sim-col}||\label{ln:adg-dec_sim-col}||\iftr\m{fix} SIM-COL($G(\ell)$, $\{B_v \mid v \in R(\ell)\}$)\fi|
\fi

\ifsq\enlargethispage{\baselineskip}\fi

How large should $B_v$ be to minimize storage overheads but also ensure that
each vertex has enough colors to choose from? We observe that a single
bitmap~$B_v$ should be able to contain at most as many colors as neighbors
of~$v$ in a partition currently being colored ($G(\ell)$), and in all
partitions that have already been colored ($G(\ell'), \ell' > \ell$). We denote
this neighbor count with $deg_{\ell}(v)$. Observe that any $deg_{\ell}(v)$ is
\emph{at most} $kd = \lceil 2(1 + \epsilon) d \rceil$, as partitions are
created according to a partial
\ifconf
$k$-approximate degeneracy order where $k = 2(1+\epsilon)$.
\fi
\iftr
$k$-approximate degeneracy order where $k = 2(1+\epsilon)$ (note
that, in DEC-ADG, we use factors $\epsilon/4$ and $\epsilon/12$ instead of
$\epsilon$ for more straightforward proofs; this is possible as $\epsilon$ can
be an arbitrary non-negative value).
\fi
\ifall\m{?}
To ensure the proper size of $B_v$ (i.e., not too large but also such that each
vertex has enough colors to choose from), we use (in
Algorithm~\ref{algo:sim-col}) an additional degree function $deg'$,
in which we also count edges of partitions which have been colored earlier. As
partitions are created according to a partial $j$-approximate degeneracy order,
we know that each $v \in R(\ell)$ has at most $j \cdot d$ neighbors in
$G(\ell)$ and in partitions that have already been colored ($G(\ell'), \ell' >
\ell$). 
\fi
Now, when coloring a partition $G(\ell)$, we know that SIM-COL, by its design,
chooses colors for $v$ only in the range of $\{ 1 ... (1+\mu)deg_{\ell}(v) + 1 \}$
(Algorithm~\ref{algo:sim-col}, Line~\ref{ln:sim-col_choose_color}; as we will
show, using such a
range will enable the advantageous bounds for DEC-ADG).
Thus, it suffices to keep bitmaps of size $\ceil{(1+\mu)kd}+1$ for each
vertex, where $k =  2(1+\epsilon/12)$.

In \textbf{SIM-COL}, we color a single low-degree partition~$G(\ell) = (V(\ell),
E(\ell))$. 
SIM-COL takes two arguments: (1) the partition to be colored (it can be an
arbitrary graph~$G = (V,E)$ but for clarity we explicitly use $G(\ell) =
(V(\ell), E(\ell))$ that denotes a partition from a given iteration~$\ell$ in
DEC-ADG) and (2) bitmaps associated with vertices in a given partition~$R(i)$.
By design, SIM-COL delivers a $((1+\mu)\Delta)$-coloring; $\mu > 0$ can be 
an arbitrary value. To be able to derive the final bounds for DEC-ADG
we set $\mu = \epsilon / 4$.
$U$ are vertices still to be colored, initialized as $U = V(\ell)$.  In each
iteration, vertices in~$U$ are first colored randomly. Then, each vertex~$v$
compares its color~$C[v]$ to the colors of its \emph{active} (not yet colored)
neighbors in $N_U$ and checks if $C[v]$ is not already taken by other neighbors
inside \emph{and outside} of $V(\ell)$ (by checking $B_v$), see
Lines~\ref{ln:sim-col_compare_start}--\ref{ln:sim-col_compare_end}.  The goal
is to identify \emph{whether at least one such neighbor has the same color as
$v$}. For this, we use $Reduce$ over $N_U(v)$ with the operator $f$ defined as
$f_{eq}(u) = (C[v] == C[u])$ (the ``$==$'' operator works analogously to the
same-name operator in C++) and a simple lookup in $B_v$. If $v$ and $u$ have
the same color, $f_{eq}(u)$ equals $1$.  Thus, if \emph{any} of $v$'s neighbors
in $U$ have the same color as $v$, $Reduce(N_U(v), f_{eq}) > 0$.  This enables
us to attempt to re-color $v$ by setting $C[v] = 0$. 
If a vertex gets colored, we remove it from~$U$
(Line~\ref{ln:sim-col_remove_colored}) and update the bitmaps of its neighbors
(Line~\ref{ln:sim-col_update_bmaps}).  We iterate until $U$ is empty.
\ifall
\maciej{Address this bitmap size?}
\fi

% The ordering produces $k \in O(\log n)$ induced subgraphs $G(1), \dotsc,
% G(k)$ ($G(\ell)$ \maciej{$G(i)$?} consists of the vertices removed in the
% $i$-th iteration of the ordering.)

\ifsq\enlargethispage{\baselineskip}\fi

\begin{lstlisting}[float=t,label=algo:sim-col,
caption=\textmd{\textbf{SIM-COL}, our simple coloring routine used by DEC-ADG.
It delivers a $((1+\mu)\Delta)$-coloring, where $\mu > 0$ is an arbitrary
value. When using SIM-COL as a subroutine in DEC-ADG, we instantiate $\mu$
as $\mu = \epsilon/4$; we use this value in the listing above for concreteness.
}]
/* Input: $G(V(\ell),E(\ell))$ (input graph partition), $B_v$ (a bitmap with
 * colors forbidden for each $v$). Output: color $C[v] > 0$. */
$U = V(\ell)$
while $U \neq \emptyset$ do: 
  //|\ul{Part 1}|: all vertices in $U$ are colored randomly:
  for all $v \in U$ do in parallel: 
    choose $C[v]$ u.a.r. from $\left\{ 1, ..., (1+\mu) deg_{\ell}(v) \right\}$|\label{ln:sim-col_choose_color}|

  //|\ul{Part 2}|: each vertex compares its color to its active neighbors.|\label{ln:sim-col_compare_start}|
  //If the color is non-unique, the vertex must be re-colored.
  for all $v \in U$ do in parallel: 
    //$f_{eq}(v,u) = (C[v] == C[u])$ is the operator in $Reduce$. 
    if $Reduce(N_U(v), f_{eq}) > 0$ $\mid \mid$ $C[v] \in B_v$: $C[v]$ = $0$ |\label{ln:sim-col_compare_end}|

  //|\ul{Part 3}|: Update $B_v$ for all neighbors with fixed colors: 
  for all $v \in U$ do in parallel: 
    for all $u \in N_U(v)$ do in parallel:
      if $C[u] > 0$: $B_v = B_v \cup C[u]$ |\label{ln:sim-col_update_bmaps}|
  $U$ = $U \setminus \{v \in U \mid C[v] > 0\}$ //Update vertices still to be colored |\label{ln:sim-col_remove_colored}| 
\end{lstlisting}

\ifall
|\iftr choose $C[v]$ u.a.r. from $\left\{ 1, ..., \left\lceil (1+\epsilon/4) \cdot deg'(v) \right\rceil \right\}$\fi|
$deg'(v)$, additional degree 
* function with vertices only form this and already colored partition.
choose $C[v]$ u.a.r. from $\left\{ 1, ..., \left\lceil (1+\epsilon/4) \cdot deg'(v) \right\rceil \right\}$|\label{ln:sim-col_choose_color}||\iftr choose $C[v]$ u.a.r. from $\left\{ 1, ..., \left\lceil (1+\epsilon/4) \cdot deg'(v) \right\rceil \right\}$\fi|
\fi

% TODO : Delete ?
%
% \textbf{How does DEC-ADG avoid concurrent writes?} We ensure that
% all writes are to distinct memory locations: in Algorithm~\ref{algo:adg-dec}
% (Line~\ref{ln:adg-dec_update_bmaps}), each vertex only updates its bitmap
% $B_v$ in parallel, the updates over $v$'s neighbors are serialized. In
% Algorithm~\ref{algo:sim-col}, the same holds for the bitmap updates on
% Line~\ref{ln:sim-col_update_bmaps}.

\textbf{\ul{Depth, Work} }
We now prove the time complexity of DEC-ADG. \textbf{The key observation} is that
the probability that a particular vertex becomes \emph{inactive} (permanently
colored) is constant regardless of the coloring status of its neighbors.
\textbf{The key proof technique} is to use Markov and Chernoff Bounds.

\iftr
Before we proceed with the analysis, we provide some definitions.
For each round~$\ell$ of SIM-COL (Algorithm~\ref{algo:sim-col}), we define an
indicator random variable~$X_v$ to refer to the event in which a vertex~$v$
gets removed from $U$ (i.e., becomes colored and thus inactive) in this
specific round~$\ell$. The vertex $v$ is removed if and only if the color $C[v]$, which is
selected on Line~\ref{ln:sim-col_choose_color}, is \emph{not} used by some
neighbor of $v$ (i.e., this color is not in $B_v$) and no active neighbor chose
$C[v]$ in this round. The random variable $\overline{X}_v$ indicates the complement
of event $X_v$ (i.e., a vertex~$v$ is not removed from $U$ in a given round).
%
% Let $Y_e$ be the same variable for an edge $e$.  An edge $e$ becomes inactive
% if one of its endpoints becomes inactive.
%
% As these events do not depend on the neighbors of $v$, all random variables
% $Z_v$ are \emph{mutually independent} and all random variables $\overline{Z}_v$
% are also \emph{mutually independent}.
%
%
Next, let $Z$ be a Bernoulli random variable with probability $Pr[Z = 1]
\equiv p = 1-\frac{1}{1+\mu}$, and let $\overline{Z}$ be the complement of $Z$.
Finally, we use a concept of \emph{stochastic dominance}: For two random variables
A and B that are defined over the same set of possible outcomes, we say that
$A$ stochastically dominates ~$B$ if and only if
$Pr[A \ge z] \ge Pr[B \ge z]$, for all $z$.

In the following, we show that the event of an arbitrary vertex~$v$
becoming deactivated ($X_v = 1$) is at least as probable as $Z = 1$.
This will enable us to use these \emph{mutually independent}
variables $Z$ to analyze the time complexity of SIM-COL and DEC-ADG.
\fi

\begin{clm}\label{clm:vertex}
In every iteration, for every vertex $v$, the probability that the vertex $v$
becomes inactive is at least $1-\frac{1}{1+\mu}$.
%
% for all $v \in V$ holds, that $Pr[X_v = 1] \geq Pr[Z = 1]$
%
\end{clm}

\begin{proof}
The probability that $v$ becomes inactive in any iteration ($Pr[X_v = 1]$) is
at least $1-\frac{i}{(1+\mu)deg_{\ell}(v)}$, where $i$ is the number of
distinct colors in $B_v$ \emph{and} received from neighbors in this round. 
\iftr
This is because, in each iteration, while $v$ connects to vertices with a total of $i$
distinct colors, the total number of colors to be selected from is $(1+\mu) deg_{\ell}(v)$.
\fi
Now,
as $v$ can have at most $deg_{\ell}(v)$ colored neighbors, 
we get $1-\frac{i}{(1+\mu)deg_{\ell}(v)} \geq
1-\frac{deg_{\ell}(v)}{(1+\mu)deg_{\ell}(v)} = 1-\frac{1}{(1+\mu)}$,
\iftr
%
%As $\epsilon$ can be any non-zero value, we can simplify the final expression
%as $1-\frac{1}{(1+\epsilon)} = Pr[Z = 1]$.
%
\fi
which shows that $Pr[X_v = 1] \ge Pr[Z = 1]$ holds for all active $v$.
%and thus $X_v$
%statistically dominates $Z$ (as both $X_v$ and $Z$ are
%Bernoulli variables with values in $\{0,1\}$).
%
\end{proof}

Thus, in expectation, a constant fraction of the vertices becomes inactive in
every iteration.
Now, in the next step, we will apply Markov and Chernoff bounds to an
appropriately chosen binomial random variable, showing that the number of
vertices that are removed is concentrated around its expectation. Hence, the
algorithm terminates after $O(\log n)$ iterations.

\enlargeSQ

\begin{table*}[t]
\centering
\ifsq
\vspace{-1em}
%\scriptsize
\renewcommand{\arraystretch}{0.5}
\else
\renewcommand{\arraystretch}{1.6}
\fi
\setlength{\tabcolsep}{2pt}
%\footnotesize
%\small
\sf
\resizebox{\textwidth}{!}{ % fit table width to page
\begin{tabular}{lllllllllllllll}
\toprule
\multirow{2}{*}{\textbf{GC Algorithm}} & \multicolumn{11}{c}{\textbf{Theoretical Properties}} & \multicolumn{2}{c}{\textbf{Practical Properties}} & \multirow{2}{*}{\textbf{Remarks}} \\
\cmidrule(lr){2-12}  \cmidrule(lr){13-14}
 & \multicolumn{2}{l}{\textbf{Time (in PRAM) \ul{or} Depth (in W--D)}, \textbf{Work}} & & \textbf{Model} & \textbf{Quality} & \textbf{F.?} & \textbf{G.?} & \textbf{R.?} & \textbf{W.?} & \textbf{S.?} & \textbf{Q.?} & \textbf{Performance} & \textbf{Quality} & \\
\midrule
\multicolumn{15}{l}{\textbf{Class 1: Parallel coloring algorithms \emph{not based on JP}~\cite{jones1993parallel}}. Many are not used in practice (except for work by Gebremedhin~\cite{gebremedhin2000scalable} and related ones~\cite{ccatalyurek2012graph, rokos2015fast, deveci2016parallel, boman2005scalable}); we include them for completeness of our analysis.} \\
\midrule
(MIS) Alon~\cite{nalon1986afast} & $\mathbb{E}$ $O(\Delta \log{n})$ & \quad $O\left( m \Delta^2 \log n \right)$ & & CRCW & $\Delta + 1$ & \faThumbsDown & \faThumbsOUp & \faThumbsOUp & \faThumbsDown &  \faThumbsUp$^\dagger$ &  \faThumbsDown & --- & --- & $^\dagger$Depth depends on $\Delta$, $P = m \Delta$ \\
(MIS) Goldberg~\cite{goldberg1989anew} & $O\left(\Delta \log^4{n}\right)$ & \quad  $O\left((n+m)\log^4{n} \right)$ & & EREW & $\Delta + 1$ & \faThumbsOUp & \faThumbsOUp & \faThumbsDown & \faThumbsDown & \faThumbsUp$^\dagger$ & \faThumbsDown  & --- & --- & $^\dagger$Depth depends on $\Delta$, $P = (m+n)/\Delta$ \\
(MIS) Goldberg~\cite{goldberg1987paralleldelta, goldberg1987parallel} &$O( \log^*{n})$ & \quad  $O(n\log^*{n})$ & & EREW & $\Delta + 1$& \faThumbsOUp & \faThumbsDown$^\dagger$ & \faThumbsOUp & \faThumbsOUp$^\ddagger$ &  \faThumbsOUp & \faThumbsDown  & --- & --- & \makecell[l]{$^\dagger$Graphs with $\Delta \in O(1)$. $P = n$.\\$^\ddagger \log^* n$ grows very slowly.} \\
(MIS) Goldberg~\cite{goldberg1987parallel} & $O(\Delta \log{\Delta}(\Delta + \log^{*}{n}))$ & \quad  $O((n+m)\Delta \log{\Delta}(\Delta + \log^{*}{n}))$ & & EREW & $\Delta + 1$ & \faThumbsOUp & \faThumbsOUp & \faThumbsDown &  \faThumbsDown & \faThumbsUp$^\dagger$ & \faThumbsDown & --- & --- & $^\dagger$Depth depends on $\Delta$. $P = n+m$. \\
Luby~\cite{luby1993removing} & $O\left(\log^3{n} \log \log n \right)$ & \quad  $O\left((n+m)(\log^3{n}\log \log n)\right)$ & & CREW & $\Delta + 1$ & \faThumbsOUp & \faThumbsOUp & \faThumbsDown & \faThumbsDown & \faThumbsOUp & \faThumbsDown & --- & --- & $P = n+m$ \\
(MIS) Luby~\cite{luby1986simple} & $\mathbb{E}$ $O(\Delta \log{n})$ &  \quad $O(m\Delta \log n)$ & & CRCW & $\Delta + 1$ & \faThumbsDown & \faThumbsOUp & \faThumbsOUp &  \faThumbsDown &  \faThumbsUp$^\dagger$ &  \faThumbsDown  & \faStar \faStarHalfO \faStarO \faStarO & \faStar \faStarHalfO \faStarO \faStarO & $^\dagger$Depth depends on $\Delta$, $P = m$ \\
Gebremedhin~\cite{gebremedhin2000scalable} & $\mathbb{E}$ $O\left(\frac{\Delta n}{P}\right)$ & \quad  $O(\Delta n)$ & & CREW & --- & \faThumbsOUp & \faThumbsOUp & \faThumbsOUp & \faThumbsDown & \faThumbsDown & \faThumbsDown  & \faStar \faStar \faStarO \faStarO & \faStar \faStarHalfO \faStarO \faStarO & Assuming $P \le \frac{n}{2\sqrt{m}}$. \\
Gebremedhin~\cite{gebremedhin2000scalable} & $\mathbb{E}$ $O\left(\frac{\Delta P m}{n}\right) $ & \quad  $O\left(\frac{\Delta P^2 m}{n}\right) $ & & CREW & --- & \faThumbsOUp & \faThumbsOUp & \faThumbsOUp & \faThumbsUp$^\dagger$ & \faThumbsDown & \faThumbsDown  & \faStar \faStar \faStarO \faStarO & \faStar \faStarHalfO \faStarO \faStarO & \makecell[l]{Assuming $P > \frac{n}{2\sqrt{m}}$. $^\dagger$Work\\can be $\Omega(n+m)$ (for some $P$).} \\
%
% \makecell[l]{Çatalyürek~\cite{ccatalyurek2012graph}, Rokos~\cite{rokos2015fast}\\ Deveci~\cite{deveci2016parallel}, Boman~\cite{boman2005scalable}} & $O\left( \Delta \cdot I \right)$\ \ding{74} & \quad  $O\left( \Delta \cdot I \cdot p \right)$\ \ding{74} & & \noAnswer & \noAnswer &  \faThumbsOUp & \faThumbsDown & \faThumbsDown & \noAnswer & \noAnswer & \faThumbsDown & \faStar \faStar \faStar \faStarHalfO & ?? & $I$ is the number of iterations \\ 
%
\makecell[l]{ITRB (Boman et al.~\cite{boman2005scalable})} & $O\left( \Delta \cdot I \right)$\ \ding{74} & \quad  $O\left( \Delta \cdot I \cdot P \right)$\ \ding{74} & & \noAnswer & --- &  \faThumbsOUp & \faThumbsOUp & \faThumbsDown & \noAnswer$^\dagger$ & \noAnswer$^\dagger$ & \faThumbsDown & \faStar \faStarO \faStarO \faStarO & \faStar \faStar \faStarHalfO \faStarO  & \makecell[l]{$^\dagger$No detailed bounds available} \\
\makecell[l]{ITR (Çatalyürek~\cite{ccatalyurek2012graph} and\\ others~\cite{rokos2015fast, deveci2016parallel})} & $O\left( \Delta \cdot I \right)$\ \ding{74} & \quad  $O\left( \Delta \cdot I \cdot P \right)$\ \ding{74} & & \noAnswer & --- &  \faThumbsOUp & \faThumbsOUp & \faThumbsDown & \noAnswer$^\dagger$ & \noAnswer$^\dagger$ & \faThumbsDown & \faStar \faStar \faStar \faStarHalfO & \faStar \faStarHalfO \faStarO \faStarO & $^\dagger$No detailed bounds available \\ 
\makecell[l]{ITR-ASL (Patway et al.~\cite{patwary2011new})} &$O(n \cdot I)$\ \ding{74} & \quad $O(n \cdot I \cdot P)$\ \ding{74} & & \noAnswer & --- & \faThumbsOUp & \faThumbsOUp & \faThumbsDown & \noAnswer$^\dagger$ & \noAnswer$^\dagger$ & \faThumbsDown & \faStar \faStar \faStar \faStarO & \faStar \faStarO \faStarO \faStarO &  $^\dagger$No detailed bounds available \\
\iftr
%
%\rowcolor{yellow} \makecell[l]{Gebremedhin-ASL~\cite{gebremedhin2013colpack,patwary2011new,gebremedhin2000scalable}} &$O(n \cdot I)$\ \ding{74} & \quad $O(n \cdot I \cdot P)$\ \ding{74} & & \noAnswer & --- & \faThumbsOUp & \faThumbsOUp & \faThumbsOUp & \noAnswer & \noAnswer & \faThumbsDown & \noAnswer & \noAnswer &  \makecell[l]{$I$ is the number of iterations} \\
%
\fi
\midrule
\multicolumn{15}{l}{\textbf{Class 2: Coloring algorithms that are \emph{not parallel} and \emph{based on the Greedy coloring scheme}~\cite{welsh1967an}}. We include them as comparison baselines that deliver best-known coloring quality in practice.} \\
\midrule
%
% Greedy-FF~\cite{welsh1967an} & $O(n+m)$ & & $O(n+m)$ & Seq. & $\Delta + 1$ & \faThumbsOUp & \faThumbsDown &  \faThumbsDown & \faThumbsDown & \faThumbsDown & \faStarO \faStarO \faStarO \faStarO & \faStarHalfO \faStarO \faStarO \faStarO & --- \\
%
Greedy-ID~\cite{coleman1983estimation} & $O(n+m)$ & &  $O(n+m)$ & Seq. & $\Delta + 1$ & \faThumbsOUp & \faThumbsOUp & \faThumbsDown & \faThumbsDown & \faThumbsDown & \faThumbsDown & \faStarHalfO \faStarO \faStarO \faStarO  & \faStar \faStar \faStar \faStarHalfO & --- \\
Greedy-SD~\cite{brelaz1979new, whasenplaugh2014ordering} & $O(n+m)$ & &  $O(n+m)$ & Seq. & $\Delta + 1$ & \faThumbsOUp & \faThumbsOUp & \faThumbsDown & \faThumbsDown & \faThumbsDown & \faThumbsDown & \faStarHalfO \faStarO \faStarO \faStarO  & \faStar \faStar \faStar \faStar & --- \\
\midrule
\multicolumn{15}{l}{\textbf{Class 3: \emph{Parallel} heuristics that constitute the largest line of work into parallel graph coloring} and are fast in theory and practice. Most are based on JP~\cite{jones1993parallel}.} \\
\midrule
JP-FF~\cite{whasenplaugh2014ordering,welsh1967an} & \multicolumn{2}{l}{No general bounds; $\Omega{(n)}$ for some graphs} & $O(n+m)$ & W--D & $\Delta + 1$ & \faThumbsDown & \faThumbsOUp & \faThumbsDown & \faThumbsOUp & \faThumbsDown$^\dagger$ & \faThumbsDown & \faStar \faStar \faStar \faStarHalfO & \faStarHalfO \faStarO \faStarO \faStarO & $^\dagger$No general bounds \\
JP-LF~\cite{whasenplaugh2014ordering} & \multicolumn{2}{l}{No general bounds; $\Omega{\left(\Delta^2\right)}$ for some graphs} & $O(n+m)$ & W--D& $\Delta + 1$ & \faThumbsDown & \faThumbsOUp & \faThumbsOUp & \faThumbsOUp & \faThumbsDown$^\dagger$ & \faThumbsDown  & \faStar \faStar \faStar \faStarHalfO & \faStar \faStar \faStarO \faStarO &  $^\dagger$No general bounds \\
JP-SL~\cite{whasenplaugh2014ordering} & \multicolumn{2}{l}{No general bounds; $\Omega{(n)}$ for some graphs} & $O(n+m)$ & W--D & $d + 1$ & \faThumbsDown & \faThumbsOUp & \faThumbsOUp & \faThumbsOUp & \faThumbsDown$^\dagger$ & \faThumbsOUp$^\ddagger$ & \faStar \faStarO \faStarO \faStarO & \faStar \faStar \faStar \faStarO &  \makecell[l]{$^\dagger$No general bounds. $^\ddagger$Often, $d \ll \Delta$.} \\
JP-R~\cite{jones1993parallel} & \multicolumn{2}{l}{$\mathbb{E}$ $O\left(\frac{\log{n}}{\log{\log n}}\right)$} & $O(n+m)$ & W--D & $\Delta + 1$ & \faThumbsDown & \faThumbsDown$^\dagger$ & \faThumbsOUp & \faThumbsOUp & \faThumbsOUp & \faThumbsDown & \faStar \faStar \faStar \faStarHalfO & \faStar \faStarO \faStarO \faStarO & $^\dagger$Graphs with $\Delta \in O(1)$ \\
%
% This is not a duplicate!
%
JP-R~\cite{whasenplaugh2014ordering} & \multicolumn{2}{l}{$\mathbb{E}$ $O\left(\log{n} + \log{\Delta} \cdot \min{\left\{\sqrt{m}, \Delta +\frac{\log{\Delta}\log{n}}{\log{\log n}}\right\}} \right)$} & $O(n+m)$ & W--D & $\Delta + 1$ & \faThumbsDown & \faThumbsOUp & \faThumbsOUp & \faThumbsOUp & \faThumbsUp$^\dagger$ & \faThumbsDown  & \faStar \faStar \faStar \faStarHalfO & \faStar \faStarO \faStarO \faStarO & $^\dagger$Depth depends on $\sqrt{m}$ or $\Delta$ \\
JP-LLF~\cite{whasenplaugh2014ordering} & \multicolumn{2}{l}{$\mathbb{E}$ $O\left(\log{n} + \log{\Delta} \cdot \left( \min{\{\Delta,\sqrt{m}\}} + \frac{\log^2{\Delta}\log{n}}{\log{\log n}} \right)\right)$} & $O(n+m)$ & W--D &$\Delta + 1$ & \faThumbsDown & \faThumbsOUp & \faThumbsOUp & \faThumbsOUp & \faThumbsUp$^\dagger$ & \faThumbsDown & \faStar \faStar \faStar \faStarHalfO &\faStar \faStar \faStarHalfO \faStarO & $^\dagger$Depth depends on $\sqrt{m}$ or $\Delta$ \\
JP-SLL~\cite{whasenplaugh2014ordering} &\multicolumn{2}{l}{ $\mathbb{E}$ $O\left(\log{\Delta}\log{n} + \log{\Delta} \cdot \left( \min{\{\Delta,\sqrt{m}\}} + \frac{\log^2{\Delta}\log{n}}{\log{\log n}} \right)\right)$} & $O(n+m)$ & W--D & $\Delta + 1$  & \faThumbsDown & \faThumbsOUp & \faThumbsOUp & \faThumbsOUp & \faThumbsUp$^\dagger$ & \faThumbsDown  & \faStar \faStarHalfO \faStarO \faStarO & \faStar \faStar \faStar \faStarO & $^\dagger$Depth depends on $\sqrt{m}$ or $\Delta$ \\
JP-ASL~\cite{gebremedhin2013colpack,patwary2011new} & \multicolumn{2}{l}{$O(n \cdot I)$\ \ding{74}} & \makecell[c]{$O(n \cdot I \cdot P)$\ \ding{74}} & W--D & $\Delta + 1$ & \faThumbsDown & \faThumbsOUp & \faThumbsOUp & \noAnswer & \noAnswer & \faThumbsDown & \faStar \faStarO \faStarO \faStarO & \faStar \faStarO \faStarO \faStarO &  \makecell[l]{$I$ is the number of iterations} \\
\ifall\m{?}
(?)JP-ASLL~\cite{gebremedhin2014colpack,patwary2011new} & \multicolumn{2}{l}{No general bounds} & $O(n+m)$ & W--D & $\Delta + 1$ & \faThumbsDown & \faThumbsOUp & \faThumbsOUp & \faThumbsOUp & \faThumbsDown$^\dagger$ & \faThumbsDown & \noAnswer & \noAnswer &  \makecell[l]{$^\dagger$No general bounds.} \\
\fi
\midrule
\pmb{JP-ADG}~\textbf{[This Paper]} & \multicolumn{2}{l}{$\mathbb{E}$ $O\left(\log^2{n} +\log{\Delta} \cdot \left( d \log{n} + \frac{\log{d} \cdot \log^2{n}}{\log{\log n}}\right)\right)$} & $O(n+m)$ & W--D & $2(1+\epsilon)d + 1$ & \faThumbsDown & \faThumbsOUp & \faThumbsOUp & \faThumbsOUp & \faThumbsOUp$^\dagger$ & \faThumbsOUp$^\dagger$ & \faStar \faStar \faStarHalfO \faStarO & \faStar \faStar \faStar \faStarO  & $^\dagger$Often, $d \ll \Delta$ \\
\iftr
\makecell[l]{\pmb{JP-ADG-M}~\textbf{[This Paper]}\\(a variant described in~\cref{sec:design-impl})} & \multicolumn{2}{l}{$\mathbb{E}$ $O\left(\log^2{n} +\log{\Delta} \cdot \left( d \log{n} + \frac{\log{d} \cdot \log^2{n}}{\log{\log n}}\right)\right)$} & $O(n+m)$ & W--D & $4d + 1$ & \faThumbsDown & \faThumbsOUp & \faThumbsOUp & \faThumbsOUp & \faThumbsOUp$^\dagger$ & \faThumbsOUp$^\dagger$ & \faStar \faStar \faStarHalfO \faStarO & \faStar \faStar \faStar \faStarO  & $^\dagger$Often, $d \ll \Delta$ \\
\fi
\ifall
\pmb{JP-ADG-AVG}~\textbf{[This Paper]} & \multicolumn{2}{l}{$\mathbb{E}$ $O\left(\log^2{n} +\log{\Delta} \cdot \left( d \log{n} + \frac{\log{d} \cdot \log^2{n}}{\log{\log n}}\right)\right)$} & $O(n+m)$ & W--D & $2(1+\epsilon)d + 1$ & \faThumbsDown & \faThumbsOUp & \faThumbsOUp & \faThumbsOUp & \faThumbsOUp$^\dagger$ & \faThumbsOUp$^\dagger$ & \faStar \faStar \faStarHalfO \faStarO & \faStar \faStar \faStarHalfO \faStarO  & $^\dagger$Often, $d \ll \Delta$ \\
\pmb{JP-ADG-MED}~\textbf{[This Paper]} & \multicolumn{2}{l}{$\mathbb{E}$ $O\left(\log^2{n} +\log{\Delta} \cdot \left( d \log{n} + \frac{\log{d} \cdot \log^2{n}}{\log{\log n}}\right)\right)$} & $O(n+m)$ & W--D & $2(1+\epsilon)d + 1$ & \faThumbsDown & \faThumbsOUp & \faThumbsOUp & \faThumbsOUp & \faThumbsOUp$^\dagger$ & \faThumbsOUp$^\dagger$ & \faStar \faStar \faStarHalfO \faStarO & \faStar \faStar \faStarHalfO \faStarO  & $^\dagger$Often, $d \ll \Delta$ \\
\pmb{JP-ADG-O-AVG}~\textbf{[This Paper]} & \multicolumn{2}{l}{$\mathbb{E}$ $O\left(\log^2{n} +\log{\Delta} \cdot \left( d \log{n} + \frac{\log{d} \cdot \log^2{n}}{\log{\log n}}\right)\right)$} & $O(n+m)$ & W--D & $2(1+\epsilon)d + 1$ & \faThumbsDown & \faThumbsOUp & \faThumbsOUp & \faThumbsOUp & \faThumbsOUp$^\dagger$ & \faThumbsOUp$^\dagger$ & \faStar \faStar \faStarHalfO \faStarO & \faStar \faStar \faStarHalfO \faStarO  & $^\dagger$Often, $d \ll \Delta$ \\
\pmb{JP-ADG-O-MED}~\textbf{[This Paper]} & \multicolumn{2}{l}{$\mathbb{E}$ $O\left(\log^2{n} +\log{\Delta} \cdot \left( d \log{n} + \frac{\log{d} \cdot \log^2{n}}{\log{\log n}}\right)\right)$} & $O(n+m)$ & W--D & $2(1+\epsilon)d + 1$ & \faThumbsDown & \faThumbsOUp & \faThumbsOUp & \faThumbsOUp & \faThumbsOUp$^\dagger$ & \faThumbsOUp$^\dagger$ & \faStar \faStar \faStarHalfO \faStarO & \faStar \faStar \faStarHalfO \faStarO  & $^\dagger$Often, $d \ll \Delta$ \\
\pmb{DEC-ADG-AVG}~\textbf{[This Paper]} & \multicolumn{2}{l}{$O\left( d \log^2 n \right)$ w.h.p.} & $\mathbb{E}$ $O(n+m)$ & W--D & $(2+\epsilon)d$ & \faThumbsOUp & \faThumbsOUp & \faThumbsOUp & \faThumbsOUp$^\dagger$ & \faThumbsOUp & \faThumbsOUp & ?? & ?? & $^\dagger$Work efficient in expectation \\
\pmb{DEC-ADG-MED}~\textbf{[This Paper]} & \multicolumn{2}{l}{$O\left( d \log^2 n \right)$ w.h.p.} & $\mathbb{E}$ $O(n+m)$ & W--D & $\ceil{(4+\epsilon)d}$ & \faThumbsOUp & \faThumbsOUp & \faThumbsOUp & \faThumbsOUp$^\dagger$ & \faThumbsOUp & \faThumbsOUp & ?? & ?? & $^\dagger$Work efficient in expectation \\
\fi
\pmb{DEC-ADG}~\textbf{[This Paper]} & \multicolumn{2}{l}{\mbox{$O\left( \log d \log^2 n \right)$} w.h.p.} & \makecell[c]{$O(n+m)$\\(w.h.p.)} & W--D & ${(2+\epsilon)d}$ & \faThumbsUp$^*$ & \faThumbsOUp & \faThumbsOUp & \faThumbsOUp$^\dagger$ & \faThumbsOUp & \faThumbsOUp & \faStar \faStar \faStarO \faStarO & \faStar \faStar \faStar \faStarO & \makecell[l]{$^*$Using CREW gives $\mathbb{E} O(nd + m)$ \\ work. $^\dagger$Work efficient in expectation.} \\
\iftr
\makecell[l]{\pmb{DEC-ADG-M}~\textbf{[This Paper]}\\(a variant described in~\cref{sec:design-impl})} & \multicolumn{2}{l}{$O\left( \log d \log^2 n \right)$ w.h.p.} & \makecell[c]{$O(n+m)$\\(w.h.p.)} & W--D & ${(4+\epsilon)d}$ & \faThumbsUp$^*$ & \faThumbsOUp & \faThumbsOUp & \faThumbsOUp$^\dagger$ & \faThumbsOUp & \faThumbsOUp & \faStar \faStar \faStarO \faStarO & \faStar \faStar \faStar \faStarO & \makecell[l]{$^*$Using CREW gives $\mathbb{E} O(nd + m)$ \\ work. $^\dagger$Work efficient in expectation.} \\
\fi
%
% \pmb{DEC-ADG-ITR}~\textbf{[This Paper]} & \multicolumn{2}{l}{$O\left( I \cdot d\log n \right)$, {\footnotesize $O\left(\sum_{\ell = 1}^{\overline{\rho}} \left(\sum_{i=1}^{I} \left[ \widetilde{\Delta} \cdot \abs{U_i} + \sum_{v \in U_i}{deg(v)} \right] + \sum_{v \in R(\ell)}{deg(v)} \right) \right)$} } &  & W--D & $2(1+\epsilon)d + 1$ & \faThumbsOUp & \faThumbsOUp &\faThumbsDown & \faThumbsDown & \faThumbsDown & \faThumbsOUp & ?? & ?? & $I$ is the number of iterations in~\cite{ccatalyurek2012graph} \\
%
\pmb{DEC-ADG-ITR}~\textbf{[This Paper]} & \multicolumn{2}{l}{$O\left( I \cdot d\log n \right)$,\quad [work bound is complex, details in text (\cref{sec:dec-adg-itr})]} &  & W--D & $2(1+\epsilon)d + 1$ & \faThumbsOUp & \faThumbsOUp &\faThumbsDown & \faThumbsDown & \faThumbsDown & \faThumbsOUp & \faStar \faStar \faStar \faStarO & \faStar \faStar \faStar \faStarO & $I$ is the number of iterations in~\cite{ccatalyurek2012graph} \\
\bottomrule
\end{tabular}}
\ifsq
\vspace{-0.5em}
\fi
\caption{
\ifsq\ssmall\fi
\textbf{Comparison of parallel graph coloring algorithms.}
\textbf{``Greedy-X''} is a the Greedy sequential coloring scheme~\cite{welsh1967an}, with ordering~X.
\textbf{``JP-X''} is a Jones and Plassmann scheme~\cite{jones1993parallel}, with ordering~X.
\textbf{``(MIS)''} indicates that a given algorithm solves the Minimum Independent Set (MIS) problem, but
it can be easily converted into a parallel GC algorithm.
\textbf{``EREW, CREW, and CRCW''} are well-known variants of the PRAM model.
\textbf{``Performance''} and \textbf{``Quality''} summarize the run-times and
coloring qualities, based on the extensive existing
evaluations~\cite{whasenplaugh2014ordering, gebremedhin2000scalable,
jrallwright1995} and our analysis (\cref{ch:measurements}), focusing on modern
real-world graphs as input sets (we exclude Class~1 as it is not relevant
for practical purposes). 
\textbf{``Seg.''} is a simple sequential (RAM) model.
\textbf{``W--D''} indicates work-depth.
\textbf{``F.~(Free)?'':} Is the heuristic \ul{provably} free from using concurrent writes?
\textbf{``G.~(General)?'':} Does the algorithm work for general graphs?
\textbf{``R.~(Randomized)?'':} Is the algorithm randomized?
\textbf{``W.~(Work-efficient)?'':} Is the algorithm \ul{provably} work-efficient (i.e., does it take $O(n+m)$ work)?
\textbf{``S.~(Scalable)?'':} Is the algorithm \ul{provably} scalable (i.e., does it take $O(a \log^b n)$, where $a \ll n$ and $b \in \mathbb{N}$)?
\textbf{``Q.~(Quality)?'':} Does the algorithm \ul{provably} ensure the coloring quality better than the trivial $\Delta + 1$ bound?
``\faThumbsOUp'': full support, ``\faThumbsUp'': partial support, ``\faThumbsDown'': no support, ``\noAnswer'': unknown,
``\ding{74}'': bounds based on our best-effort derivation.
``\textbf{w.h.p.}'': with high probability.
$I$ is the number of iterations. $^\ddagger \log^* n$ grows very slowly.
All symbols are explained in Table~\ref{tab:symbols}.
\emph{The proposed schemes are the only ones with \ul{provably} good work \textbf{and} depth \textbf{and} quality.}
}
\ifsq\vspace{-2em}\fi
\label{tab:gc-all}
\end{table*}

\iftr

First, to proof the intuitive fact that $Z$ can be used to approximate
the number of vertices removed in each round, we use the technique of
coupling (See \cref{sec:randomization}) together with a handy equivalence
between stochastic dominance and the coupling of two random variables
\cite{roch2015modern}.

\begin{lma}\label{clm:bin_dominance}
	The random variable $X = \sum_{v \in U}{X_v}$ stochastically dominates
	$Y = \sum_{i=1}^{\abs{U}}{Z}$.
\end{lma}

\begin{proof}
	Without loss of generality let us assume that all vertices in $U$
	are numerated from $1$ to $\abs{U}$, such that we can write
	$X = \sum_{i=1}^{\abs{U}}{X_i}$. Further note, that we have shown
	in Claim \ref{clm:vertex} that the event $X_v = 1$ happens
	\emph{independently} with a probability of at least $1-\frac{1}{1+\mu}$
	and therefore $Pr[X_v = 1] \geq Pr[Z = 1]$ holds.
	We want to construct a coupling $(\widehat{X},\widehat{Y})$ with
	$Pr[\widehat{X}  \geq \widehat{Y}] = 1$ to show the stochastic
	dominance of $X$ \cite{roch2015modern}.
	
	First let $\widehat{Y}_i$ be equal to $Y_i$ for all $i=1..\abs{U}$
	and define $\widehat{Y}$ to be $\sum_{i=1}^{\abs{U}}{\widehat{Y}_i}$.
	Next define $\widehat{X}_i$ to be $0$ whenever $\widehat{Y}_i = 0$. If 
	$\widehat{Y}_i$ is $1$ we define the probability of $\widehat{X}_i$ to be
	$Pr[X_i = 1]$ conditioned under ${X_i \geq 1}$.
	
	We can see that by design $\widehat{X}_i \geq \widehat{Y}_i$ holds
	for all $i=1..\abs{U}$ and therefore, with Theorem 4.23 of other work
	\cite{roch2015modern}, conclude that X dominates Y.
	 
\end{proof}

\begin{lma}\label{lma:sim-col_dominance}
A random variable $\overline{X} = \sum_{v \in U}{\overline{X}_v}$ is stochastically
dominated by a random variable $\overline{Y} = \sum_{i=1}^{|U|}{\overline{Z}}$.
\end{lma}

\begin{proof}
From Claim \ref{clm:vertex} we already know that
$\forall v \in U. \ Pr[X_v = 1] \geq Pr[Z = 1]$. Since $X_v$ and $Z$ are
Bernoulli random variables, $Pr[\overline{Z} = 1] \geq Pr[\overline{X}_v = 1]$
holds for their complements. Now with a similar argument as shown in Lemma
\ref{clm:bin_dominance} we can again conclude the proof.

%Then, by the complementarity of
%\overline{X}_v$ and $X_v$, and by the statistical dominance of $X_v$ over $Z$,
%Therefore we have $\sum_{v \in U} Pr\left[\overline{X}_v \ge z\right] \leq \sum_{v \in U}
%Pr\left[\overline{Z} \ge z\right] = Pr\left[\overline{Y} \ge z \right]$.
%
\end{proof}

\fi

\begin{lma}\label{lma:sim-col_iter}
SIM-COL performs $O( \log n)$ iterations w.h.p. (for constant $\mu > 1$).
\end{lma}

\iftr
\begin{proof}
The number of vertices that are not removed (i.e., not permanently colored) in
one round can be expressed as $\overline{X} = \sum_{v \in U}{\overline{X}_v}$
($\overline{X}_v = 1$ means that $v$ is \emph{not} removed from~$U$ in a
given iteration). 
%
% From Claim~\ref{clm:vertex}, we know that $\overline{X}$ is stochastically
% dominated by the binomially distributed random variable $\overline{Y} =
% \sum_{i=1}^{n'}{ \overline{Z}}$.  Thus, for any $z$, we have that
% $Pr[\overline{Y} \ge z] \ge Pr[\overline{X} \ge z]$.
%
Now, we apply the Markov inequality to $\overline{Y}$: 
$$Pr\left[\overline{Y} \geq \frac{\mu |U|}{1+\mu}\right] \leq
\frac{\mathbb{E}\left[\overline{Y}\right]}{\frac{\mu |U|}{1+\mu}}
= \frac{\frac{|U|}{1+\mu}}{\frac{\mu |U|}{1+\mu}} = \frac{1}{\mu}.$$ 
As $\overline{Y}$ dominates $\overline{X}$ (Lemma~\ref{lma:sim-col_dominance}), we have 
$$Pr\left[\overline{Y} \geq \frac{\mu
|U|}{1+\mu}\right] \geq \Pr\left[\overline{X} \geq \frac{\mu
|U|}{1+\mu}\right].$$
Consequently, $\Pr\left[\overline{X} \geq \frac{\mu
|U|}{1+\mu}\right] \leq \frac{1}{\mu}$.
Hence, 
$$\Pr\left[X \geq \frac{\mu
|U|}{1+\mu}\right] > 1 - \frac{1}{\mu}.$$
Note that $0 < 1 - \frac{1}{\mu} \leq 1$ for $\mu > 1$.
Thus, with probability at least $1 - \frac{1}{\mu}$, at least
$\frac{\mu}{1+\mu} |U|$ vertices are removed from $U$ (i.e.,
permanently colored) in each round.

For some constants $c_1, c_2$, let now $c_1 \cdot \log n + c_2$ be the number
of iterations of SIM-COL performed. We want to derive a number of iterations
where we deactivated at least $\frac{\mu}{1+\mu} |U|$ vertices. This can be
expressed as a random variable $Q = \sum_{i=1}^{c_1\log n + c_2}{Q_i}$, where
$Q_i$ is the Bernoulli random variable that indicates if in iteration $i$
sufficiently many vertices were deactivated. This happens independently, as
shown above, with probability $Pr[Q_i = 1] \geq \frac{\mu-1}{\mu}$. Now we
apply a Chernoff bound to $X$ and get:

\ssmall
\begin{align}
%
% & Pr\left[Q \geq (1-\eta)\frac{1}{2}(c_1 \log n + c_2)\right] \leq e^{\frac{-\eta (c_1 \log n + c_2)}{6}} \\
%
& Pr\left[Q < (1+\eta)\frac{\mu-1}{\mu}(c_1 \log n + c_2)\right] \geq 1 - \left(n^{\frac{-\eta^2 c_1(\mu-1)}{3 \mu}}\cdot e^{\frac{-\eta^2 c_2(\mu-1)}{3 \mu}}\right)
\end{align}
\normalsize

Thus, if we choose $0 < \eta < 1$ and $c_1$ appropriately, it follows that we
perform $O(\log n)$ iterations w.h.p..
\end{proof}
\fi

To bound the work of SIM-COL, we can observe that similarly to the number of
vertices, the number of edges incident to at least one active vertex also
decreases by a constant factor in each iteration with high probability. The
work in an iteration is bounded by this number of edges and every iteration has
depth $O(\Delta)$ (in the CREW setting) and $O(\log
\Delta)$ (in the CRCW setting). Hence, we conclude:

% \marginpar{\vspace{0em}\colorbox{green}{\textbf{ALL:}}\\\colorbox{green}{{better}}\\\colorbox{green}{{(tighter)}}\\\colorbox{green}{{bounds,}}\\\colorbox{green}{{check}}\\\colorbox{green}{{rebuttal}}\\\colorbox{green}{{for more}}\\\colorbox{green}{{details}}}

\begin{lma} \label{lma:sim-col_runtime}
SIM-COL takes $O(\Delta \log n )$ depth (in the CREW
setting) or $O(\log \Delta \log n )$ depth (in the CRCW
setting), and it has $O(n+m)$ work w.h.p~in the CREW
setting.
\end{lma}

\iftr
\begin{proof}
%
% TODO : (Report) bound over the vertices in each iteration
%
Since the reduction on Line \ref{ln:sim-col_compare_end} can be implemented in
$O(\log{(deg(v))})$ depth for each $v \in U$ and since updating the bitmaps in
Part 3 takes $O(\Delta)$ (assuming CREW) and $O(\log \Delta)$ (assuming CRCW)
we get, together with Lemma~\ref{lma:sim-col_iter}, an overall depth of
$O(\Delta \log{n})$ (assuming CREW) and $O(\log \Delta \log{n})$ (assuming
CRCW) w.h.p..
% 
%Second, to analyze the work performed by the algorithm we will first look at the work
%performed in one outer-loop iteration. Let $U_i$ be the set $U$ in iteration $i$.
%Choosing colors for all vertices $v \in U_i$ (Line \ref{ln:sim-col_choose_color}) 
%takes $O(\abs{U_i})$ work in each iteration. The reduction for all $v \in U_i$
%(Line \ref{ln:sim-col_compare_end}) takes $O(\sum_{v \in U_i}{deg(v)})$ work.
%Updating the bitmaps (Line \ref{ln:sim-col_update_bmaps}) also takes
%$O(\sum_{v \in U_i}{deg(v)})$ work. Removing vertices that have been colored
%in the current iteration from $U_i$ can be done in $O(\abs{U_i})$ work.
%Therefore the work for one iteration is bound by $O(\sum_{v \in U_i}{deg(v)} + \abs{U_i})$.

%Now let the random variable $Y_i = \sum_{v \in U_i}{\overline{X}_v}$ be the number
%of vertices that are not removed in round $i$. Since we remove, in expectation,
%at least $\frac{\epsilon}{1+\epsilon}$ vertices in each iteration and since
%$Z_v$ scholastically dominates $X_v$ we get that $\mathbb{E}[Y_i] \geq
%\frac{1}{1+\epsilon}n'$. Now for a subsequent iteration $i+1$ we get $Y_{i+1}
%\geq \sum_{i = 1}^{\abs{U_i}}{\overline{Z}}$ and therefore we know that
%$\mathbb{E}[Y_{i+1}] \geq \frac{1}{1 + \epsilon}\mathbb{E}[Y_i]$.
%
%Thus, the sum of all vertices in U over all iterations is at most
%$\sum_{i=0}^{\infty}{\left(\frac{1}{1+\epsilon}\right)^i n} = O(n)$ in
%expectation. 
%
%As the total \#vertices in $U$ over all iterations is $O(n)$ in expectation, we
%get, as in ADG, work of $O(n+m)$ in expectation over all iterations.
%
\iftr
%
%As the total \#vertices in $U$ over all iterations is $O(n)$ in
%expectation, we get, similarly as in the case of ADG-AVG, work of $O(n+m)$ in
%expectation over all iterations.
%
\fi

Second, to analyze the work performed by the algorithm we will first look at the work
performed in one outer-loop iteration. Let $U_i$ be the set $U$ in iteration $i$. Let $n_i = \abs{U_i}$ be the number of active vertices and $m_i$ be the number of active edges at the start of iteration $i$.
Choosing colors for all vertices $v \in U_i$ (Line \ref{ln:sim-col_choose_color}) 
takes $O(\abs{U_i})$ work in each iteration. The reduction for all $v \in U_i$
(Line \ref{ln:sim-col_compare_end}) takes $O(\sum_{v \in U_i}{deg(v)})$ work.
Updating the bitmaps (Line \ref{ln:sim-col_update_bmaps}) also takes
$O(\sum_{v \in U_i}{deg(v)})$ work. Removing vertices that have been colored
in the current iteration from $U_i$ can be done in $O(\abs{U_i})$ work.
Therefore the work for one iteration is bound by $O(\sum_{v \in U_i}{deg(v)} + \abs{U_i}) = O(m_i + n_i)$.

To argue further we want to make a case distinction over the number of active vertices.\\

\textbf{\ul{Case 1}} Iterations where $n_i \in \Omega(\log n)$ but $n_i \not \in \Theta(\log n)$:
We can model the number of vertices that get removed in each iteration as $X =
\sum_{v \in U_i}{X_v}$. This random variable, as seen from Lemma
\ref{clm:bin_dominance}, stochastically dominates $Y = \sum_{i}^{n_i}{Z}$, i.e.,
for all valid $x$ we have $Pr[X \geq x] \geq Pr[Y \geq x]$.  Note that
$\mathbb{E}[Y] = \frac{\mu}{1+\mu}n_i$ and that because of our assumption we
have $n_i > c_1 \log n + c_2 $ for any constants $c_1 > 0$ and $c_2$. As 
$Y$ is defined over independent random variables $Z$, we can apply a Chernoff
bound and, as $X$ stochastically dominates $Y$:

\begin{align*}
& Pr\left[Y \leq (1-\eta)\mathbb{E}[Y]\right] \leq e^{-\frac{\eta^2}{2} \mathbb{E}[Y] }  \\
& Pr\left[X > (1-\eta)\frac{\mu}{1+\mu}n_i\right] \geq 1 - \left( n^{-\frac{\eta^2 \mu c_1}
{3(1+\mu)}} \cdot e^{-\frac{\eta^2 \mu c_2}{3(1+\mu)}} \right)
\end{align*}

This implies, for appropriately chosen constants $0 < \eta < 1$ and $c_1$, that
in each iteration, where we have active vertices in order of $\log n$, at least
a constant fraction of them will be deactivated w.h.p.. Now, by
using a union bound over all the iterations, which are in $O(\log n)$ w.h.p. 
(Lemma~\ref{lma:sim-col_iter}), we can also conclude that w.h.p. we will always
deactivate a constant fraction of vertices. Further, we know that an edge
is deactivated with at least the same probability as one of its incident
vertices. Hence, w.h.p., we will also deactivate a constant
fraction of our active edges in each iteration.  Thus, the overall work
performed is in $O(n+m)$ w.h.p..\\

\textbf{\ul{Case 2}} Iterations where $n_i \in O(\log n)$: 
As there are at most $O(\log n)$ active vertices left, we can bound the work
in one iteration by $O(\log n + \log^2 n)$. With Lemma~\ref{lma:sim-col_iter},
we get overall work of $O(\log^3 n)$ w.h.p. for these iterations. \\

Therefore, we get $O(n + m)$ work w.h.p. over all iterations.

\ifall
(\emph{sketch})

We already know that the algorithm performs $O(\log n)$ outer loop iterations
w.h.p.

To analyze the algorithm complexity we will first look at the 3 parallel for
loops that ar located in the outer-loop body.

Loop 1 (Lines 6--8) can be implemented in $O(1)$ depth and $O(\abs{U})$ work.

Loop 2 (Lines 10--16) can be implemented in $O(1)$ depth and $O(\sum_{v \in
U}{deg(v)})$ work. This works, since we assume concurrent reads to run in
$O(1)$ and because all $B_v[u]$ are distinct memory locations.

Loop 3 (Lines 18--23) can be implemented in $O(\log \Delta)$ depth and
$O(\sum_{v \in U}{deg(v)})$ work, since the Reduce operation takes $O(\log
\Delta)$ depth and $O(deg(v))$ work for each $v \in U$.

Line 24 can be implemented in $O(1)$ depth and $O(\abs{U})$ work. 

Line 25 can be implemented in $O(1)$ depth and $O(\abs{R})$ work. 

Since the algorithm performs $O(\log n)$ w.h.p. we get a depth of  $O(\log n
\log \Delta)$ w.h.p.

Since we remove a constant fraction of vertices from $U$ w.h.p. in each round,
we get over all iterations a total work of $O(n+m)$ (Similar argument as for
ADG : the sum of the number of vertices of U over all rounds is in $O(n)$ )
\fi
\end{proof}
\fi 

Now, we turn our attention back to DEC-ADG. As DEC-ADG decomposes the edges of
the input graph into $O(\log n)$ disjoint subgraphs of maximum degree $O(d)$,
we have:

\begin{lma}
  %
  %%% CREW -> CRCW
 DEC-ADG takes $O(\log d \log^2 n )$ depth and
  $O(n+m)$ work w.h.p. in the CRCW setting.
\end{lma}

\iftr
\ifsq\enlargethispage{\baselineskip}\fi
\begin{proof}
Since the graph partitions $G(\ell)$ are induced by a partial $kd$-approximate
degeneracy order, we know that each $G(\ell)$ has a maximal degree of at most
$kd$. Therefore, we can conclude together with Lemma~\ref{lma:sim-col_runtime},
that SIM-COL (Algorithm~\ref{algo:sim-col}), if called on $G(\ell)$, induces
$O(\log d \log n)$ depth w.h.p. (assuming CRCW) and performs $O(\abs{R(\ell)} + \abs{E[R(\ell)]} )$
work in expectation (i.e., work proportional to the number of vertices and
edges in $G(\ell)$). Computing ADG takes $O(\log^2 n)$ depth and $O(n+m)$ work.
Updating bitmaps takes $O(\log d)$ depth (with a simple Reduce), since each vertex in $G(\ell)$ has at most
$kd$ neighbors in $Q$, and $O(\sum_{v \in R(\ell)}{deg(v)})$ work.

Now, since DEC-ADG (Algorithm~\ref{algo:adg-dec}) performs $O(\log n)$
iterations, we get an overall depth of $O(\log d \log^2 n)$ w.h.p.. More precisely,
this can be seen from Lemma \ref{lma:sim-col_iter}. Since we use a Chernoff
bound to bound the number of iterations performed by SIM-COL, we can also
guarantee that each of the $O(\log n)$ instances of SIM-COL will perform
$O(\log n)$ iterations w.h.p..

The work in one loop iteration can be bound by $O(\abs{R(\ell)} +
\abs{E[R(\ell)]} + \sum_{v \in R(\ell)}{deg(v)})$ in expectation (as seen
above).
Thus, since each vertex is in exactly one partition $R(\ell)$ and each edge
is in at most one (i.e., in at most one set $E[R(\ell)]$ of edges that belong to
a subgraph $G(\ell)$), we can conclude that Algorithm
\ref{algo:adg-dec} performs $O(n+m)$ work in expectation.
\end{proof}
\fi 

\textbf{\ul{Coloring Quality} }
Finally, we prove the coloring quality. 

\ifsq\enlargethispage{\baselineskip}\fi

\begin{clm} \label{clm:dec-adg-avg_colors}
DEC-ADG produces a $(2+\epsilon)d$ coloring of the graph for $0 < \epsilon \leq 8$.
\end{clm}

\begin{proof}
Since we use ADG to partition the graph into $(\rho,\mathcal{G})$ ($\mathcal{G}
\equiv [ G(1)\ ...\ G(\overline{\rho})]$) on Line $\ref{ln:adg-dec_part}$, we
know that $\rho$ is a partial $2(1+\epsilon/12)$-approximate degeneracy
ordering.  Therefore, we also know that each vertex $v \in G(i)$ has at most
$2(1+\epsilon/12)d$ neighbors in partitions $G(i')$ with $i' \geq i$.
This implies, that if we run SIM-COL on each partition $G(i)$, we will color
each partition with at most ${(1+\epsilon/4)2(1+\epsilon/12)d}$ colors. This is
by the design of SIM-COL, which delivers a $((1+\mu)\Delta)$-coloring for any
graph~$G$; in our case, we have $G \equiv G(i)$, $\Delta = 2(1+\epsilon/12)$,
and $\mu \equiv \epsilon/4$.  Now, ${(1+\epsilon/4)2(1+\epsilon/12)d}$ is
smaller or equal to ${(2+\epsilon)d}$ for $\epsilon \leq 8$, as
$(1+\epsilon/4)2(1+\epsilon/12) = \left(2 +\frac{16 \epsilon + \epsilon^2}{24} \right)  \leq
(2+\epsilon)$, for $\epsilon \leq 8$.
\end{proof}

Now we have already seen that the runtime bounds of DEC-ADG hold for $\mu > 1$
(Lemma \ref{lma:sim-col_runtime}). Since we defined $\mu$ to be $\epsilon/4$ they hold
for $\epsilon > 4$. Thus, together with Claim \ref{clm:dec-adg-avg_colors} we
can see that \emph{the algorithm attains it's runtime and color bounds for} $4 < \epsilon \leq 8$.

\subsection{Enhancing Existing Coloring Algorithms}
\label{sec:dec-adg-itr}

Finally, we illustrate that ADG does not only provide new provably efficient
algorithms, \emph{but also can be used to enhance existing ones}.
For this, we seamlessly replace our default SIM-COL routine with a recent
speculative coloring heuristic, ITR, by \c{C}ataly\"urek et
al.~\cite{ccatalyurek2012graph}. The result, \textbf{DEC-ADG-ITR}, is similar
to DEC-ADG, except that the used SIM-COL differs in
Line~\ref{ln:sim-col_choose_color} from the default
Algorithm~\ref{algo:sim-col}: colors are not picked randomly, but we choose the
smallest color not in $B_v$.

% with our ADG based low-degree decomposition $G(1), ..., G(\overline{\rho})$. 

%\ifsq\enlargethispage{\baselineskip}\fi

Using ADG enables deriving similar bounds on coloring
quality ($2(1+\epsilon)d + 1$) as before.
\iftr
Since each vertex in partition $G(\ell)$ can have at most $kd$ colored
neighbors when it gets colored, as shown before, we always choose the smallest
color form $\{1, ..., kd + 1\}$. 
\fi
%
% We evaluate this scheme empirically and compare it to all our other coloring
% schemes in \cref{ch:measurements}.
%
However, deriving good bounds on work and depth is hard because the lack of
randomization (when picking colors) prevents us from using techniques such as
Chernoff bounds.
\iftr
We were still able to provide new results.
\fi
\ifconf
We were still able to provide new results, detailed in a full report (also
cf.~Table~\ref{tab:gc-all}).
\fi

\iftr
Selecting colors can be done in $O(\widetilde{\Delta})$ depth
and $O(\Delta \cdot \abs{U_i})$ work, where $U_i$ is the set $U$ in
iteration~$i$ (of the \texttt{while} loop in SIM-COL) and $\widetilde{\Delta} =
\max_{v \in V(\ell)}{deg_{\ell}(v)}$. For the modified SIM-COL, we
get, since all other operations are equal, $O(\Delta I)$ depth ($I$ is the number of iterations
of ITR); the work is
{\small $$O\left(\sum_{i=1}^{I} \left[ \widetilde{\Delta} \cdot \abs{U_i} + \sum_{v
\in U_i}{deg(v)} \right] \right).$$}

Depth and work in DEC-ADG-ITER are, respectively

$$O(I d \log n)$$

and (as a simple sum over all iterations~$I$)

{\small
$$O\left(\sum_{\ell = 1}^{\overline{\rho}} \left(\sum_{i=1}^{I} \left[ \widetilde{\Delta} \cdot \abs{U_i} + \sum_{v \in U_i}{deg(v)} \right] + \sum_{v \in R(\ell)}{deg(v)} \right) \right).$$}
Note that \emph{these bounds are valid in the CREW setting}.
\fi

\ifall

\subsection{Using Concurrent Reads}
% TODO: (Report) Change ADG to CRCW
JP-ADG performs parallel asynchronous updates on a single counter for each
vertex to keep track of the number of uncolored predecessors
(Algorithm~\ref{jpcolor}, Line~\ref{ln:jp_color_join}). To prevent this
operation from being serialized, the algorithm needs some concurrent write
capabilities, therefore putting it in the CRCW setting. Then, DEC-ADG
(Algorithm~\ref{algo:adg-dec}) avoids concurrent writes to one memory cell by
ensuring that all such writes access \emph{distinct} memory locations.  In
Algorithm~\ref{algo:adg-dec} (Line~\ref{ln:adg-dec_update_bmaps}) each vertex
only updates its own bitmap $B_v$ in parallel and the updates over $v$'s
neighbors are serialized.  In Algorithm~\ref{algo:sim-col}, the same holds for
the bitmap updates on Line~\ref{ln:sim-col_update_bmaps}.

\fi

\subsection{Using Concurrent Reads } 

Similarly to past work~\cite{whasenplaugh2014ordering}, our algorithms rely
on concurrent writes. However, a small modification to ADG gives 
variants that only need concurrent \emph{reads} (a weaker
assumption that is more realistic in architectures that use caches).
Specifically, one can implement UPDATE from Algorithm~\ref{alg:adg-avg}
by iterating over all $v \in U$ in parallel, and for each~$v$,
appropriately modifying its degree: $D[v] = D[v] - Count(N_{U}(v) \cap R)$.
This makes both ADG and DEC-ADG rely only on concurrent
reads but it adds a small factor of~$d$ to work ($O(n d + m)$).

\subsection{Comparison to Other Coloring Algorithms}
\label{sec:comp-others}

We exhaustively compare JP-ADG and DEC-ADG to other algorithms in
Table~\ref{tab:gc-all}.
We consider: non-JP parallel schemes (Class~1), the best sequential greedy
schemes (Class~2), and parallel algorithms based on JP (Class~3).
We consider depth (time), work, used model, quality, generality, randomized
design, work-efficiency, and scalability.  We also use past empirical
analyses~\cite{whasenplaugh2014ordering, gebremedhin2000scalable,
jrallwright1995} and our results (\cref{ch:measurements}) to summarize
run-times and coloring qualities of algorithms used in practice, focusing on
modern real-world graphs as input. 
Details are in the caption of Table~\ref{tab:gc-all}.

\enlargeSQ

\ifrev\marginpar{\vspace{-34em}\colorbox{orange}{\textbf{R-3}}}\fi

\ifrev\marginpar{\vspace{-24em}\colorbox{orange}{\textbf{R-3}}}\fi

As explained in Section~\ref{sec:intro}, \emph{only our algorithms work for
arbitrary graphs, deliver strong bounds for depth and work and quality, and are
often competitive in practice}. Now, JP-SL may deliver a little higher quality
colorings, as it uses the \emph{exact} degeneracy ordering (although without
explicitly naming it) and its quality is (provably) $d + 1$.  However, JP-SL
comes with much lower performance.
On the other hand, most recent JP-LLF and JP-SLL only provide the
straightforward $\Delta+1$ bound for coloring quality. These two are however
inherently parallel as they depend linearly on $\log n$, while JP-ADG depends
on $\log^2 n$. Yet, \emph{JP-ADG has a different advantage in depth}:
JP-SLL and JP-LLF depend linearly on $\sqrt{m}$ or $\Delta$ while JP-ADG on
$d$. In today's graphs, $d$ is usually much (i.e., orders of magnitude) smaller
than $\sqrt{m}$ and $\Delta$~\cite{dhulipala2018theoretically}.
\ifconf
In the full report, we also provide a small lemma showing that
$d/2 \le \sqrt{m}$.
\fi
\iftr
To further investigate this, we now provide a small lemma.

\begin{lma} \label{lma:d-m}
For any $d$-degenerate graph, we have $\sqrt{m} \geq d / 2$.
\end{lma}

\begin{proof}
First we can see, from the definition of degeneracy, that $G$ is $d$-degenerate
(i.e. every induced subgraph has at least one vertex with degree $\leq d$), but
not $(d-1)$-degenerate.  Therefore, since $G$ is not $(d-1)$-degenerate, we will
have at least one subgraph $\tilde{G}(\tilde{V},\tilde{E})$ where each vertex
has a degree larger than $d-1$. This then also implies, that there are at least
$d$ vertices in $\tilde{G}$. Now we get for the edges of $\tilde{G}$ that
$\abs{\tilde{E}} = \frac{1}{2} \cdot \sum_{v \in \tilde{V}}{deg_{\tilde{V}}(v)}
\geq \frac{1}{2}\sum_{i = 1}^{d}{d} = \frac{d^2}{2} $. Thus,
$G$ has at least $d^2/2$ edges and therefore we get $\sqrt{m} \geq d / 2$
\end{proof}

Lemma~\ref{lma:d-m} and the fact that $d \leq \Delta$ illustrate that the
expected depth of JP-ADG is up to a logarithmic factor comparable to JP-R,
JP-LLF, and JP-SLL.
\fi
This further illustrates that our bounds on depth in JP-ADG offer an
interesting tradeoff compared to JP-LF and JP-LLF.

We finally observe that the design of ADG combined with the parametrization
using $\epsilon$ enables a \emph{tunable parallelism-quality tradeoff}.
When $\epsilon \to 0$, coloring quality in JP-ADG approaches $2d+1$, only 
$\approx$2$\times$ more than JP-SL.
On the other hand, for $\epsilon \to \infty$, $\rho_{\text{ADG}}$ becomes
irrelevant and the derived final ordering $\rho = \langle \rho_{\text{ADG}},
\rho_X \rangle$ converges to~$\rho_X$. Now, $X$ could be the
random order R but also the low-depth LF and LLF orders based on largest degrees.
This enables JP-ADG to increase parallelism \emph{tunably}, depending on
user's needs.

\ifall
We also provide an interesting observation: \emph{JP-ADG offers a formal
parallelism-quality tradeoff between JP-SL on one hand, and other orders such
as JP-LLF, JP-LF, JP-R} (Section~\ref{sec:intro} explains all these schemes).
Specifically, when $\epsilon \to 0$ in ADG, the ADG ordering approaches $2d$, 
being only a factor of 2 worse than
JP-SL, which offers provably highest coloring quality
of~$d+1$ because it is based on the exact degeneracy order SL
(cf.~Table~\ref{tab:gc-all}), at the cost of having worse depth, to the point
of being inherently sequential for some graphs~\cite{whasenplaugh2014ordering}.
On the other hand, for $\epsilon \to \infty$, $\rho_{\text{ADG}}$ becomes
irrelevant and the derived final ordering $\rho = \langle \rho_{\text{ADG}},
\rho_X \rangle$ converges to~$\rho_X$. Now, depending on $X$, this may be the
random order $\rho_X$ but also the recent orderings based on largest degrees:
LF and LLF.  These two are at the other end of the spectrum of the
parallelism-quality tradeoff: they are inherently parallel because they depend
linearly on $\log n$, but they only come with the naive $\Delta+1$ bound on
coloring quality (cf.~Table~\ref{tab:gc-all}).
We conclude that \emph{JP-ADG offers a tunable tradeoff between parallelism
and obtained quality of colorings}.
\fi

We also compare JP-ADG to works based on speculation and conflict
resolution~\cite{ccatalyurek2012graph, rokos2015fast, deveci2016parallel,
boman2005scalable}, based on an early scheme by
Gebremedhin~\cite{gebremedhin2000scalable}. A direct comparison is
difficult because these schemes do no offer detailed theoretical
investigations. 
\ifall
While their coloring quality can be bounded with $\Delta+1$
(similarly to JP schemes), it is difficult to derive conclusive bounds on
depth and work because the design of these schemes (1) is not based on the
work--depth analysis, and (2) it does not easily generalize to our theoretical
approach based on bounding the longest path in the DAG~$G_\rho$
(cf.~Section~\ref{ch:analysis_th}). Thus, we provide straightforward bounds of
$O(\Delta I)$ (depth) and $O(\Delta I p)$ (work), where $I$ is the number of
iterations and $p$ is the number of processors. 
\fi
Simple bounds on coloring quality, depth, and work are -- respectively -- $\Delta+1$,
$O(\Delta I)$, and $O(\Delta I P)$, where $I$ is
\#iterations and $P$ is \#processors. Here, we illustrate that using ADG in combination
with these works simplifies deriving better bounds
for such algorithms, as seen by the example of DEC-ADG-ITR, see~\cref{sec:algorithm_dec-adg}.

\enlargeSQ

\ifbsp
\input{bsp-algorithm}
\fi

\section{Optimizations and Implementation}
\label{sec:design-impl}

\iftr
  \ifsq\enlargethispage{\baselineskip}\fi
  We explored different design choices and optimizations for more performance.
  Some relevant design choices were already discussed in
  Section~\ref{ch:algorithm}; these were key choices for \emph{obtaining our
    bounds without additional asymptotic overheads}.
\fi
Here, the main driving question that we followed was \emph{how to maximize the
  practical performance of the proposed coloring algorithms}?
\iftr
  For brevity, we discuss optimizations by extending Algorithm~\ref{alg:adg-avg}
  and~\ref{jp}.  We also conduct their {theoretical analysis}.
\fi
\ifconf
  All details, theorems, and proofs are in the full report.
\fi

\iftr\subsection{Representation of $U$ and $R$ }\label{sec:design-reps}\fi

\ifconf\textbf{Representation of $U$ and $R$ }\fi
The first key optimization (ADG, Alg.~\ref{alg:adg-avg}) is to maintain set~$U$
(vertices still to be assigned the rank $\rho$) and sets~$R(\cdot)$ (vertices
removed from $U$, i.e., $R(i)$ contains vertices removed in iteration~$i$)
\emph{together in the same contiguous array} such that all elements in
$R(\cdot)$ \emph{\textbf{precede}} all elements in $U$.
In iteration~$i$, this gives an array $[R(1)\ ...\ R(i)\ \text{\texttt{index}}\
      U]$, where \texttt{index} points to the first element of $U$. Initially, in iteration~$0$,
\texttt{index} is 0.
\iftr
Then, in iteration~$i+1$, one extracts $R(i+1)$ from $U$ and substitutes
$[R(0)\ R(1)\ ...\ R(i)\ \text{\texttt{index}}\ U]$ with $[R(0)\ R(1)\ ...\
      R(i)\ \text{\texttt{index}}\ R(i+1)\ U\setminus R(i+1)]$.
Formally, we keep an invariant that, if any vertex~$v$ has a degree $\le
  (1+\epsilon)\widehat{\delta}$ ($D[v] \le (1+\epsilon)\widehat{\delta}$), then
$v \in R$ (i.e., $v \in R(0) \cup ... \cup R(i+1)$).
Prior to removing $R(i+1)$ from $U$, we partition $U$ into $[R(i+1)\ U
      \setminus R(i+1)]$ beforehand, which can be done in time $O(|U|)$. We do this
by iterating over~$U$, comparing the degree of each vertex to
$(1+\epsilon)\widehat{\delta}$, and placing this vertex either in $R(i+1)$ or in
$U \setminus R(i+1)$, depending on its degree.  Then, the actual \emph{removal}
of $R(i+1)$ from $U$ takes $O(1)$ time by simply moving the \texttt{index}
pointer by $|R(i+1)|$ positions ``to the right'', giving -- at the end of
iteration~$i+1$ -- a representation $[R(0)\ R(1)\ ...\ R(i)\ R(i+1)\
      \text{\texttt{index}}\ U]$.

%\ifsq\enlargethispage{\baselineskip}\fi

\iftr\subsection{Explicit Ordering in $R(\cdot)$}\label{sec:design-sort}\fi

\ifconf\textbf{Explicit Ordering in $R(\cdot)$}\fi
%
%\zur{New Text from here}
%
We advocate to sort each $R(i)$ by the increasing count of neighbors within the
vertex set $U \cup R(i)$ (i.e., for any two vertices $v,u \in R(i)$, we place
$v$ before $u$ in the array representation of $R(i)$ if and only if $u$ has
more neighbors in $U \cup R(i)$ than $v$). This gives an explicit
ordering within the set of vertices that are removed in the same iteration of ADG.
Further, this induces a total order on $\rho_{\text{ADG}}$ (i.e., random tie
breaking becomes not necessary). Our evaluation
indicates that such sorting often enhances the accuracy of the obtained
approximate degeneracy ordering, which in turn consistently improves the
ultimate coloring accuracy.
Sorting can be performed with linear time integer sort and the relative
neighbor count of each vertex can be obtained using array~$D$.  We also observe
that, for some graphs, this additional vertex sorting improves the overall
runtime.
%
%\zur{END New Text}
%
We additionally explored parallel integer sort schemes used
to maintain the above-described representation of $U \cup R$. We tried
different algorithms (radix sort~\cite{mcilroy1993engineering}, counting
sort~\cite{seward1954information}, and quicksort~\cite{hoare1962quicksort}).

\ifall
  and we discovered that counting sort gives most performance.
\fi

\iftr\subsection{Combining JP and ADG}\label{sec:design-combine}\fi

\ifconf\textbf{Combining JP and ADG }\fi
We observe that Part~1 of JP-ADG
(Lines~\ref{ln:jp_dag_start}--\ref{ln:jp_dag_end}, Algorithm~\ref{jp}), where
one derives predecessors and successors in a given ordering to construct the
DAG~$G_{\rho}$, can also be implemented as a part of UPDATE in ADG, in
Algorithm~\ref{alg:adg-avg}. To maintain the output semantics of ADG, we use an
auxiliary priority function $rank: V \to \mathbb{R}$ that simultaneously
specifies the needed DAG structure. For each $v \in V$, $rank[v]$ is defined as
the number of neighbors $u \in N(v)$, for which $\rho_{\text{ADG}}(u) > \rho_{\text{ADG}}(v)$ holds.
Analogously, $rank[v]$ resembles the number of neighbors of $v$ that were removed from $U$ after the removal of $v$ itself.
This optimization does not change the theoretical results.

\iftr\subsection{Degree Median Instead of Degree Average}
  \label{sec:design-med}\fi

\ifconf\textbf{Median }\fi
We also use degree median instead of degree average in ADG, to derive $R$: $R =
  \{u \in U \mid D[u] \le (1+\epsilon)\delta_{m}\}$, where $\delta_m$ is median
of degrees of vertices in~$U$. We developed variants of ADG (``ADG-M''), JP-ADG
(``JP-ADG-M''), and DEC-ADG (``DEC-ADG-M'') that use $\delta_m$, and analyzed
them extensively; they enable speedups for some graphs.

\iftr
  One advantage of ADG-M is that deriving median takes $O(1)$ time in a
  sorted array.
  %
  %\zur{Start New text}
  %
  However, the whole $U$ has to be sorted in each pass. We incorporate linear-time integer sorting, which was shown
  to be fast in the context of sorting vertex IDs~\cite{malicevic2017everything}
  %
  %\zur{End New Text}.
  %
  ADG-M only differs from
  Algorithm~\ref{alg:adg-avg} in that, (1) instead of~$\widehat{\delta}$, we
  select the median degree~$\delta_m$ (of the vertices in~$U$), and (2) $R$, the
  set of vertices removed from~$U$ in a given iteration, is now defined as the
  vertices $v \in U$, which satisfy $ deg_U(v) \leq \delta_m$. Additionally we
  limit $R$ to half the size of $U$ (+1 if $\abs{U}$ is odd). We refer to the
  priority function produced by ADG-M as $\rho_{\text{ADG-M}}$.
\fi

\ifall
  Now, In JP-ADG, we first call ADG to derive $\rho_{\text{ADG}}$. Then, we run
  JP using $\rho_{\text{ADG}}$. More precisely, we use $ \rho = \langle
    \rho_{\text{ADG}}, \rho_R \rangle$ where $\rho_R$ is a random priority function
  used to break ties of vertices that have the same rank in the
  $\rho_{\text{ADG}}$, i.e., vertices that were removed in the same iteration in
  Algorithm~\ref{alg:adg-avg}. Depending on using variants ADG-AVG or ADG-M,
  we call the resulting coloring algorithm either JP-ADG-AVG or JP-ADG-M,
  respectively.
  These two obtained algorithms are similar to past work based on JP in that they
  follow the same ``skeleton'' in which coloring of vertices is guided by the
  pre-computed order, on our case ADG-AVG or ADG-M.  However, as we prove in
  Section~\ref{ch:analysis_th}, \emph{thanks to incorporating ADG, both variants
    of JP-ADG come with substantially new theoretical guarantees}, contrarily to
  all past work: provable work-efficiency and low depth and
  high coloring quality.
\fi

\iftr\subsection{Push vs.~Pull}\fi

\ifconf\textbf{Push vs.~Pull }\fi
%
%\zur{Adapted Text from here. Deleted following: Computing~$\rho'$ } 
%
In JP-ADG, computing the updates to $D$ can be implemented either in the push
or the pull style (pushing updates to a shared state or pulling updates to a
private state)~\cite{besta2017push}.
%
%\zur{New Text from here} 
%
More precisely, one can either iterate through the vertices in $R(\cdot)$ and
decrement counters in $D$ accordingly for each neighbor in $U$. This is
\emph{pushing} as the accesses to $D$ modify the shared state. On the other
hand, in \emph{pulling}, one iterates through the remaining vertices in $U
  \setminus R(\cdot)$ and counts the number of neighbors in $R(\cdot)$ for each
vertex separately. This number is then subtracted from the relative counter in
$D$. Here, there are no concurrent write accesses to $D$.
%
%\zur{End New Text}
%
We analyzed both variants and found that, while pushing needs atomics, pulling
incurs more work. Both options ultimately give similar performance.
%
%\zur{New Text from here}
%
\ifall
  Whenever the ADG routine is \emph{not} to be used in combination with JP, the
  pulling variant should be selected as one does not have to iterate through
  $R(\cdot)$.
\fi

\ifsq\enlargethispage{\baselineskip}\fi

\iftr\subsection{Caching Sums of Degrees}\fi

\ifconf\textbf{Caching Sum of degrees }\fi
%
%\zur{New, Maciej can decide if we want to include this}
%
In each iteration, the average degree of vertices in $U$ is computed and stored
in $\widehat{\delta}$. Here, instead of computing $\widehat{\delta}$ in each
iteration by explicitly summing respective degrees, one can maintain the sum of
degrees in $U$, $\Sigma_U$ (in addition to~$\widehat{\delta}$), and update
$\Sigma_U$ accordingly by subtracting the number of edges in the cut
$(R(\cdot), U \setminus R(\cdot))$ in each iteration.
\iftr
We omit this enhancement from the listings for clarity; our evaluation
indicates that it slightly improves the runtime (by up to $\approx$1\%).
\fi

\iftr\subsection{Infrastructure Details}\fi

\ifconf\textbf{Infrastructure Details }\fi
We integrated our algorithms with both the GAP benchmark
suite~\cite{beamer2015gap} and with GBBS~\cite{dhulipala2018theoretically}, a
recent platform for testing parallel graph algorithms. We use
OpenMP~\cite{chandra2001parallel} for parallelization. We settle on GBBS.

\iftr\subsection{Detailed Algorithm Specifications}

  We now provide detailed specifications of our algorithms with the optimizations
  described earlier in this section. We refer to them as ADG-O and ADG-M-O. The
  former is our fundamental ADG algorithm (described in
  Section~\ref{ch:algorithm}) that uses the \emph{average} degree to select
  vertices for removal in a given iteration. The latter is the ADG variant that
  uses the \emph{median} degree instead of the average degree (described
  in~\cref{sec:design-med}). The respective listing (of ADG-O) is in
  Algorithm~\ref{alg:iadg-avg}. We omit the listing of ADG-M-O because it is
  identical to that of ADG-O.  The only difference is in the way the
  \texttt{PARTITION} subroutine works.  Specifically, it results in $U = [R \;\;
    \delta_m \;\; u_{|R|+2} \;\; \ldots \;\; u_{|U|}]$ such that $\forall u \in U:
    \;\; D[u] \leq D[\delta_m] \; \Leftrightarrow \; u \in R$.  $\delta_m$ is the
  median of the degrees of vertices in $U$ based on the increasing degree ordering.

\iftr
Note that in the \texttt{UPDATEandPRIORITIZE} subroutine, when we refer to $N(v)$,
this technically is $N_{U \cup R}(v)$. However, as enforcing the induced
neighbor set entails performance overheads, and as using $N(v)$ is not
incorrect, we use $N(v)$.
\fi

  \begin{lstlisting}[aboveskip=0.0em,abovecaptionskip=0.0em,belowskip=0.0em,float=h,label=alg:iadg-avg,
caption=\textmd{\textbf{ADG-O}, a variant of ADG from
Algorithm~\ref{alg:adg-avg} with optimizations described
in~\cref{sec:design-reps}, \cref{sec:design-sort},
and~\cref{sec:design-combine}.}]
/* Input: A graph $G(V,E)$.
* Output: A priority (ordering) function $\rho_{\text{ADG}}:V \to \mathbb{R}$,
          an auxiliary priority function $rank:V \to \mathbb{R}$. */

$D = [\ deg(v_1)\ deg(v_2)\ ...\ deg(v_n)\ ]$ //An array with vertex degrees
$U = V$ //$U$ is the induced subgraph used in each iteration
$\ell = 0$ //A counter variable keeping track of removed vertices
for all $v \in V$  do in parallel: |\label{ln:iadg-avg-init-start}|
   $\rho_{\text{ADG}}(v) = n$ //Initialize the priorities |\label{ln:iadg-avg-init-end}|

while $U \neq \emptyset $ do: |\label{ln:adg_main_start}|
  $\widehat{\delta}$ = $\frac{1}{\abs{U}}\sum_{v \in U}{D[v]}$ //Update the average degree of vertices in $U$|\label{ln:iadg-avg-avg}|

  //Partition $U$ resulting in an array $U=[R \;\; u_{|R|+1} \;\; \ldots \;\; u_{|U|}]$ |\label{ln:iadg-avg-part}|
  //s.t. $\forall u \in U: \;\; D[u] \leq (1+\epsilon)\widehat{\delta} \; \Leftrightarrow \; u \in R$; $R$ is a subarray of $U$,
  //it contains vertices assigned priority in a given iteration.
  //PARTITION runs in $O(|U|)$ time and can be implemented
  //as specified in |\cref{sec:design-reps}|. 
  PARTITION($U$, $(1+\epsilon)\widehat{\delta}$) 
  
  //Sort $R$ by increasing degrees based on $D$, see |\cref{sec:design-reps}| and |\cref{sec:design-sort}|. |\label{ln:iadg-avg-sort}|
  //As $D$ is updated every iteration (see below), $v$ is before $u$
  //in $R$ if and only if $u$ has more neighbors in $U \cup R$ than $v$.
  SORT($R$, $D$) 

  //Explicitly impose an ordering for processing the vertices: 
  for all $r_i \in [r_0 \;\; r_1 \;\; \ldots \;\; r_{|R|-1}] = R$ do in parallel: |\label{ln:iadg-avg-rank-start}|
    $\rho_{\text{ADG}}(r_i) = \ell + i$ |\label{ln:iadg-avg-rank-end}|

  //Remove selected low-degree vertices (that are in $R$)
  //by moving the pointer to $U$, denoted as $\&U$, by $|R|$ cells.
  $\&U = \&U + |R|$ 

  //Update $D$ to reflect removing $R$ from $U$ (definition below):
  UPDATEandPRIORITIZE($D$, $U$, $R$, $rank$) 

  //Update the counter by adding the count of removed vertices
  $\ell = \ell + |R|$

//Update $D$ to reflect removing vertices in $R$ from a set $U$:
UPDATEandPRIORITIZE($D$, $U$, $R$, $rank$):
  for all $v \in R$ do in parallel:
    $c=0$ //A counter to count neighbours with higher priority
    for $w \in N(v)$ do:
      if $\rho_{\text{ADG}}(w) > \rho_{\text{ADG}}(v)$:
        //$w$ comes after $v$ in the explicit processing order
        $DecrementAndFetch(D[w])$
        $c = c+1$
    $rank(v) = c$ 
\end{lstlisting}

  % \begin{lstlisting}[aboveskip=0em,abovecaptionskip=0em,belowskip=0em,float=h,label=alg:iadg-med,
  % caption=\textmd{\textbf{ADG-M-O}, a variant of ADG-M (described
  % in~\cref{sec:design-med}) with optimizations described
  % in~\cref{sec:design-reps}, \cref{sec:design-sort},
  % and~\cref{sec:design-combine}.}]
  % //The listing of ADG-M-O is identical to that of ADG-O.
  % //The only difference is in the way PARTITION works.
  % //Specifically, it results in $U=[R \;\; \delta_m \;\; u_{|R|+2} \;\; \ldots \;\; u_{|U|}]$
  % //s.t. $\forall u \in U: \;\; D[u] <= D[\delta_m] \; \Leftrightarrow \; u \in R$.
  % //$Med$ is the median of $U$ based on the increasing degree ordering
  % \end{lstlisting}

  % \zur{I put my newest design for JP-ADG and JP-MED here. We have to move it to the right place after discussion}
  %
\fi

\ifall
  \begin{lstlisting}[aboveskip=-0.5em,abovecaptionskip=0.1em,belowskip=-0.5em,float=h,label=approxd2,
    caption=\textmd{ADG-AVG}]
    /* Input: A graph $G(V,E)$.
     * Output: A priority (ordering) function $\rho:V \to \mathbb{R}$. */

    $D = [\ deg(v_1)\ deg(v_2)\ ...\ deg(v_n)\ ]$ //An array with vertex degrees
    $U = V$ //$U$ is the induced subgraph used in each iteration
    $\ell= 0$ //$\ell$ counts the number of excluded vertices
    //Initialize $\rho_{ADG}$ s.t. $\forall v \in V: \; \rho_{ADG}(v) = n$
    while $U \neq \emptyset $ do: |\label{ln:adg_main_start}|
      $\widehat{\delta}$ = $(1+\epsilon)\frac{1}{\abs{U}}\sum_{v \in U}{D[v]}$ //Update the average degree for vertices in $U$|\label{ln:adg-avg-deg}|
      //$R$ is a subarray of $U$ and contains vertices assigned priority in a given iteration:
      PARTITION($U$, $\widehat{\delta}$) // Partition $U$ resulting in $U=[R \;\; u_{|R|+1} \;\; \ldots \;\; u_{|U|}]$ s.t. $\forall u \in U: \;\; D[u] <= \widehat{\delta} \; \longleftrightarrow \; u \in R$
      SORT($R$, $D$) //Sort $R$ by increasing degrees based on $D$
      $Head(U) = Head(U)+ |R|$ //Remove selected low-degree vertices (that are in $R$) |\label{ln:remove_R}| by moving the head pointer forward
      for all $v \in R$ and $i=0, 1, \ldots$ do in parallel:
        $\rho_{ADG}(v) = \ell + i$ //Assign the index within the ordering to the vertex
      UPDATE($U$, $R$, $D$) //Update $D$ to reflect removing $R$ from $U$
      $\ell = \ell + |R|$ //Update counter

    //Update $D$ to reflect removing vertices in $R$ from a set $U$:
    UPDATE($U$, $R$, $D$):
      for all $v \in R$ do in parallel:
        for all $w \in N(v)$ do:
          if($\rho_{ADG}(w) = n$) // $w$ is in $U \setminus R$
            $D[w] = D[w]-1$ //Atomic update
\end{lstlisting}

  \begin{lstlisting}[aboveskip=-0.5em,abovecaptionskip=0.1em,belowskip=-0.5em,float=h,label=approxd3,
  caption=\textmd{ADG-M}]
  /* Input: A graph $G(V,E)$.
   * Output: A priority (ordering) function $\rho:V \to \mathbb{R}$. */

  $D = [\ deg(v_1)\ deg(v_2)\ ...\ deg(v_n)\ ]$ //An array with vertex degrees
  $U = V$ //$U$ is the induced subgraph used in each iteration
  $\ell= 0$ //$\ell$ counts the number of excluded vertices
  //Initialize $\rho_{ADG}$ s.t. $\forall v \in V: \; \rho_{ADG}(v) = n$
  while $U \neq \emptyset $ do: |\label{ln:adg_main_start}|
    //$R$ is a subarray of $U$ and contains vertices assigned priority in a given iteration:
    //$Med$ is the median of $U$ based on a increasing degree ordering
    PARTITION_MEDIAN($U$) // Partition $U$ resulting in $U=[R \;\; Med \;\; u_{|R|+2} \;\; \ldots \;\; u_{|U|}]$ s.t. $\forall u \in U: \;\; D[u] <= D[Med] \; \longleftrightarrow \; u \in R$
    SORT($R$, $D$) //Sort $R$ by increasing degrees based on $D$
    $Head(U) = Head(U)+ |R|$ //Remove selected low-degree vertices (that are in $R$) |\label{ln:remove_R}| by moving the head pointer forward
    for all $v \in R$ and $i=0, 1, \ldots$ do in parallel:
      $\rho_{ADG}(v) = \ell + i$ //Assign the index within the ordering to the vertex
    UPDATE($U$, $R$, $D$) //Update $D$ to reflect removing $R$ from $U$
    $\ell = \ell + |R|$ //Update counter

  //Update $D$ to reflect removing vertices in $R$ from a set $U$:
  UPDATE($U$, $R$, $D$):
    for all $v \in R$ do in parallel:
      for all $w \in N(v)$ do:
        if($\rho_{ADG}(w) = n$) // $w$ is in $U \setminus R$
          $D[w] = D[w]-1$ //Atomic update
\end{lstlisting}
\fi

\iftr

  \subsection{Theoretical Analysis}

  We also provide theoretical analysis of the routine for deriving the
  approximate degeneracy ordering using median degrees (ADG-M). Moreover, we
  analyze coloring algorithms based on ADG-M (JP-ADG-M and DEC-ADG-M). Finally,
  we also analyze the coloring algorithms enhanced with the optimizations
  described earlier in this section (ADG-O, ADG-M-O, JP-ADG-O, JP-ADG-M-O).
  The only change in bounds is in schemes based on median degrees, which increase
  the ensured coloring counts by 2$\times$. Importantly, \emph{all other bounds
    for coloring algorithms maintain their advantageous properties described in
    Sections~\ref{ch:algorithm}--\ref{sec:col-algs}};

  All proofs are very similar to the earlier proofs for the ADG, JP-ADG, and
  DEC-ADG algorithms. Thus, we only provide sketches, underlying any differences
  from previous analyses.

  \subsubsection{Analysis of ADG-M}

  We first investigate ADG-M.

  \begin{lma} \label{lma:adg-med_runtime}
    For a constant $\epsilon > 0$, ADG-M can be implemented such that it has
    $O(\log n)$ iterations, $O(\log^{2}{n})$ depth, and $O(n + m)$
    work in the CRCW setting.
  \end{lma}

  \begin{proof}
    Since the algorithm removes at least half the vertices of $U$ in each iteration
    ($\pm$1 vertex) we can conclude, as for the default version of ADG, that ADG-M
    performs $O(\log n)$ iterations in the worst case.
    As in each iteration all operations can run in $O(\log{n})$ depth, ADG-M's
    depth is $O(\log^{2}{n})$.

    Let $k$ be the number of performed iterations, let $U_i$ be the set $U$ in
    iteration $i$. To calculate work performed by ADG-M, we first consider the
    work in one iteration . Finding the median takes $O(\abs{U_i})$ work.
    Since all other operations are the same as for ADG, we get again
    $O(\left(\sum_{v \in U_i}{deg(v)} \right) + \abs{U_i})$ work in one iteration.
    As we also remove a constant number of vertices in each iteration from $U$,
    we conclude, as with ADG, that ADG-M performs $O(n+m)$ work over all
    iterations.
  \end{proof}

  \begin{lma} \label{lma:adg-med_correctness}
    ADG-M computes a partial 4-approximate degeneracy ordering of~$G$.
  \end{lma}

  \begin{proof}
    In each step $\ell$ of ADG-M, only less than then half of the vertices
    in the induced subgraph $G[U]$ can have a degree that is strictly larger than
    $2\widehat{\delta}_\ell$, where $\widehat{\delta}_\ell$ is the average degree
    of vertices in this subgraph. Now, as we always remove the vertices with the
    lowest degree and as we remove half of the vertices $(\pm 1)$ from $U$, all vertices
    that get removed have a degree of at most $2\widehat{\delta}_\ell$.
    From Lemma~\ref{lma:adeg}, we know that $\widehat{\delta} \leq 2d$. Thus, each
    vertex has a degree of at most $4d$ in the subgraph~$G[U]$ (in the current
    step). Consequently, each vertex has at most $4d$ neighbors that are ranked equal or higher.
  \end{proof}

  \subsubsection{Analysis of JP-ADG-M}

  We now derive the bounds for JP-ADG-M.  We proceed similarly as for JP-ADG.
  The partial approximate degeneracy ordering delivered by ADG-M is referred
  to as $\rho_{\text{ADG-M}}$.

  \begin{crl} \label{crl:adg-med_color}
    JP-ADG-M colors a graph with at most $4d + 1$ colors
    using the priority function $\rho = \langle \rho_{\text{ADG-M}}, \rho_{\text{R}} \rangle$,
  \end{crl}

  \begin{proof}
    Since $\langle \rho_{\text{ADG-M}}, \rho_R \rangle$ is a $4$-approximate
    degeneracy ordering, this result follows from Lemma \ref{lma:colorbound} and
    Lemma \ref{lma:adg-med_correctness}.
  \end{proof}

  \begin{crl}
    JP-ADG-M colors a graph~$G$ with degeneracy $d$ in expected depth $O(\log^2{n}
      +\log{\Delta} \cdot ( d\log{n} + \frac{\log{d}\log^2{n}}{\log{\log n}} ))$ and $O(n+m)$
    work in the CRCW setting.
  \end{crl}

  \begin{proof}
    Since $\rho_{\text{ADG-M}}$ is a partial $4$-approximate degeneracy ordering
    (Lemma \ref{lma:adg-med_correctness}) and since ADG-M performs at most $O(\log
      n)$ iterations (Lemma~\ref{lma:adg-med_runtime}), the depth follows from
    Lemma~\ref{lma:longestpath}, Lemma~\ref{lma:adg-med_runtime} and past
    work~\cite{whasenplaugh2014ordering}, which shows that JP runs in $O(\log{n} +
      \log{\Delta} \cdot \abs{\mathcal{P}} )$ depth. Since both JP and ADG-M perform
    $O(n+m)$ work, the same holds for JP-ADG-M.
  \end{proof}

  \subsubsection{Analysis of DEC-ADG-M}

  The analysis of DEC-ADG-M is analogous to that of DEC-ADG, with the difference
  in that we use the ADG-M routine instead of ADG. The work and depth bounds
  remain the same as in DEC-ADG ($O(n+m)$ w.h.p. and $O(d \log^2 n)$), and the
  coloring bound -- similarly to JP-ADG-M -- increases from $(2+\epsilon)d$ to
  $(4+\epsilon)d$.

\fi

\ifall
  \begin{clm}
    ADG-DEC-MED produces a $\ceil{(4+\epsilon)d}$ coloring of the
    graph for $ \epsilon \geq  0$.
  \end{clm}

  \begin{proof}
    The proof works analogous to Claim \ref{clm:dec-adg-avg_colors}, but since
    ADG-M produces a partial $4$-approximate degeneracy ordering we color
    each partition with at most $\ceil{(1+\epsilon/4)4d} = \ceil{(4+\epsilon)d}$ colors.
  \end{proof}

\fi

\iftr

  \subsubsection{Impact of Optimizations}

  We also investigate the impact of optimizations described
  in~\cref{sec:design-reps}, \cref{sec:design-sort},
  and~\cref{sec:design-combine} on the bounds of the associated routines (ADG-O,
  ADG-M-O, JP-ADG-O, JP-ADG-M-O).  The key part is to illustrate that (1) moving
  ``Part~1'' of JP-ADG into ADG, and (2) the modified combined representation for
  $R$ and $U$ increase neither depth or work. This can be shown
  straightforwardly.

\fi

\ifall
  \begin{lma} \label{lma:iadg-avg_correctness}
    IADG-AVG produces a priority function $\rho'_{\text{ADGA}}$, where for all $v
      \in V$ holds $\rho'_{\text{ADGA}}(v) = \abs{\{ u \in N(v) \mid \rho(u) >
        \rho(v) \}}$ for a $2(1+\epsilon)$-approximate degeneracy ordering $\rho$
  \end{lma}

  \begin{proof}
    First we want to define a priority function $\rho_{ADG}$ similar as in
    Algorithm \ref{alg:adg-avg} (ADG-AVG). Let $R_i$ be the set $R$ in IADG-AVG for
    each iteration $i=1..k$. Since these sets $R_i$ are exactly equivalent to the
    sets $R$ (Line \ref{ln:R_def}, Algorithm \ref{alg:adg-avg}) in each iteration
    of ADG-AVG, we can see with Lemma \ref{lma:adg-avg_correctness} that
    $\rho_{ADG}$ is a partial $2(1+\epsilon)$-approximate degeneracy ordering.  Now
    let $rank_i$ be the ranking function as defined on Lines
    \ref{ln:iadg-avg-rank-start}--\ref{ln:iadg-avg-rank-end} for each iteration
    $i=1..k$. With this we can define for all $v \in V$ $\rho(v)$ as $ \langle r,
      rank_{\rho_{r}} \rangle$, with $r =\rho_{ADG}(v)$ and conclude that $\rho$ is a
    $2(1+\epsilon)$-approximate degeneracy ordering.
    Now we look at an arbitrary vertex $v \in V$. By the design of IADG-AVG this
    vertex is present in exactly one set $R_i$. Let $U_i$ be the set of vertices as
    in IADG-AVG for each this iteration i.  We want to count the number of vertices
    that get assigned to $\rho'_{\text{ADGA}}(v)$ in UPDATEandPRIORITIZE.  First we
    see that UPDATEandPRIORITIZE checks all neighbors of v. Now for neighbors $u$
    not in $U_i$ we always have $rho(u) < rho(v)$. Since rank[u] is set to $-1$ for
    all vertices that have already been removed from $U$ no such vertices $u$ are
    counted. For neighbors in $u \in U_i$ $rho(u) > rho(v)$ holds iff $rank_i(u) >
      rank_i(v)$. By definition only such vertices are counted.
    Therefore $\rho'_{\text{ADGA}}(v)$ is equal to the number of vertices in
    $N(v)$, which are ranked higher in $\rho$ than $v$ and thus we can conclude the
    proof.
    \armon{TODO: make this clearer}
  \end{proof}

  \begin{crl}
    Algorithm \ref{alg:iadg-avg} performs $O(\log n)$ iterations, has
    $O(\log^2{n})$ depth and performs $O(n+m)$ work in the CRCW setting.
  \end{crl}

  \begin{proof}
    First we can conclude as with ADG-AVG (See proof of Lemma
    \ref{lma:adg-avg_runtime}) that we perform at most $O(\log n)$ outer loop
    iterations, since we remove a constant fraction of vertices from $U$ in each
    iteration.

    Let $U_i$ denotes the set $U$ in iteration $i$, $R_i$ the set $R$ in iteration
    $i$ and $k$ the number of iterations.

    PARTITION (Line \ref{ln:iadg-avg-part}) can be implemented in $O(\log
      \abs{U_i})$ depth and $O(\abs{U_i})$ \cite{blelloch1996programming}. As
    discussed before (See \cref{sec:algorithm_adg} and
    \cref{sec:algorithm_adg-med}) calculating the average degree (Line
    \ref{ln:iadg-avg-avg}) and SORT (Line \ref{ln:iadg-avg-sort}) can also be
    implemented with $O(\log{\abs{U_i}})$ depth and $O(\abs{U_i})$ in each
    iteration. The UPDATEandPRIORITIZE function has depth $O(1)$ (since we perform
    concurrent writes) and performs $O(sum_{v \in R_i}{deg(v)})$ work.  Therefore
    the overall depth amounts to $O(\log^2{n})$ and the work performed in one
    iteration can be bound by $O(\left(sum_{v \in R_i}{deg(v)} \right) +
      \abs{U_i})$.

    Similar as for ADG-AVG (See proof of Lemma \ref{lma:adg-avg_runtime}) since we
    remove a constant fraction of vertices from $U$ in each iteration we get that
    $\sum_{i=1}^{k}{\abs{U_i}} \in O(n)$. Now since each vertex of $V$ is in only
    one $R_i$ we get $\sum_{i=1}^{k}{\sum_{v \in R_i}{deg(v)}} = \sum_{v \in
        V}{deg(v)} = 2m$ Therefore we can also conclude that we perform $O(n+m)$ work
    over all iterations.
  \end{proof}

  \subsubsection{Analysis of ADG-M-O}

  \begin{lma} \label{lma:iadg-med_correctness}
    IADG-M produces a priority function $\rho'_{\text{ADG-M}}$, where for all $v
      \in V$ holds $\rho'_{\text{ADG-M}}(v) = \abs{\{ u \in N(v) \mid \rho(u) >
        \rho(v) \}}$ for a $4$-approximate degeneracy ordering $\rho$
  \end{lma}

  \begin{proof}
    The proof works the same as for IADG-AVG, but since we calculate the median we
    use Lemma \ref{lma:adg-med_correctness} instead of Lemma
    \ref{lma:adg-avg_correctness}.
  \end{proof}

  \begin{crl}
    Algorithm \ref{alg:iadg-med} performs $O(\log n)$ iterations, has
    $O(\log^2{n})$ depth and performs $O(n+m)$ work.
  \end{crl}

  \begin{proof}
    Again we can conclude as with ADG-M ( Lemma \ref{lma:adg-med_runtime}) that
    we perform at most $O(\log n)$ outer loop iterations.

    Since all operations of IADG-M are the same as for IADG-AVG (other than the
    absence of the average degree computation), we can conclude that we also get
    $O(\log^2 n)$ depth and $O(n+m)$ overall work.
  \end{proof}

  \begin{crl}
    \m{?} JP-ADG-AVG colors a graph with at most $2(1+\epsilon)d$ colors using the
    priority function $\rho_{\text{\text{ADGA}}'}$
  \end{crl}

  \begin{proof}
    Since $\rho_{\text{\text{ADGA}}'}$ is defined by a $4$-approximate degeneracy
    ordering (Lemma \ref{lma:iadg-med_correctness}), this result follows from Lemma
    \ref{lma:colorbound}.
  \end{proof}

  \begin{crl}
    JP-ADG-AVG colors a graph with at most $4d + 1$ colors using the priority
    function $\rho_{\text{ADG-M}'}$
  \end{crl}

  \begin{proof}
    Since $\rho_{\text{ADG-M}'}$ is defined by a $4$-approximate degeneracy ordering
    (Lemma \ref{lma:iadg-med_correctness}), this result follows from Lemma
    \ref{lma:colorbound}.
  \end{proof}

  If we in addition rank all vertices in $R$ (Algorithm \ref{alg:iadg-avg}, Lines
  \ref{ln:iadg-avg-rank-start}--\ref{ln:iadg-avg-rank-end}\ref{alg:iadg-med},
  Lines\ref{ln:iadg-med-rank-start}--\ref{ln:iadg-med-rank-end}) randomly, we can
  see that the order $\rho$ from Lemmas \ref{lma:iadg-avg_correctness} and
  \ref{lma:iadg-med_correctness}  is equivalent to $<\rho_{\text{ADGA}}, \rho_R>$
  and $<\rho_{\text{ADG-M}}, \rho_R>$ respectively. This can be seen from the fact
  that the vertices that we remove in each iteration of both IADG-AVG/MED and
  ADG-AVG/MED are the same. Therefore we can get the same asymptotic JP bounds
  for these modified versions (IADG-AVG'/IADG-M') as for ADG-AVG and ADG-M

  \begin{crl}
    JP-IADG-AVG' colors a graph~$G$ with degeneracy $d$ in expected depth
    $O(\log^2{n} +\log{\Delta} \cdot ( d\log{n} + \frac{\log{d}\log^2{n}}{\log{\log
          n}} ))$ and $O(n+m)$ work in the CRCW setting.
  \end{crl}

  \begin{proof}
    Since $\rho_{\text{\text{ADGA}}'}$ is defined by a partial
    $2(1+\epsilon)$-approximate degeneracy ordering and since IADG-AVG' performs at
    most $O(\log n)$ iterations, the depth follows from
    Theorem~\ref{thm:longestpath}, Lemma~\ref{lma:adg-med_runtime} and past
    work~\cite{whasenplaugh2014ordering}, which shows that JP runs in $O(\log{n} +
      \log{\Delta} \cdot \abs{\mathcal{P}} $) depth.  Since both JP and ADG-AVG'
    perform $O(n+m)$ work, the same holds for JP-ADG-AVG'.
  \end{proof}

  \begin{crl}
    JP-IADG-M' colors a graph~$G$ with degeneracy $d$ in expected depth
    $O(\log^2{n} +\log{\Delta} \cdot ( d\log{n} + \frac{\log{d}\log^2{n}}{\log{\log
          n}} ))$ and $O(n+m)$ work in the CRCW setting.
  \end{crl}

  \begin{proof}
    Since $\rho_{\text{ADG-M}'}$ is defined by a partial $4$-approximate degeneracy
    ordering and since IADG-M' performs at most $O(\log n)$ iterations, the depth
    follows from Theorem~\ref{thm:longestpath}, Lemma~\ref{lma:adg-med_runtime} and
    past work~\cite{whasenplaugh2014ordering}, which shows that JP runs in
    $O(\log{n} + \log{\Delta} \cdot \abs{\mathcal{P}} $) depth.  Since both JP and
    ADG-M' perform $O(n+m)$ work, the same holds for JP-ADG-M'.
  \end{proof}

\fi

\ifall
\input{implementation.tex}
\fi
\begin{figure*}[hbtp]
%\vspace{-2em}
\centering
\includegraphics[width=0.98\textwidth]{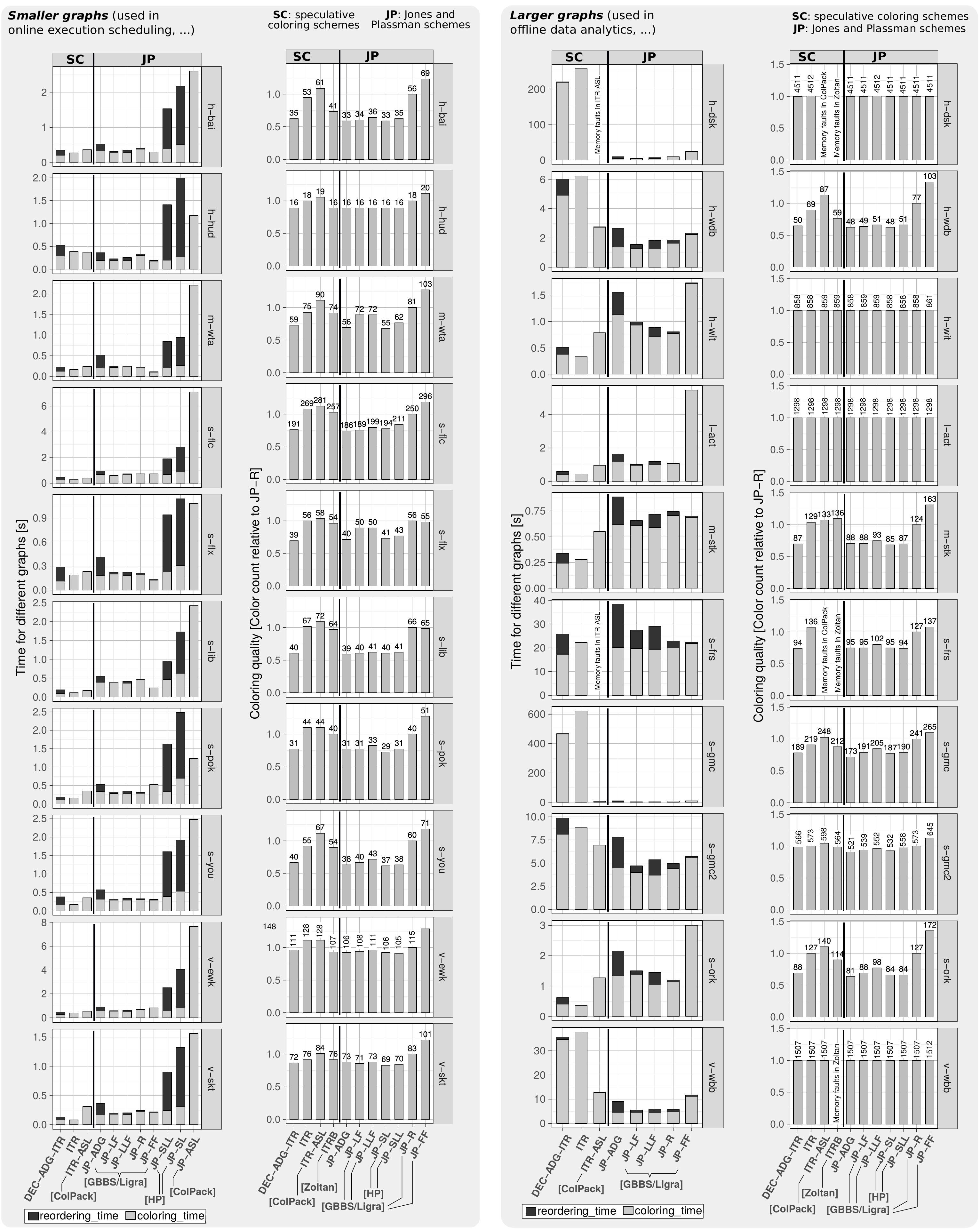}
\vspace{-0.5em}
\caption{\scriptsize\textmd{\textbf{Run-times (1st and 3rd columns) and coloring quality (2nd and 4th columns)}.} 
Two plots next to each other correspond to the same graph.
Graphs are representative (other results follow similar patterns).
Parametrization: 32 cores (all available), $\epsilon = 0.01$, sorting: Radix sort, direction-optimization: push,
JP-ADG variant based on average degrees $\widehat{\delta}$.
SL and SLL are excluded from run-times in the right column (for larger graphs) because they performed consistently worse than others.
We exclude DEC-ADG for similar reasons and because it is of mostly theoretical interest; instead, we focus on
DEC-ADG-ITR, which is based on core design ideas in DEC-ADG.
Numbers in bars for color counts are numbers of used colors. 
\textbf{``SC''}: results for the class of algorithms based on \textbf{speculative coloring}
(ITR, DEC-ADG-ITR).
\textbf{``JP''}: results for the class of algorithms based on the \textbf{Jones and Plassman}
approach (\textbf{color scheduling}, JP-*).
A vertical line in each plot helps to separate these two classes of algorithms.
DEC-ADG-ITR uses dynamic scheduling.
JP-ADG uses linear time sorting of $R$.
Any schemes that are always outperformed in a given respect (e.g., Zoltain in runtimes or ColPack in qualities) are excluded from the plots.
}
\vspace{-1em}
\label{fig:results}
\end{figure*}

\begin{figure*}[t]
%\vspaceSQ{-2em}
\centering
\includegraphics[width=1.0\textwidth]{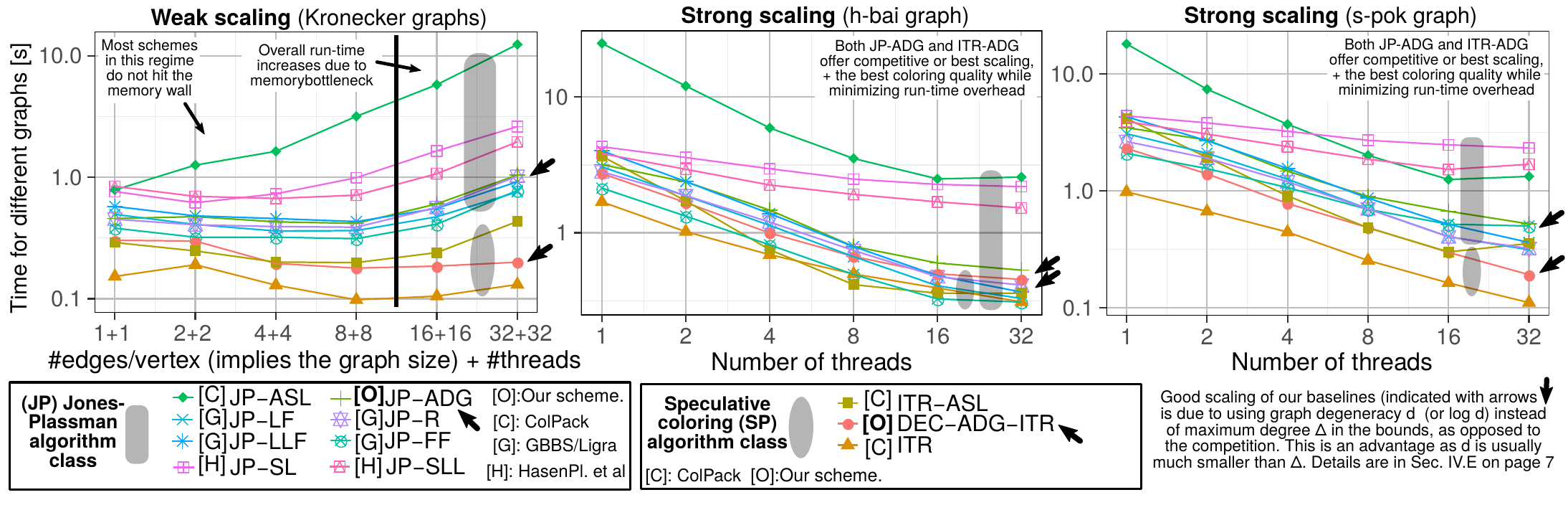}
\vspace{-1.5em}
\caption{\scriptsize{\textmd{\textbf{Weak and strong scaling}.} 
Graphs are representative (other results follow similar patterns).
Parametrization: \mbox{$\epsilon = 0.01$}, sorting: Radix sort, direction-optimization: push,
JP-ADG variant based on average degrees \mbox{$\widehat{\delta}$}.
DEC-ADG-ITR uses dynamic scheduling.
JP-ADG uses linear time sorting of $R$.
In weak scaling, we use $n=1$M vertices.
}}
\vspaceSQ{-1em}
\label{fig:results-scaling}
\end{figure*}

\section{Evaluation} \label{ch:measurements}

In evaluation, \emph{we found that empirical results follow theoretical
predictions}, already scrutinized in Table~\ref{tab:gc-all} and
Section~\ref{sec:col-algs}.
Thus, for brevity, we now summarize the most important observations. 
\ifall
All results are in the extended report (they follow similar performance
patterns), including more detailed discussions. 
\fi
A comprehensive comparison of run-times and coloring qualities of different
algorithms is in Table~\ref{tab:gc-all} (together with abbreviations of used
comparison baselines).

\ifall
We focus on describing results that \emph{could not be easily predicted from
theory}, either because they are hidden behind the big $O$ notation, or because
they are related to schemes that have no provable guarantees. 
\fi

\iftr
\subsection{Methodology, Architectures, Parameters}

We first provide the details of the evaluation methodology, used architectures,
and considered parameters, to facilitate interpretability and reproducibility
of experiments~\cite{hoefler2015scientific}.
\fi

\textbf{Used Architectures }
In the first place, we use
\textbf{Einstein}, an in-house Dell PowerEdge R910 server with an Intel Xeon X7550 CPUs @ 2.00GHz with
18MB L3 cache, 1TiB RAM, and 32 cores per CPU (grouped in four sockets).
\ifall
Euler has an HT-enabled Intel Xeon Gold 6150 CPUs @ 2.70GHz with 24.75MB L3
cache, 64 GiB RAM, and 36 cores per CPU (grouped in two sockets).
\fi
We also conducted experiments on
\textbf{Ault} (a CSCS server with Intel Xeon Gold 6140
CPU @ 2.30GHz, 768 GiB RAM, 18 cores, and 24.75MB L3) and
\textbf{Fulen} (a CSCS server with Intel
Skylake @ 2GHz, 1.8 TiB RAM, 52 cores, and 16MB L3).

\ifbsp
Finally, we also used XC50 compute nodes in the Piz Daint Cray supercomputer
(one such node comes with 12-core Intel Xeon E5-2690 HT-enabled CPU 64 GiB
RAM).
\fi

\textbf{Methodology }
We provide absolute runtimes when reporting speedups. In our measurements, we
exclude the first measured 1\% of performance data as warmup. We derive enough
data to obtain the mean and 95\% non-parametric confidence intervals. Data is
summarized with arithmetic means.

\textbf{Algorithms \& Comparison Baselines }
We focus on modern heuristics from Table~\ref{tab:gc-all}.
For each scheme, we always pick the most competitive implementation (i.e.,
fewest colors used and smallest performance overheads), selecting from
existing repositories, illustrated in Table~\ref{tab:baselines} (\textbf{ColPack}~\mbox{\cite{gebremedhin2013colpack,
gebremedhin2010colpack}}, \textbf{Zoltan}~\mbox{\cite{boman2012zoltan,
bozdaug2008framework, devine2009getting, rajamanickam2012parallel,
heroux2005overview}}, original code by \textbf{Hasenplaugh et
al.~(HP)}~\mbox{\cite{whasenplaugh2014ordering}}, \textbf{GBBS with
Ligra}~\mbox{\cite{dhulipala2018theoretically, shun2013ligra,
dhulipala2020graph}}), and our implementation.
Detailed parametrizations are in the reproducibility appendix.
For \textbf{Zoltan}, we tested different variants of ITRB and picked
the best configuration, consistently with past
work~\cite{boman2005scalable} (100 steps, synchronous, ``I'' coloring
order).
For \textbf{ColPack}, we also picked the best variant of ITR-ASL and JP-ASL
(``D1\_OMP\_GMMP\_SL'' and ``D1\_OMP\_MTJP\_SL'').
For other JP baselines, we pick the best performing ones out
of GBBS+Ligra and the original Hasenplaugh et al.'s code.

\begin{table}[h]
\centering
\ifsq
%\vspace{-1em}
%\scriptsize
\renewcommand{\arraystretch}{0.5}
\else
\renewcommand{\arraystretch}{1}
\fi
%
%\setlength{\tabcolsep}{0.8pt}
%\footnotesize
%\small
\sf
\begin{tabular}{ll}
\toprule
\textbf{GC Baseline} & \textbf{Available Codes} \\
\midrule
\makecell[l]{(MIS) Luby~\cite{luby1986simple}} &  {ColPack}  \\
\makecell[l]{Gebremedhin~\cite{gebremedhin2000scalable}} & {\mbox{ColPack}}  \\
\makecell[l]{Gebremedhin~\cite{gebremedhin2000scalable}} & {\mbox{ColPack}} \\
\makecell[l]{ITRB (Boman et al.~\cite{boman2005scalable}} & \makecell[l]{{\mbox{{Zoltan}}}} \\ 
\makecell[l]{ITR (Çatalyürek et al.~\cite{ccatalyurek2012graph}\\ and others~\cite{rokos2015fast, deveci2016parallel})} & \makecell[l]{{\mbox{{ColPack}}}} \\ 
 \makecell[l]{ITR-ASL (Patwary et al.~\cite{patwary2011new})} & \makecell[l]{{\mbox{{ColPack}}}}  \\
\midrule
Greedy-ID~\cite{coleman1983estimation} & {ColPack, GBBS} \\
Greedy-SD~\cite{brelaz1979new, whasenplaugh2014ordering} & {ColPack, GBBS} \\
\midrule
JP-FF~\cite{whasenplaugh2014ordering,welsh1967an} & {ColPack, GBBS} \\
JP-LF~\cite{whasenplaugh2014ordering} &{ColPack, GBBS, HP} \\
JP-SL~\cite{whasenplaugh2014ordering} & {ColPack, GBBS, HP}  \\
JP-R~\cite{jones1993parallel} & {ColPack, GBBS, HP} \\
JP-R~\cite{whasenplaugh2014ordering} & {ColPack, GBBS, HP}  \\
JP-LLF~\cite{whasenplaugh2014ordering} & {GBBS, HP} \\
JP-SLL~\cite{whasenplaugh2014ordering}  & {GBBS, HP} \\
JP-ASL~\cite{gebremedhin2013colpack,patwary2011new} & {ColPack} \\
\bottomrule
\end{tabular}
\ifsq
\vspace{-0.5em}
\fi
\caption{
\ifsq\ssmall\fi
\textbf{Existing implementations of parallel graph coloring algorithms, considered in this work.}
The existing repositories or frameworks are as follows:
\textbf{ColPack}~\mbox{\cite{gebremedhin2013colpack, gebremedhin2010colpack}},
\textbf{Zoltan}~\mbox{\cite{boman2012zoltan, bozdaug2008framework,
devine2009getting, rajamanickam2012parallel, heroux2005overview}}, original
code by \textbf{Hasenplaugh et
al.~(HP)}~\mbox{\cite{whasenplaugh2014ordering}}, and \textbf{GBBS with
Ligra}~\mbox{\cite{dhulipala2018theoretically, shun2013ligra,
dhulipala2020graph}}.
More details of the associated algorithms are in Table~\ref{tab:gc-all}.
}
\ifsq\vspace{-2em}\fi
\label{tab:baselines}
\end{table}

\textbf{Datasets }
We use real-world graphs from SNAP~\protect\cite{snapnets},
KONECT~\protect\cite{kunegis2013konect},
DIMACS~\protect\cite{demetrescu2009shortest}, and WebGraph
datasets~\protect\cite{boldi2004webgraph}; see Table~\ref{tab:graphs} for
details. We analyze synthetic power-law graphs (generated with the Kronecker
model~\cite{leskovec2010kronecker}). \emph{This gives a \textbf{large
evaluation space}; we only summarize selected findings}.

\begin{table}[b]
\vspaceSQ{-1.5em}
\centering
%\footnotesize
\scriptsize
\ssmall
%\sf
%
\setlength{\tabcolsep}{2pt}
\renewcommand{\arraystretch}{1}
\begin{tabular}{l}
\toprule
\makecell[l]{
\textbf{\ul{Friendships:}}
%
% & LiveJournal~\cite{snapnets} & 4,847,571 & 68,993,773 & & 16 & \textbf{285,730,264} & \\
% & LiveJournal (snapshot 1)~\cite{snapnets} & 3,997,962 & 34,681,189 & & 17 & \textbf{177,820,130} & \\
% & Friendster (snapshot 1)~\cite{snapnets} & 65,608,366 & 1,806,067,135 & & 32 & \textbf{4,173,724,142} & \\
 Friendster ({\textbf{s-frs}}, 64M, 2.1B),
 Orkut ({\textbf{s-ork}}, 3.1M, 117M),
 LiveJournal ({\textbf{s-ljn}}, 5.3M, \\ 49M),
 Flickr ({\textbf{s-flc}}, 2.3M, 33M),
 Pokec ({\textbf{s-pok}}, 1.6M, 30M),
 Libimseti.cz ({\textbf{s-lib}}, 220k, 17M),\\
 Catster/Dogster ({\textbf{s-cds}}, 623k, 15M),
 Youtube ({\textbf{s-you}}, 3.2M, 9.3M),
 Flixster ({\textbf{s-flx}}, 2.5M, 7.9M),
% 
% \midrule
%
% \multirow{3}{*}{\makecell[c]{Road\\Networks}} 
%
% & California & 1,965,206 & 2,766,607 & & 849 & \textbf{120,676} \\
% & Pennsylvania & 1,088,092 & 1,541,898 & & 786 & \textbf{67,150} \\
% & Texas & 1,379,917 & 1,921,660 & & 1054 & \textbf{82,869} \\
%
}\\
\midrule
\makecell[l]{\textbf{\ul{Hyperlink graphs:}}
\ifall\m{run!}
 Web Data Commons 2012 ({\textbf{h-wdc}}, 3.5B, 128B),
 EU domains (2015)\\({\textbf{h-deu}}, 1.07B, 91.7B),
 UK domains (2014) ({\textbf{h-duk}}, 787M, 47.6B),
 ClueWeb12 ({\textbf{h-clu}}, 978M, 42.5B),\\
\fi
 GSH domains ({\textbf{h-dgh}}, 988M, 33.8B),
 SK domains ({\textbf{h-dsk}}, 50M, 1.94B),\\
 IT domains ({\textbf{h-dit}}, 41M, 1.15B),
 Arabic domains ({\textbf{h-dar}}, 22M, 639M),\\
 Wikipedia/DBpedia (en) ({\textbf{h-wdb}}, 12M, 378M),
 Indochina domains ({\textbf{h-din}}, 7.4M, 194M),\\
 Wikipedia (en) ({\textbf{h-wen}}, 18M, 172M),
 Wikipedia (it) ({\textbf{h-wit}}, 1.8M, 91.5M),\\
 Hudong ({\textbf{h-hud}},  2.4M, 18.8M),
 Baidu ({\textbf{h-bai}}, 2.1M, 17.7M),
 DBpedia ({\textbf{h-dbp}}, 3.9M, 13.8M),
}\\
\midrule
\makecell[l]{\textbf{\ul{Communication:}}
 Twitter follows ({\textbf{m-twt}}, 52.5M, 1.96B),
 Stack Overflow interactions\\ ({\textbf{m-stk}}, 2.6M, 63.4M),
 Wikipedia talk (en) ({\textbf{m-wta}}, 2.39M, 5.M),
}\\
\midrule
\makecell[l]{\textbf{\ul{Collaborations:}}
 Actor collaboration ({\textbf{l-act}}, 2.1M, 228M),
 DBLP co-authorship ({\textbf{l-dbl}}, 1.82M,\\13.8M),
 Citation network (patents) ({\textbf{l-cit}}, 3.7M, 16.5M),
 Movie industry graph ({\textbf{l-acr}}, 500k, 1.5M)
}\\
\midrule
\makecell[l]{\textbf{\ul{Various:}}
 UK domains time-aware graph ({\textbf{v-euk}}, 133M, 5.5B),
 Webbase crawl\\({\textbf{v-wbb}}, 118M, 1.01B),
 Wikipedia evolution (de) ({\textbf{v-ewk}}, 2.1M, 43.2M),\\
 USA road network ({\textbf{v-usa}}, 23.9M, 58.3M),
 Internet topology (Skitter) ({\textbf{v-skt}}, 1.69M, 11M),
}\\
\bottomrule
\end{tabular}
\vspace{-0.5em}
\caption{\textmd{Considered real
graphs from established datasets~\protect\cite{snapnets,
kunegis2013konect, 
demetrescu2009shortest, 
boldi2004webgraph}.
{Graph are sorted by $m$ in each category.}
For each graph, we show its ``(\textbf{symbol used}, $n$, $m$)''.
\ifall
If there are multiple related graphs, for
example different snapshots of the .eu domain, we select the largest one. One
exception is the additional Italian Wikipedia snapshot selected due to its
interestingly high density.\fi
}}
%
%\vspace{-2em}
\label{tab:graphs}
\vspaceSQ{-1.5em}
\end{table}

\ifall
\begin{figure}
\vspace{-2em}
\centering
\includegraphics[width=1.0\columnwidth]{plots_runtimes_colors_paper.pdf}
\vspace{-1.5em}
\caption{\textmd{\textbf{Run-times (on the left) and coloring quality (on the right)}.} 
Two plots next to each other correspond to the same graph.
Graphs are representative (all other results follow similar patterns).
Parametrization: 32 cores (all available), $\epsilon = 0.01$, sorting: Radix sort, direction-optimization: push,
JP-ADG variant based on average degrees $\widehat{\delta}$.
Numbers in bars for color counts are colors used.
}
\vspace{-1em}
\label{fig:results}
\end{figure}
\fi

The results are in Figure~\ref{fig:results}. 
Following past analyses~\cite{whasenplaugh2014ordering}, we consider
\emph{separately} two distinctive families of algorithms: those based on
\textbf{speculative coloring (SC)}, and the ones with the
\textbf{Jones and Plassman} structure (\textbf{color scheduling}).
These two classes of algorithms -- especially for larger datasets -- are often
\emph{complementary}, i.e., whenever one class achieves lower performance, the
other thrives, and vice versa. This is especially visible for larger graphs,
such as h-dsk, h-wdb, or s-gmc.  The reason is that the structure of some
graphs (e.g., with dense clusters) entails many coloring conflicts which may
need many re-coloring attempts, giving long tail run-times. 

\tr{\subsection{Summary of Insights}}

\cnf{\textbf{\ul{Summary of Insights}}}
Our algorithms almost always offer superior coloring quality.  Only JP-SL,
JP-SLL {(HP)}, {and sometimes ITRB by Boman et
al.~\mbox{\cite{boman2005scalable}} (Zoltan)} use comparably few colors,
{but they are at least 1.5$\times$ and 2$\times$ slower, respectively.}
Simultaneously, run-times of our algorithms are comparable or marginally higher
than the competition (in the class of algorithms with speculative coloring) and
within at most 1.3-1.4$\times$ of the competition (in the class of JP
baselines).
Thus, \emph{we offer the best coloring quality at the smallest required
run-time overhead}.
Finally, our routines are the only ones with theoretical guarantees on work,
depth, and quality.
%
%   Our algorithms -- in almost all cases -- offer superior
%   coloring quality.  Only JP-SL and JP-SLL use comparably few colors,
%   but they are at least 1.5$\times$ slower.
%   %
%   Simultaneously, run-times of our algorithms are comparable or marginally
%   higher than the competition (in the class of algorithms based on speculative
%   coloring) and within at most 1.3-1.4$\times$ of the competition (in the class
%   of JP baselines).
%   %
%   Thus, \emph{we offer the best coloring quality at the
%   smallest required run-time overhead}. 
%   %
%   Finally, our routines are the only ones with theoretical guarantees on work,
%   depth, and quality.

\iftr\subsection{Analysis of Run-Times with Full Parallelism}\fi

\cnf{\textbf{\ul{Analysis of Run-Times with Full Parallelism}}}
We analyze \textbf{run-times} using all the available
cores. Whenever applicable, we show fractions due to
\textbf{reordering} (preprocessing, e.g., the ``ADG'' phase in JP-ADG) and the
actual \textbf{coloring} (e.g., the ``JP'' phase in JP-ADG).
JP-SL, JP-SLL {(HP)}, and JP-ASL {(ColPack)} are the slowest as 
they offer least parallelism. JP-LF, JP-LLF, and JP-R {(GBBS/Ligra)} are
very fast, as their depth is in $O(\log n)$.
{We also analyze speculative coloring from ColPack and Zoltan; we 
summarize the most competitive variants. ITR does not come with clear
bounds on depth, but its simple and parallelizable structure makes it very
fast. ITRB schemes are $>$2$\times$ slower than other baselines and
are thus excluded from run-time plots.  We also consider an additional variant
of ITR based on ASL~\mbox{\cite{gebremedhin2013colpack}}, ITR-ASL.
In several cases, it approaches the performance of ITR.}
%
%   We analyze \textbf{run-times} using all the available parallelism
%   (max.~\#cores). We show fractions due to \textbf{reordering} (preprocessing,
%   e.g., the ``ADG'' phase in JP-ADG) and the actual \textbf{coloring} (e.g.,
%   the ``JP'' phase in JP-ADG).
%   %
%   JP-SL, JP-SLL, {and JP-ASL} are the slowest as they offer least parallelism.  JP-LF,
%   JP-LLF, and JP-R are -- as expected -- very fast, as their depth depends on
%   $O(\log n)$.
%   %
%   ITR does not come with clear bounds on depth, but its simple and parallelizable
%   structure also makes it very efficient. {We also consider an additional
%   variant of ITR based on ASL~\mbox{\cite{gebremedhin2013colpack}}, ITR-ASL.
%   It is in most cases slower than ITR.}

\ifrev\marginpar{\vspace{-7em}\colorbox{orange}{\textbf{R-3}}}\fi

\ifrev\marginpar{\vspace{-1em}\colorbox{orange}{\textbf{R-3}}}\fi

%\enlargeSQ

The {coloring} run-times of JP-ADG are comparable to JP-LF,
JP-LLF, and others. This is not surprising, as this phase is dominated by the
common JP skeleton (with some minor differences from varying schedules of
vertex coloring). However, the {reordering} run-time in JP-ADG comes with certain
overheads because it depends on $\log^2 n$. This
is expected, as JP-ADG -- by its design -- performs several sequential
iterations, the count of which is determined by $\epsilon$ (i.e., how well the
degeneracy order is approximated).
Importantly, JP-ADG is consistently faster (by more than 1.5$\times$) than JP-SL
and JP-SLL that also focus on coloring quality.

DEC-ADG-ITR -- similarly to JP-ADG -- entails ordering overheads because it
precomputes the ADG low-degree decomposition. However, total run-times are
only marginally higher, and in several cases \emph{lower} than those in ITR.  This is
because the low-degree decomposition that we employ, despite enforcing some
sequential steps in preprocessing, \emph{reduces counts of coloring conflict}, translating
to performance gains.

% ITR sfers from overheads due to conflicts; its run-times are consistent with
% results in past work~\cite{whasenplaugh2014ordering}. 

\tr{\subsection{Analysis of Coloring Quality}}

\cnf{\textbf{\ul{Coloring Quality}}}
{Coloring quality} also follows the theoretical predictions: JP-SL outperforms
JP-SLL, JP-LF, and JP-LLF (by up to 15\%), as it strictly follows the
degeneracy order. Overall, all four schemes {(GBBS/Ligra, HP)} are competitive.
As expected, JP-FF and JP-R come with \emph{much} worse coloring qualities
because they do not focus on minimizing color counts.  As observed
before~\cite{whasenplaugh2014ordering}, ITR {(ColPack)} outperforms JP-FF and JP-R but
falls behind JP-LF, JP-LLF, JP-SLL, and JP-SL.
JP-ASL and ITR-ASL {(ColPack)} offer low (often the lowest) quality.
{ITRB (Zoltan) sometimes approaches the quality of JP-SL, JP-SLL,
DEC-ADG-ITR, and JP-ADG.}
%
%   {Coloring quality} also follows the theoretical
%   predictions: JP-SL outperforms JP-SLL, JP-LF, and JP-LLF (by up to 15\%), as it strictly
%   follows the degeneracy order. Overall, all four schemes are competitive.
%   As expected, JP-FF and JP-R come with \emph{much} worse coloring
%   qualities because they do not focus on minimizing color counts.
%   As observed before~\cite{whasenplaugh2014ordering}, ITR
%   is better than JP-FF and JP-R but falls behind JP-LF, JP-LLF, JP-SLL, and JP-SL. 
%   %
%   JP-ASL and ITR-ASL offer low (often the lowest) quality colorings.

\ifrev\marginpar{\vspace{-2em}\colorbox{orange}{\textbf{R-3}}}\fi

\emph{The coloring quality of our schemes outperforms others in almost all
cases}. Only JP-SL, JP-SLL {(GBBS/Ligra, HP)}, {and sometimes ITRB (Zoltan)}
are competitive, but they are always much slower. In some cases, JP-ADG (e.g.,
in s-ork) and DEC-ADG-ITR (e.g., in s-gmc) \emph{are better than JP-SL and
JP-SLL} (by 3-10\%). Hence, while the strict degeneracy order is in general
beneficial when scheduling vertex coloring, it does \emph{not} always give best
qualities.
%   
%   \emph{The coloring quality of our schemes outperforms others in almost all cases}.
%   Only JP-SL and JP-SLL are competitive, but they are always much slower.
%   In some cases, JP-ADG (e.g., in s-ork) and DEC-ADG-ITR
%   (e.g., in s-gmc) \emph{are better than JP-SL and JP-SLL} (by 3-10\%). Hence, while
%   the strict degeneracy order is in general beneficial when scheduling
%   vertex coloring, it does \emph{not} always give best qualities. 
%
% (which is
% expected, as one can compute the degeneracy order in linear time while
% deriving optimal coloring is NP-complete). 
%
JP-ADG consistently outperforms
others, reducing used color counts by even up to 23\%
compared to JP-LLF (for m-wta).
Finally,
DEC-ADG-ITR always ensures much better quality than ITR, up to 40\%
(for s-lib). Both DEC-ADG-ITR and JP-ADG offer similarly high
coloring qualities across all comparison targets.

\enlargeSQ

\tr{\subsection{Analysis of Strong Scaling}}

\cnf{\textbf{\ul{Strong Scaling}}}
We also investigate strong scaling (i.e., run-times for the increasing
thread counts). Relative performance differences between baselines do not
change significantly, except for SLL that becomes more competitive
when the thread count approaches 1, due to the tuned sequential implementation
that we used~\mbox{\cite{whasenplaugh2014ordering}}. Representative results are in
Figure~\mbox{\ref{fig:results-scaling}}; all other graphs result in analogous
performance patterns.
{Most variants from ColPack, Zoltan, GBBS/Ligra, and HP scale well (we still
exclude Zoltan due to high runtimes).}
Importantly, scaling of our baselines is also \emph{advantageous} and
comparable to others. This follows theoretical predictions, as the
\mbox{$\log^2 n$} factor in our depth bounds is alleviated by the presence of
the degeneracy \mbox{$d$} (or \mbox{$\log d$}) instead of \mbox{$\Delta$}, as
opposed to the competition; see~\mbox{\cref{sec:comp-others}} on page 7 for details.
%
%   We also investigate {strong scaling} (i.e., run-times for the increasing
%   thread counts). Relative performance differences between baselines do not
%   change significantly, with an exception of SLL that becomes more competitive
%   when the thread count approaches 1, due to the tuned sequential implementation
%   that we used~\mbox{\cite{whasenplaugh2014ordering}}. {Representative results are in
%   Figure~\mbox{\ref{fig:results-scaling}}; all other graphs result in analogous
%   performance patterns.
%   %
%   Importantly, scaling of our baselines is also \emph{advantageous} and
%   comparable to others. This follows theoretical predictions, as the
%   \mbox{$\log^2 n$} factor in our depth bounds is alleviated by the presence of
%   the degeneracy \mbox{$d$} (or \mbox{$\log d$}) instead of \mbox{$\Delta$}, as
%   opposed to the competition; see~\mbox{\cref{sec:comp-others}} on page 7 for details.}

\tr{\subsection{Analysis of Weak Scaling}}
{\cnf{\textbf{Weak Scaling}}
Weak scaling is also shown in Figure~\mbox{\ref{fig:results-scaling}}.  We
use Kronecker graphs~\mbox{\cite{leskovec2010kronecker}} of the increasing
sizes by varying the number of edges/vertex; this \emph{fixes the used graph
model}. JP-ADG scales comparably to other JP baselines; DEC-ADG-ITR
scales comparably or better than ITR or ITR-ASL.  }

\tr{\subsection{Analysis of Impact from $\epsilon$}}

{\cnf{\textbf{Impact of \mbox{$\epsilon$}}}
Representative results
of the impact of \mbox{$\epsilon$}
are in Fig.~\mbox{\ref{fig:results-eps}}.
As expected, larger \mbox{$\epsilon$} offers more parallelism and thus lower
runtimes, but coloring qualities might decrease. Importantly, the decrease is
\emph{minor}, and the qualities remain the highest or competitive across
almost the whole spectrum of~\mbox{$\epsilon$}.}

\begin{figure}[h]
\vspaceSQ{-1em}
\centering
\includegraphics[width=0.48\textwidth]{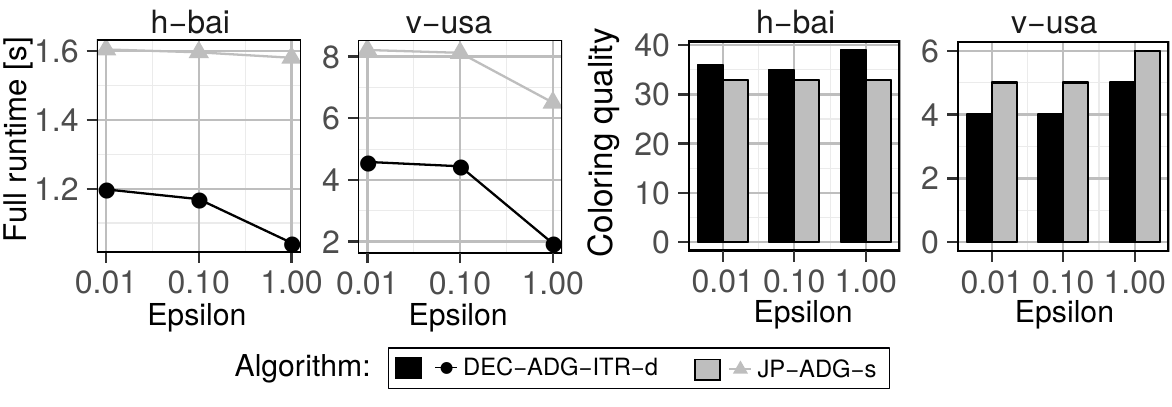}
\vspace{-0.5em}
\caption{\scriptsize{\textmd{\textbf{Impact of \mbox{$\epsilon$} on run-times and coloring quality}.} 
Parametrization: 32 cores, sorting: Radix sort, direction-optimization: push,
JP-ADG variant based on average degrees \mbox{$\widehat{\delta}$}.
DEC-ADG-ITR uses dynamic scheduling.
%
%JP-ADG uses linear time sorting of \mbox{$R$}.
%
}}
\vspaceSQ{-0.5em}
\label{fig:results-eps}
\end{figure}

\enlargeSQ

\tr{\subsection{Memory Pressure and Idle Cycles}}
{\cnf{\textbf{Memory Pressure }}
We also investigate the pressure on the memory bus, see
Figure~\mbox{\ref{fig:results-papi}}.
For this, we use PAPI~\mbox{\cite{mucci1999papi}} to gather data about idle CPU
cycles and L3 cache misses. Low ratios of L3 misses or idle cycles
indicate high locality and low pressure on the memory bus.
Overall, our routines have comparable or
best ratios of active cycles and L3 hits.}

\begin{figure}[h]
\vspaceSQ{-1em}
\centering
\includegraphics[width=0.5\textwidth]{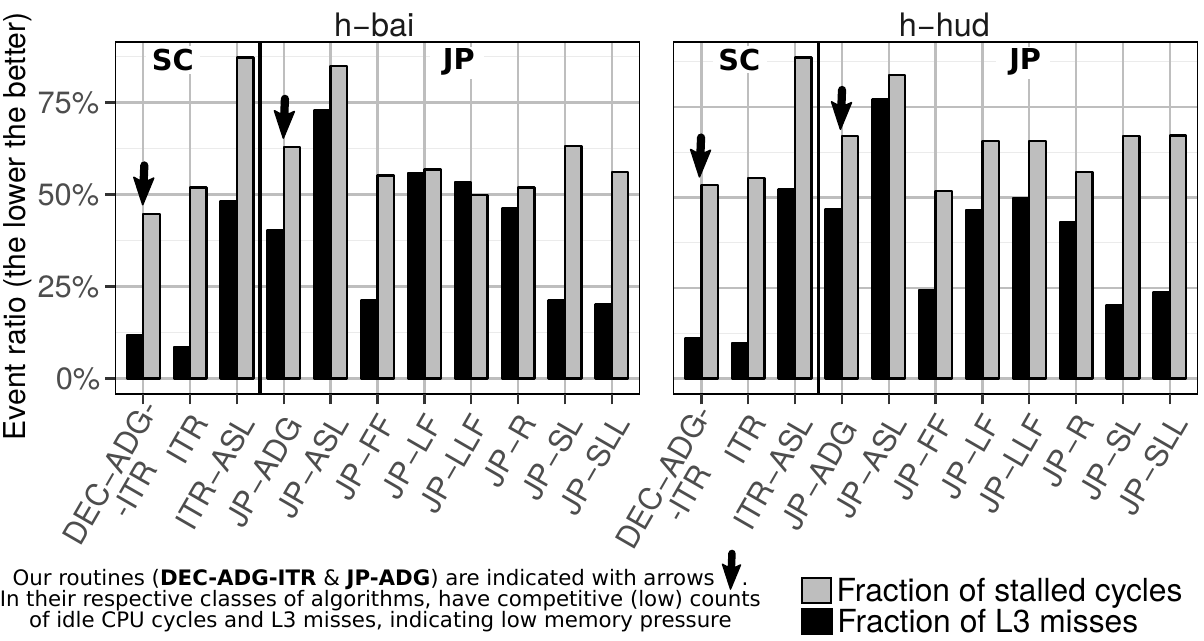}
\vspace{-1.5em}
\caption{\scriptsize{\textmd{\textbf{Fractions of L3 misses (out of all L3 accesses)
and idle (stalled) CPU cycles (out of all CPU cycles) in each algorithm execution}.} 
Parametrization: graph~h-hud, 32 cores, sorting: Radix sort, direction-optimization: push,
JP-ADG uses average degrees \mbox{$\widehat{\delta}$}.
DEC-ADG-ITR uses dynamic scheduling.
JP-ADG uses linear time sorting of \mbox{$R$}.
}}
\vspaceSQ{-0.5em}
\label{fig:results-papi}
\end{figure}

\tr{\subsection{Performance Profiles}}
{\cnf{\textbf{Performance Profiles }}
We also summarize the results from Figure~\mbox{\ref{fig:results}} using
\emph{performance profiles}~\mbox{\cite{dolan2002benchmarking}}, see
Figure~\mbox{\ref{fig:results-profile-color}} for a representative profile for
coloring quality. Intuitively, such a profile shows cumulative
distributions for a selected performance metric (e.g., a color count).
The summary in Figure~\mbox{\ref{fig:results-profile-color}} confirms the
previous insights: DEC-ADG-ITR, JP-ADG, and JP-SL offer the best colorings.}

\tr{\begin{figure*}[t]
\vspaceSQ{-1em}
\centering
\includegraphics[width=0.7\textwidth]{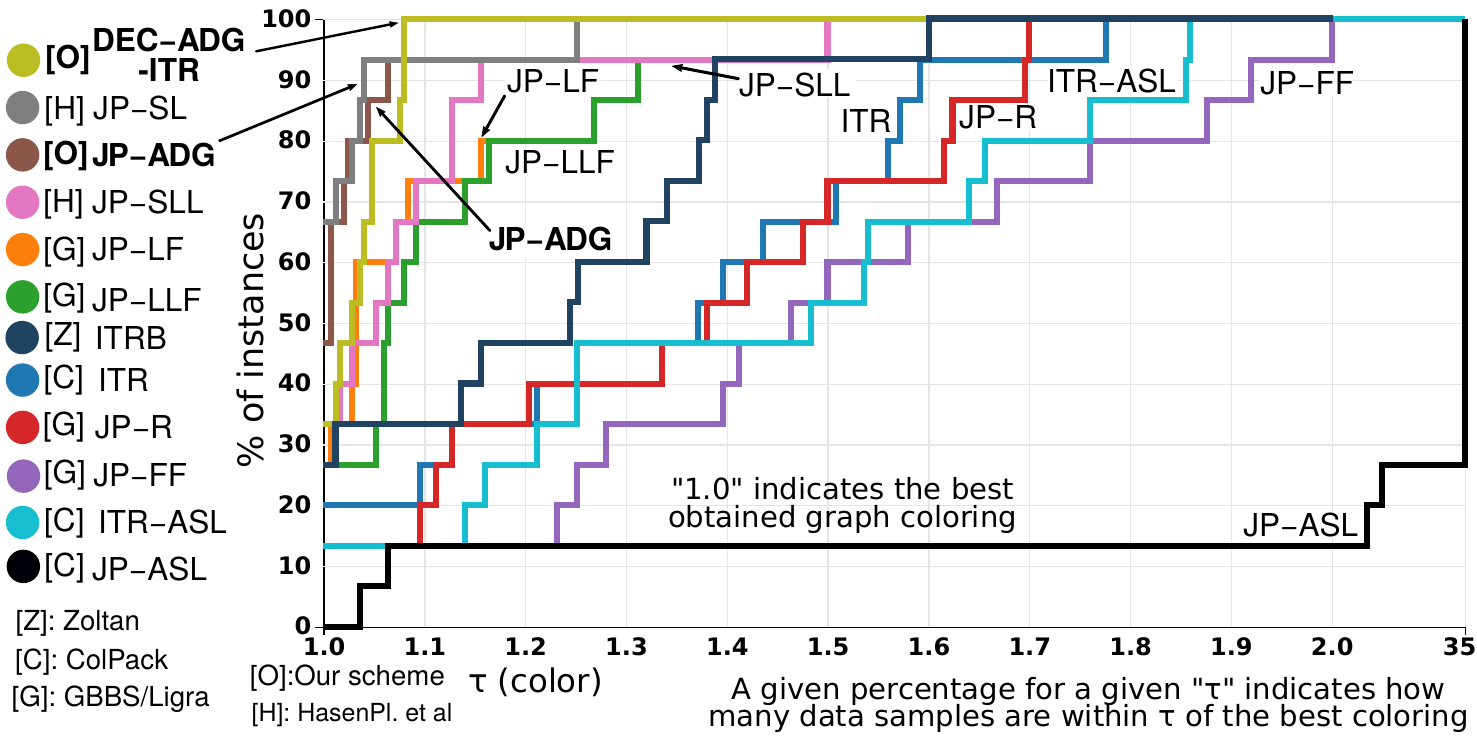}
\vspaceSQ{-1.5em}
\caption{\scriptsize{\textmd{\textbf{Color qualities from Figure~\mbox{\ref{fig:results}} summarized with a performance profile}.} 
}}
\vspaceSQ{-0.5em}
\label{fig:results-profile-color}
\end{figure*}}

\cnf{\begin{figure}[h]
\vspaceSQ{-1em}
\centering
\includegraphics[width=0.5\textwidth]{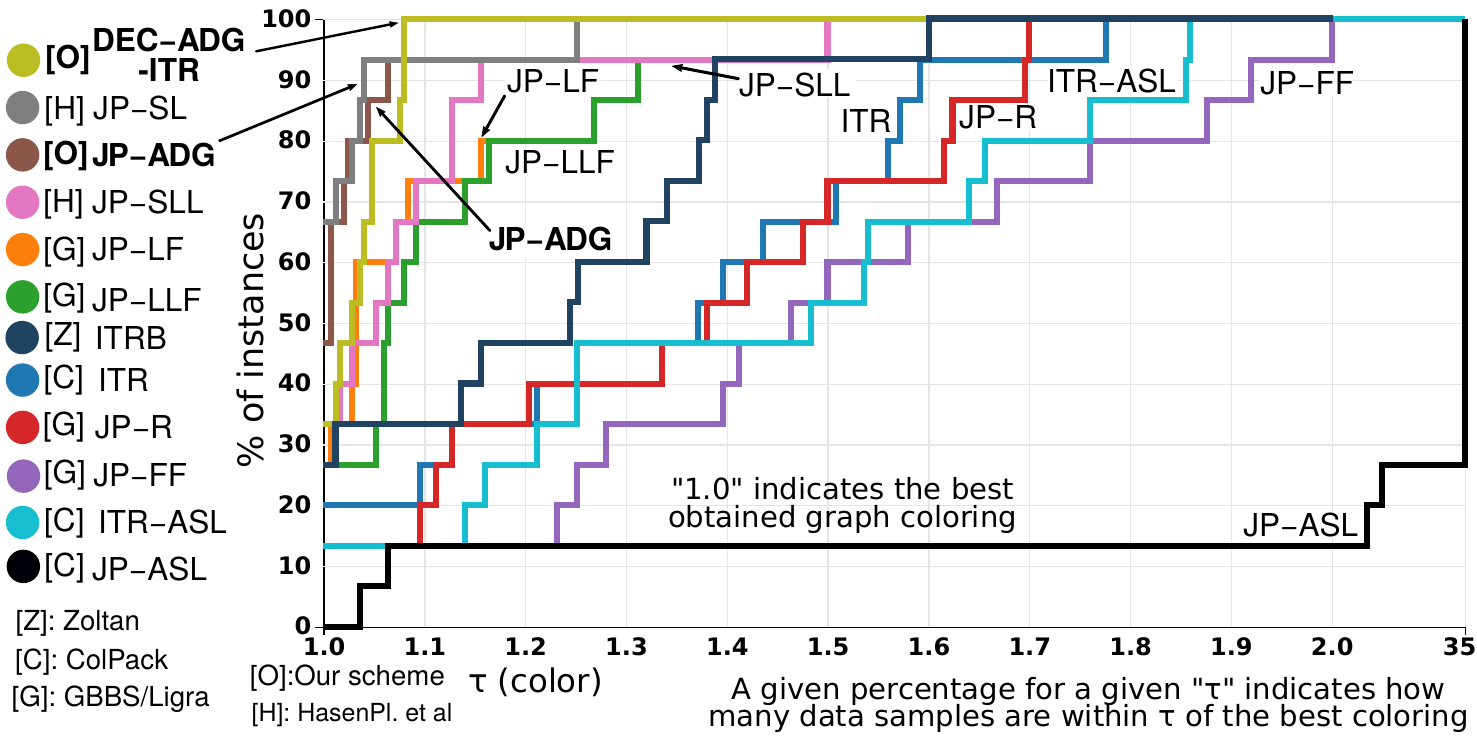}
\vspace{-1.5em}
\caption{\scriptsize{\textmd{\textbf{Color qualities from Figure~\mbox{\ref{fig:results}} summarized with a performance profile}.} 
}}
\vspaceSQ{-0.5em}
\label{fig:results-profile-color}
\end{figure}}

\tr{\subsection{Additional Analyses of Design Choices}}
\cnf{\textbf{\ul{Additional Analyses of Design Choices}}}
We also analyze variants of JP-ADG and DEC-ADG-ITR,
considering \textbf{sorting of set~$R$}, using static vs.~dynamic
\textbf{schedules} and other aspects from Section~\ref{sec:design-impl}
(push vs.~pull, median vs.~average degree, different sorting algorithms).  In
Figure~\ref{fig:results}, we use JP-ADG with counting sorting of~$R$ and DEC-ADG-ITR
with dynamic scheduling. All these design choices have a certain (usually up
to 10\% of relative difference) impact on performance, but (1) it strongly
depends on the input graph, and (2) does not change fundamental performance
patterns.

\section{Related Work}

\enlargeSQ

\iftr
One of the first methods for computing a $(\Delta + 1)$-coloring \textbf{in
parallel} for \emph{general} graphs arose from the parallel
polylogarithmic-depth \emph{maximal independent set} (MIS) algorithm by Karp
and Wigderson~\cite{karp1985afast}, further improved to $O(\log n)$ depth by a
simpler MIS algorithm in the influential paper by Luby~\cite{luby1986simple}
The key idea in this simple coloring strategy is to (1) find a MIS~$S$, (2)
apply steps of Greedy in parallel for all vertices in $S$ (which is possible
as, by definition, no two vertices in a MIS are adjacent), coloring all
vertices in~$S$ with a new color, (3) remove the colored vertices from~$G$, and
(4) repeat the above steps until all vertices are colored.

Karp and Wigderson's, and Luby's algorithms started a large body of
\textbf{parallel graph coloring heuristics}~\cite{nalon1986afast,
luby1986simple, goldberg1989anew, luby1993removing, goldberg1987paralleldelta,
goldberg1987parallel, gebremedhin2000scalable, welsh1967an, brelaz1979new,
whasenplaugh2014ordering, jones1993parallel, jrallwright1995,
coleman1983estimation, jones1994scalable, gebremedhin2005what,
arkin1987scheduling, marx2004graph, ramaswami1989distributed,
matula1983smallest, karp1985afast, matula1972graph}.  We already exhaustively
analyzed them in Section~\ref{sec:intro} and in Table~\ref{tab:gc-all}. 
\else
\ifconf
We already exhaustively analyzed a large body of \textbf{sequential \& parallel
graph coloring heuristics}~\cite{nalon1986afast, luby1986simple,
goldberg1989anew, luby1993removing, goldberg1987paralleldelta,
goldberg1987parallel, gebremedhin2000scalable, welsh1967an, brelaz1979new,
whasenplaugh2014ordering, jones1993parallel, jrallwright1995,
coleman1983estimation, jones1994scalable, gebremedhin2005what,
arkin1987scheduling, marx2004graph, ramaswami1989distributed,
matula1983smallest, karp1985afast, matula1972graph} in Section~\ref{sec:intro}
and in Table~\ref{tab:gc-all}. 
\fi
\fi
Almost all of them have theoretical guarantees based on the
work-depth~\cite{cormen2009introduction} or the PRAM
model~\cite{blelloch1996programming}. \emph{We build on and improve on these
works in several dimensions, as explained in detail in
Section~\ref{sec:intro}.}

Many works exist in the \textbf{theory of distributed graph coloring} based on
models such as LOCAL or CONGESTED-CLIQUE~\cite{linial1987distributive,
linial1992locality, barenboim2013distributed, barenboim2016thelocality,
barenboim2011deterministic, schneider2010anew, barenboim2009distributed,
harris2016distributed, chang2018anoptimal, panconesi1992improved,
barenboim2018locally, chang2019anexponential, johansson1999simple,
grable2000fast, kuhn2006complexity}.
These algorithms are highly theoretical and do not come with any
implementations. Moreover, they come with assumptions that are unrealistic in
HPC settings, for example distributed LOCAL and CONGEST algorithms do not
initially know the interconnection graph, or the message size in LOCAL
algorithms can be unbounded. Finally, they cannot be directly compared to
Work-Depth or PRAM algorithms.
Thus, they are \emph{of little relevance to our work}.

\ifbsp
In recent years, distributed vertex-coloring algorithms, together with other
symmetry-breaking problems~\cite{barenboim2016thelocality}, have seen more and
more attention in research. These algorithms are evaluated in a
\emph{message-passing model}, where most of the time \emph{the cost-model is
based on communication rounds and message-size instead of local computation}.
One of the most well known models is LOCAL which was first formalized by Linial
\cite{linial1987distributive, linial1992locality}. This model works on an
arbitrary Graph with $n$ vertices, where each vertex can send messages along
its incident edges. At the beginning, the vertices are only aware of their
local data and $n$. In each round, a vertex can perform some arbitrary
computation, send a message to it's neighbors and receive from them. The cost
of an algorithm is then expressed in terms of the amount of communication
rounds.  Since this model is pretty much orthogonal to PRAM, where the cost
model is based on local computation, we didn't consider these algorithms for a
comparison to ours.

One early influential contribution in the field of distributed vertex coloring
was a deterministic $O(\Delta^2)$ coloring method, which runs in $O(\log^*{n})$
communication rounds by Linial \cite{linial1992locality}. Since then a lot of
different algorithms have been presented for different coloring bounds like
$O(a)$ coloring \cite{barenboim2011deterministic}, where $a$ is the arboricity
of $G$,  $O(\Delta + \log n)$ coloring \cite{schneider2010anew} or $O(\Delta)$
coloring \cite{barenboim2011deterministic}. Unlike in the shared memory
setting, where research on non-heuristic $(\Delta + 1)$ coloring algorithms has
essentially stopped around the time of Luby's Algorithm
\cite{barenboim2016thelocality, }, research for $(\Delta + 1)$ coloring
algorithms in settings like LOCAL has been flourishing
\cite{barenboim2016thelocality, }.

For randomized $(\Delta + 1)$ coloring some important improvements over they
years have been a $O(\Delta + \log^*{n})$ round algorithm
\cite{barenboim2009distributed}, a $O(\log{\Delta} + \sqrt{\log n})$ round
algorithm \cite{schneider2010anew} and a $O(\log{\Delta}) + 2^{O(\sqrt{\log
\log n})}$ round algorithm \cite{barenboim2016thelocality}. One method that
runs in sublogarithmic $O(\sqrt{\log \Delta}) + 2^{O(\sqrt{\log \log n})}$
rounds \cite{harris2016distributed}, and one that runs in $2^{O(\sqrt{\log \log
n})}$ \cite{chang2018anoptimal} were also presented quite recently.  For
deterministc $(\Delta + 1)$ coloring, the best algorithms are a fairly old
$2^{O(\sqrt{\log n})}$ round method \cite{panconesi1992improved} and a recent
$O(\Delta + \log^*{n})$ round algorithm \cite{barenboim2018locally}.  An
interesting observation to make recently for symmetry-breaking problems in
LOCAL is, that the best randomized algorithm often runs exponentially faster
than the best deterministic one \cite{chang2019anexponential,
barenboim2016thelocality}. And indeed as shown in a paper of Chang et al.,
there exists an inherit exponential gap between the run-time of deterministic
and randomized algorithms for some symmetry-breaking problems
\cite{chang2019anexponential}.

Most of these distributed algorithms however, are highly theoretic and more
important in settings, where communication between distributed machines is the
bottleneck and where only local data is available. For an extensive overview of
algorithms and techniques for distributed coloring we refer to the monograph of
Barenboim and Elkin \cite{barenboim2013distributed}. 
\fi

Many \textbf{practical parallel and distributed approaches}
were proposed. They often use different \emph{speculative
schemes}~\cite{ccatalyurek2011distributed, bozdaug2008framework, besta2017push,
gebremedhin2000scalable, boman2005scalable, gebremedhin2000graph,
ccatalyurek2012graph, saule2012early, sariyuce2012scalable, rokos2015fast,
grosset2011evaluating, deveci2016parallel, finocchi2005experimental}, where
vertices are colored speculatively and potential conflicts are resolved in a
second pass. 
Some of these schemes were implemented within frameworks or
libraries~\cite{gebremedhin2013colpack, gregor2005what, bdoruk2008aframework,
naumov2015amgx}.
%
% Some of these schemes focus on incorporating
% vectorization~\cite{grosset2011evaluating}.
%
Another line of schemes incorporates \emph{GPUs and
vectorization}~\cite{naumov2010cusparse, saule2012early, deveci2016parallel,
cohen2012efficient, chen2017efficient, che2015graph, grosset2011evaluating,
naumov2015parallel}.
Other schemes use \emph{recoloring}~\cite{culberson2011iterated,
asariyuce2011improving} in which one improves an already existing coloring. 
Patidar and Chakrabarti use \emph{Hadoop} to implement graph
coloring~\cite{gandhi2015performance}. Alabandi et
al.~\cite{alabandi2020increasing} illustrate how to \emph{increase parallelism}
of coloring heuristics.
\emph{These works are orthogonal to this paper}: they do not provide
theoretical analyses, but they usually offer numerous architectural and design
optimizations that can be combined with our algorithms for further performance
benefits. As we focused on theoretical guarantees and its impact on
performance, and \emph{not} on architecture-related optimizations, \emph{we leave
integrating these optimizations with our algorithms as future work}.

There are works on \textbf{coloring specific graph classes}, such as
planar graphs~\cite{diks1986fast, matula1980two, boyar1987coloring,
hagerup1989optimal}. Some works impose \textbf{additional restrictions}, 
for example coloring \emph{balance}, which limits differences
between numbers of vertices with different colors~\cite{gjertsen1996parallel,
lu2017balanced, agebremedhin2002parallel}. Other lines of related work also
exist, for example on \textbf{edge coloring}~\cite{holyer1981np},
\textbf{dynamic} or \textbf{streaming coloring}~\cite{sallinen2016graph,
yuan2017effective, bossek2019runtime, bhattacharya2018dynamic,
solomon2019improved, bera2018coloring, besta2019substream}, $k$-distance-coloring and
other generalizations~\cite{hlu2017algorithms,
bozdaug2010distributed, bozdaug2005parallel}, and \textbf{sequential exact
coloring}~\cite{lin2017reduction, verma2015solving, hebrard2019hybrid}.  There
are even works on solving graph coloring with evolutionary and genetic
algorithms~\cite{abbasian2011efficient, eiben1998graph, fleurent1996genetic} and
with machine learning methods~\cite{zhou2018improving, lemos2019graph,
zhou2016reinforcement, musliu2013algorithm, ben2019modular, huang2019coloring}.
All these works are unrelated as \emph{we focus on unrestricted, parallel, and
1-distance vertex coloring with provable guarantees on performance and quality},
targeting \emph{general, static, and simple graphs}.

The \textbf{general structure of our ADG algorithm}, based on iteratively
removing vertices with degrees in certain ranges defined by the approximation
parameter~$\epsilon$, was also used to solve other problems, for example the
$(2+\epsilon)$-approximate maximal densest subgraph algorithms by Dhulipala et
al.~\cite{dhulipala2018theoretically}. Finding more applications of ADG is left
for future work.

We note that, while our ADG scheme is \emph{the first parallel algorithm for
deriving approximate degeneracy ordering {with a provable approximation
factor}}, two algorithms in the \emph{streaming setting}
exist~\cite{farach2016tight, farach2014computing}.

\ifall

\hl{Besides SLL, another parallel algorithm that \emph{heuristically
approximates} SL can be found in the literature
(ASL)~\mbox{\cite{patwary2011new}}. ASL is similar in structure to a parallel
SL implementation using ``degree bins''. The approximation is achieved by
limiting degree updates to a vertex's local thread.}

% (?)The same algorithm could
% be used to approximate SLL in the same way by using logarithmic degrees.

\fi

\iftr
Graph coloring has been targeted in several recent works related to \textbf{broad graph
processing paradigms, abstractions, and
frameworks}~\cite{gonzalez2012powergraph, besta2017push, besta2015accelerating,
besta2018log, gianinazzi2018communication, besta2019practice,
besta2019demystifying, besta2019graph, besta2020substream}.
Several HPC works~\cite{thebault2016scalable, firoz2018runtime,
gregor2005parallel} consider distributed graph coloring in the context of
high-performance RDMA networks and RMA programming~\cite{besta2015active,
besta2014fault, fompi-paper, gerstenberger2018enabling, schweizer2015evaluating, schmid2016high}.
Different coloring properties of graphs were also analyzed in the context of
\textbf{graph compression} and \textbf{summarization}~\cite{besta2019slim, besta2018survey}.
\fi

\section{Conclusion}
\vspaceSQ{-0.25em}

\enlargeSQ

We develop graph coloring algorithms with strong theoretical guarantees on all
three key aspects of parallel graph coloring: work, depth, and coloring
quality. No other existing algorithm provides such guarantees.

One algorithm, JP-ADG, is often superior in coloring quality to all other
baselines, even including the tuned SL and SLL algorithms specifically designed
to reduce counts of used colors~\cite{whasenplaugh2014ordering}.  It also
offers low run-times for different input graphs.  As we focus on algorithm
design and analysis, one could combine JP-ADG with many orthogonal
optimizations, for example in the GPU landscape, to achieve more performance
without sacrificing quality.
Another algorithm, DEC-ADG, is of theoretical interest as it is the first
routine -- in a line of works based on speculative coloring -- with strong
theoretical bounds. While being less advantageous in practice, we use its
underlying design to enhance a recent coloring
heuristic~\cite{ccatalyurek2012graph} obtaining DEC-ADG-ITR, an algorithm with
(1) strong quality bounds and (2) competitive performance, for example up to
40\% fewer colors used then compared to the base
design~\cite{ccatalyurek2012graph}.

Our algorithms use a very simple (but rich in outcome) idea of provably
relaxing the strict vertex degeneracy order, to maximize parallelism when
deriving this order. This idea, and our corresponding parallel ADG algorithm,
are of separate interest, and could enhance other algorithms that rely on
vertex ordering, for example in mining maximal
cliques~\cite{DBLP:conf/isaac/EppsteinLS10, cazals2008note}.

We provide the most extensive theoretical study of parallel graph coloring
algorithms.  This analysis can be used by other researchers as help in
identifying future work.

\vspace{1em}

\scriptsize 

%\subsubsection*{Acknowledgments}
\macb{Acknowledgements: }
We thank Hussein Harake, Colin McMurtrie, Mark Klein, Angelo Mangili, and the
whole CSCS team granting access to the Ault and Daint machines, and for their
excellent technical support. 
We thank Timo Schneider
for his immense help with computing infrastructure at SPCL.

\normalsize

\bibliographystyle{IEEEtran}
%\balance
\bibliography{references}

\maciej{! check Lukas' example}

\iftr
\maciej{! check: Fast distributed algorithms for Brooks–Vizing colorings}
\fi

%\appendix
%\input{appendix.tex}

\end{document}